\documentclass[12pt]{article}
\usepackage{amsfonts}
\usepackage{amssymb}
\usepackage{amsmath}
\usepackage{amsthm}
\usepackage{bbm}
\usepackage{xcolor}
\usepackage{natbib}
\usepackage{rotating}
\usepackage{booktabs}
\usepackage{rotating}
\usepackage{mathtools}
\usepackage[colorlinks=true, linkcolor=blue, citecolor=blue, urlcolor=blue]{hyperref}
\usepackage{graphicx} 
\usepackage{array}    
\usepackage{adjustbox}

\renewenvironment{proof}{\par\noindent}{\hfill$\square$}

\renewcommand{\baselinestretch}{1.3}
\newcommand{\bc}{\begin{center}}
\newcommand{\ec}{\end{center}}
\newcommand{\be}{\begin{equation}}
\newcommand{\ee}{\end{equation}}
\newcommand{\bea}{\begin{eqnarray}}
\newcommand{\eea}{\end{eqnarray}}
\newcommand{\bean}{\begin{eqnarray*}}
\newcommand{\eean}{\end{eqnarray*}}

\newtheorem{theorem}{Theorem}

\newtheorem{assumption}{Assumption}

\newtheorem{corollary}{Corollary}

\newtheorem{definition}{Definition}

\newtheorem{lemma}{Lemma}

\newtheorem{proposition}{Proposition}

\theoremstyle{definition} 
\newtheorem{remark}{Remark}

\newcounter{pkt}
\newenvironment{tlist}{%
  \begin{list}{(\roman{pkt})}{%
    \usecounter{pkt}%
    \setlength{\leftmargin}{0.5em}%
    \setlength{\itemindent}{0pt}%
    \setlength{\labelwidth}{0pt}%
    \setlength{\labelsep}{0.5em}
    \parskip0ex%
    \parsep0ex%
    \itemsep0ex%
    \topsep0ex%
  }%
}{%
  \end{list}%
}

\def\EE{\mathbb{E}}

\def\RR{\mathbb{R}}
\def\Proj{\operatorname{Proj}}
\newcommand{\argmin}{\operatorname*{argmin}}
\newcommand{\vech}{\operatorname*{vech}}
\newcommand{\diag}{\operatorname*{diag}}
\newcommand{\Cov}{\operatorname*{Cov}}
\newcommand{\Var}{\operatorname*{Var}}

\newcounter{saveeqn}

\begin{document}

\title{\vspace*{-0.5 in}  \textbf{How Well Are State-Dependent Local Projections Capturing Nonlinearities?}}

\author{Zhiheng You\thanks{
		\setlength{\baselineskip}{4mm}  Correspondence: Z. You: Department of Economics, University of Pennsylvania, Philadelphia, PA 19104-6297. Email: zhyou@sas.upenn.edu. I am grateful to Frank Schorfheide, Wayne Gao, and Xu Cheng for their invaluable advice. I also thank seminar participants at the University of Pennsylvania,  the Midwest Econometrics Group Conference 2025, the 40th Annual Canadian Econometrics Study Group Meeting, and the Fall 2025 Midwest Macroeconomics Meeting for helpful discussions.} \\
	{\em \large University of Pennsylvania}}

\date{\large \today}
\maketitle

\begin{abstract} 
We use quadratic vector autoregressions, motivated by pruned second-order perturbation solutions to DSGE models, as a laboratory to evaluate how well popular local projection (LP) specifications recover true impulse responses in nonlinear environments. We derive closed-form population impulse responses under each specification and compare them to truth.
Linear LP fails to capture nonlinearities when the shock is symmetrically distributed. 
State-dependent LP specifications capture distinct aspects of nonlinearity: interacting the shock with its sign captures asymmetric effects, while interacting the shock with observable state proxies captures state dependence. However, their gains over linear LP are concentrated in tail shocks or states, and for the latter depend on proxy quality.
Our proposed specification---augmenting linear LP with a squared shock term and shock-state proxy interactions---best approximates true responses.
We also establish valid estimation and inference procedures for this specification. In a monetary policy application, we find state dependence, while higher-order effects differ across outcomes.
(JEL C22, C32, E52)
\end{abstract}


\thispagestyle{empty}

\clearpage
\setcounter{page}{1}

\section{Introduction}
\label{sec:intro}

Macroeconomic shocks often propagate through the economy in fundamentally nonlinear ways ---e.g., fiscal multipliers vary with the state of the business cycle \citep{AuerbachGorodnichenko2012}, monetary policy transmission is nonlinear at the zero lower bound \citep{EggertssonWoodford2003}, and financial frictions generate endogenous amplification and crisis dynamics \citep{BrunnermeierSannikov2014}. 
To study these dynamics empirically, local projections (LP; \citealp{Jorda2005}) have become one of the most popular methods for estimating impulse response functions. While LPs offer flexibility and simplicity in implementation, capturing general nonlinearities through them remains challenging.  Although fully nonparametric approaches exist (see \citealp{GoncalvesEtAl2024}), parametric specifications remain dominant in applied work due to their ease of implementation and the fact that the relatively small sample sizes typical in macroeconomic data limit the precision of nonparametric estimates. Yet capturing general nonlinearities through parametric specifications remains challenging, as researchers must take explicit stands on the specific form of nonlinearity---choices that can substantially affect conclusions about economic dynamics.

The empirical literature has developed various LP specifications to navigate these challenges. The simplest and most widely used is linear LPs, which directly regresses future outcomes on current shocks. Recent work by \cite{KolesarPlagborg-Moller2025} provides an important theoretical foundation for this approach, showing that even when the true data-generating process is nonlinear, a linear LP on observed shocks recovers a weighted average of marginal effects with respect to the shock size. This result suggests that linear LPs can provide a meaningful causal summary of nonlinear dynamics, making them a natural starting point for empirical analysis---what they aptly describe as a good candidate for ``the first column" of any careful empirical study. While linear LPs offer this valuable baseline, empirical researchers have increasingly turned to state-dependent LP specifications to explore nonlinearities more explicitly. These specifications can be broadly classified into two categories. The first category includes specifications that interact the shock with its own sign, as in \cite{BenZeevEtAl2023}, \cite{FurceriEtAl2018}, and \cite{AlbrizioEtAl2020}. The second category interacts the shock with lagged variables or transformations thereof. For example, \cite{AuerbachGorodnichenko2013} use a smooth transition probability between ``recession" and ``expansion" states constructed from lagged observables, while \cite{AuerEtAl2021} use the lagged leverage ratio as the state variable.

These different specifications reflect researchers' focus on different aspects of nonlinearity that may be most relevant for their specific research questions. Despite the prevalence of these approaches in applied research, fundamental questions remain: What specific aspects of nonlinearity does each specification capture, and how well does it do so? Can we propose a specification that better approximates the true impulse responses? Answering these questions requires a framework that is both rich enough to generate economically meaningful nonlinearities and tractable enough to derive analytical results. Therefore, we focus on a tractable yet rich class of nonlinear models as our data-generating process: quadratic vector autoregression (QVAR) models. QVAR generalizes the QAR model introduced by \cite{AruobaEtAl2017}, which is the univariate case of QVAR, and is motivated by pruned perturbation solutions of DSGE models (\citealp{KimEtAl2008}). Higher-order perturbation methods are widely used in quantitative macroeconomics to capture key nonlinearities inherent in models and data; see \cite{FernandezEtAl2015} for stochastic volatility, \cite{Andreasen2012} for time-varying risk premia, and \cite{KimRuge-Murcia2009} for asymmetric price adjustment.
This makes the QVAR framework theoretically grounded in the DSGE literature and empirically tractable. Under this class of data-generating processes, the true impulse response features a state-dependent first-order component as well as a second-order component, providing a natural benchmark for evaluating different LP specifications.

Our first contribution is to demonstrate the limitations of linear LPs in capturing nonlinear dynamics of the QVAR class. We show that the population impulse response recovered from a linear LP under QVAR coincides with that recovered under a simple linear VAR---thus linear LP fails to capture any nonlinearities.
Linear LP averages over state variation and thus misses state dependence; under symmetric shocks, the remaining second-order effects also integrate to zero under the relevant weight function. This latter result depends critically on the symmetry of the weight function, which follows from our assumption that shocks are symmetrically distributed. Importantly, many widely used identified shocks in empirical macroeconomic studies also exhibit approximately symmetric distributions, suggesting that researchers should be cautious when interpreting results from purely linear LPs.

Our second contribution is to systematically evaluate widely-used state-dependent LP specifications. We focus on two prominent specifications: a shock-based specification that interacts the shock with its sign, and a lag-based specification that interacts the shock with a lagged observable proxy for the true latent state.
Under the univariate QVAR data-generating process, we derive the exact population impulse responses implied by each specification and show that both are misspecified. To evaluate their performance, we introduce a novel distance measure for impulse responses that enables meaningful comparison across specifications---even when they condition on different observables and when the conditioning variable does not coincide with the latent state.
Using a large-sample numerical illustration of the population distance measure, we map each specification's performance across the joint space of the lagged latent state and the shock.
We complement these results with a closed-form analytical comparison of the conditional mean-squared approximation error across specifications.

Our findings reveal that each specification captures a distinct aspect of the underlying nonlinearity: the shock-based specification captures second-order effects through asymmetric responses to positive and negative shocks, while the lag-based specification captures state dependence by allowing responses to vary with initial conditions. This validates the use of these specifications when researchers are primarily interested in the specific type of nonlinearity that each captures. However, we also uncover an important caveat: the shock-based specification improves upon linear LP only when the shock exceeds a modest threshold, with gains increasing in shock size; the lag-based specification reduces approximation error only to the extent that the chosen observable proxy is informative about the latent state, so gains are typically larger in tail states.
In the middle regions of these distributions where most observations typically lie, both specifications perform similarly to linear LP. 

Our third contribution is to propose an improved specification that more closely approximates the true impulse responses when they exhibit both higher-order effects and state-dependent nonlinear features, and to provide the associated estimation and inference theory. 
We first propose an infeasible specification that recovers the exact responses by augmenting the linear specification with a squared shock term and an interaction between the shock and the lagged latent state. We then propose a feasible analogue that replaces the latent state with its observable proxy. Among implementable LPs, this feasible analogue achieves the smallest distance to the truth. Unlike existing state-dependent specifications which improve upon linear LPs mainly at the tails of their respective distributions, our proposed specification makes substantial improvements across the entire joint distribution of shocks and latent states. We recommend it when both higher-order effects and state dependence are plausible sources of nonlinearity in the unknown data-generating process.
Compared to the fully nonparametric approach in \cite{GoncalvesEtAl2024}, \textit{Feas} is more parsimonious and easier to interpret, and it admits formal inference for the implied impulse responses. We therefore view \textit{Feas} as a practical baseline specification, with \textit{NPLP} serving as a complementary benchmark for assessing the role of functional-form restrictions.

On the inference side, we show that standard HAC/HAR confidence intervals are asymptotically valid for impulse responses estimated with our feasible specification when the proxy and control variables are constructed from finitely many lags. A common empirical shortcut is to augment the regression with lagged controls and then report Eicker--Huber--White (EHW) standard errors. While this is valid for correctly specified linear LPs \citep{MontielOleaEtAl2021}, it is not generally valid in our setting. In multivariate QVAR systems, omitted nonlinear cross-shock terms in the residual can leave the regression score serially correlated at horizons $h\ge1$ even when the latent state is observed without error. With proxy error, the score can also acquire a predictable component on impact, so EHW can fail already at $h=0$. 
For this reason, we recommend HAC/HAR as the default inference method for \textit{Feas}.

As an empirical application, we revisit the effects of monetary policy shocks following \cite{Ramey2016}. Using the Romer and Romer (2004) monetary policy shocks, we compare three LP specifications: a purely linear specification, our proposed feasible specification, and nonparametric LP. The results reveal economically meaningful state dependence: contractionary monetary shocks have substantially larger effects on real activity during economic troughs than during peaks, with the differences most pronounced at medium horizons. Higher-order effects are most important for unemployment at medium horizons and for the federal funds rate at short-to-medium horizons, while they are comparatively modest for industrial production and CPI. Overall, the application provides practical guidance for applied researchers on implementing state-dependent LPs.

Our work is most closely related to \cite{KolesarPlagborg-Moller2025}, who show that the estimand of linear LP onto observed shocks or proxies can be interpreted as an average marginal effect under a general nonlinear data-generating processes, while identification approaches via heteroskedasticity or non-Gaussianity of shocks are highly sensitive to departures from linearity. We focus on the case where applied researchers have direct measures of shocks, and show that linear LP fails to capture the nonlinearities in the QVAR class with symmetric shocks. In their discussion of \cite{KolesarPlagborg-Moller2025}, \citet{HerbstJohannsen2025} show that under a univariate QVAR with a symmetric shock, the population impulse response implied by linear LP coincides with the AR response. Relatedly, the discussion by \citet{GoncalvesEtAl2025} show more generally that linear LPs can average out nonlinearities under an even average structural function and symmetric shock distribution. We replicate and extend these insights to multivariate QVARs, and characterize how misspecified state-dependent LPs recover the ``missing'' nonlinear component---pinpointing where the gains over the linear specification are largest across shocks and states.

The univariate QVAR model, referred to QAR, was introduced by \cite{AruobaEtAl2017}. We extend it to the multivariate case, which is more relevant for empirical analysis. 
This class of models connects to the literature on pruned perturbation solutions to DSGE models. Higher-order perturbation is a workhorse solution method in quantitative macroeconomics; see \citet{FernandezEtAl2016} for a survey.
However, higher-order approximations to DSGE models often generate explosive sample paths even when the corresponding linearized solution is stable, as the higher-order terms create unstable steady states \citep{KimEtAl2008}. 
To eliminate these explosive dynamics, \cite{KimEtAl2008} propose pruning for second-order approximations---removing terms of higher-order effects than the considered approximation order when iterating the system forward. This pruning approach, extended to any approximation order by \cite{AndreasenEtAl2018}, ensures stationarity and ergodicity while maintaining the model's ability to capture nonlinearities, and motivates the QVAR model used in our analysis.

Empirically, state-dependent LPs have become widespread in empirical macroeconomics. Applications include \cite{AuerbachGorodnichenko2013} on fiscal policy over the business cycle, \cite{JordaEtAl2013} on financial crises, \cite{RameyZubairy2018} on government spending multipliers, and \cite{JordaEtAl2020} on monetary interventions, among others. On the theoretical side, \cite{GonccalvesEtAl2024} show that state-dependent LPs can asymptotically recover population responses to shocks regardless of shock size when the state is exogenous, but only recover conditional responses to infinitesimal shocks when the state depends on macroeconomic shocks. However, they focus on the case where the state-dependent model implied by their LP regression is correctly specified. In contrast, we focus on the case where all state-dependent LP specifications are misspecified and evaluate what they recover---following the principle that while all models are wrong, some are useful.

The remainder of the paper is organized as follows.
Section~\ref{sec:lp.failure} introduces the QAR model and shows the failure of linear LP.
Section~\ref{sec:state.lp.how.well} evaluates the performance of existing and proposed specifications under QAR. Section~\ref{sec:multi.extension} extends the analysis to multivariate QVAR.
Section~\ref{sec:feas.inference} discusses estimation and inference for the proposed specification.
Section~\ref{sec:empiric.exp} provides an empirical application to monetary policy shocks, and Section~\ref{sec:conclusion} concludes. Additional derivations, proofs, and results are provided in the Online Appendix.

\section{Failure of Linear Local Projections in QAR Framework}
\label{sec:lp.failure}

In this section, we demonstrate the fundamental limitations of linear LPs in capturing nonlinear dynamics: linear LPs fail to detect any nonlinearities in the class of models we study. For expositional purpose, we use the QAR model of \citet[ABS]{AruobaEtAl2017} as the data-generating process (DGP). In Section~\ref{sec:multi.extension}, we extend the analysis to its multivariate counterpart and show that the same conclusion continues to hold.

\subsection{Setup: QAR Model}
\label{subsec:qar.model}

\noindent {\bf Model.\ } The QAR(1,1) model is derived from a second-order perturbation approximation to the solution of a nonlinear difference equation
\be
y_t=f\left(y_{t-1}, \omega u_t\right), \quad u_t \stackrel{i i d}{\sim} N(0,1) .
\ee
Specifically, the QAR(1,1) model can be written as
\be
\begin{aligned}
& y_t=\phi_1 y_{t-1}+\phi_2 s_{t-1}^2+ (1+\gamma s_{t-1})\sigma u_t, \quad u_t \overset{\text{i.i.d.}}{\sim} \mathcal{N}(0,1),\\
& s_t=\phi_1 s_{t-1}+\sigma u_t, \quad\left|\phi_1\right|<1, 
\end{aligned}
\label{QAR.hetero}
\ee
where $y_t$ is the outcome variable, $s_t$ is a latent state variable that evolves as an AR(1) process, and $u_t$ is the exogenous structural shock driving both the state and outcome equations.
 
There are several distinct features of this model. First, as pointed out by ABS, unlike the alternative specification 
$$y_t=\phi_1 y_{t-1}+\phi_2 y_{t-1}^2+u_t, \quad 0<\phi_1<1,\ \phi_2>0,$$
the QAR model generates a recursively linear structure with a unique steady state and nonexplosive dynamics for suitably restricted values of $\phi_1$, which guarantees the stationarity of the process. Second, the model generates conditional heteroskedasticity: the conditional variance of $y_t$ is given by
$
\Var_{t-1}\left[y_t\right]=\left(1+\gamma s_{t-1}\right)^2 \sigma^2,
$
which depends on the state variable $s_{t-1}$.

Inspired by the causal framework in \cite{KolesarPlagborg-Moller2025}, we can represent $y_{t+h}$ as a \textit{structural function} $\psi_h$ of shock $u_t$ and predetermined variables and future shocks $U_{h,t+h}=(y_{t-1}, s_{t-1}, u_{t+1},\ldots, u_{t+h})$ that are independent of $u_t$, i.e.,  
\be
    y_{t+h}=\psi_h\left(u_t, U_{h,t+h}\right), \quad u_t \perp U_{h,t+h}.
\ee
The formula of $\psi_h$ is provided in the Online Appendix. We also define the \textit{average structural function} $\Psi_h(u) \equiv \mathbb{E}\left[\psi_h\left(u, U_{h,t+h}\right)\right],$ where the expectation is taken over the marginal distribution of $U_{h,t+h}$.

\noindent {\bf True Impulse Responses.\ } Our definition of the impulse response function (IRF) is based on the conditional average response (CAR) proposed in \cite{GonccalvesEtAl2024}, which compares the baseline outcome $y_{t+h}$ with the counterfactual outcome that would have been observed if there was a one-time shock to $u_{t}$, conditional on the information set at $t-1$.
Specifically, the CAR to a one-time shock of fixed size $\delta$ in $u_{t}$ is
$$
\begin{aligned}
\operatorname{CAR}_h(\mathcal{F}, \delta)
& = \mathbb{E}\left[\psi_h(u_{t}+\delta, U_{h,t+h})- \psi_h(u_{t}, U_{h,t+h}) \mid \mathcal{F}_{t-1} = \mathcal{F} \right],    
\end{aligned}
$$
where $\mathcal{F}_{t-1}$ denotes the entire past information set.

Given the expression of the structural function $\psi_h$ for the QAR(1,1) model, we can derive the CAR. The following proposition states the results. 
\begin{proposition} \label{prop:true.car}
    The CAR for the QAR(1,1) model is 
    \be
    \operatorname{CAR}_h(s, \delta) =
    \underbrace{\sigma \phi_1^{h}\,\delta}_{\text{baseline effect}}
    +
    \underbrace{a_h\, s\,\delta}_{\text{state-dependent effect}}
    +
    \underbrace{q_h\, \delta^2}_{\text{higher-order effect}},
    \label{eq:car_three_terms}
    \ee
    where $s$ is the realized value of state $s_{t-1}$, and the \emph{state-dependent loading} $a_h$ and the \emph{higher-order coefficient} $q_h$ are
    \be
    a_h \equiv \sigma \phi_1^{h}\!\left(\gamma + 2\phi_2\frac{1-\phi_1^h}{1-\phi_1}\right),
    \qquad
    q_h \equiv \phi_2 \sigma^2 \frac{\phi_1^{h-1}-\phi_1^{2h-1}}{1-\phi_1}.
    \label{eq:ah_qh_def}
    \ee
\end{proposition}

Although we initially define the IRF as a function of the realized value of the past information set $\mathcal{F}_{t-1}$, our calculations indicate that the CAR depends only on the realized value $s$ of the state variable $s_{t-1}$. It is therefore sufficient to redefine this causal parameter as a function of $s$ alone.

The decomposition in \eqref{eq:car_three_terms} highlights three distinct components of the true impulse response. The \emph{baseline effect} $\sigma\phi_1^h\delta$ is the response that a linear AR(1) model would produce---it is the component that does not depend on the state or the shock size nonlinearly. The \emph{state-dependent effect} $a_h\, s\,\delta$ captures how the response varies with the initial state $s_{t-1}$: when $a_h>0$, a higher state amplifies the shock's effect. The \emph{higher-order effect} $q_h\,\delta^2$ captures the nonlinear dependence on the shock size, and this effect does not depend on the state $s_{t-1}$. The coefficients $a_h$ and $q_h$ will serve as the key building blocks throughout the analysis.

\begin{remark}[Alternative IRF Concepts for Nonlinear Models]
    The literature offers various IRF definitions that differ in (i) whether they study infinitesimal shocks (the conditional marginal response in \cite{GonccalvesEtAl2024}) or finite shocks (the CAR or the generalized IRF in \cite{KoopEtAl1996}), (ii) whether they condition on the state at $t-1$ (CAR) or average over them (KP's concept), and (iii) in regime-switching contexts, whether states evolve naturally or remain fixed (regime-dependent impulse responses in \cite{EhrmannEtAl2003}). These choices reflect different economic questions: state-dependent IRFs reveal varying transmission mechanisms across states; unconditional IRFs provide average treatment effects; evolving states capture full propagation for policy analysis; fixed states isolate within-regime mechanisms. See the Online Appendix for formal definitions and relationships.
\end{remark}

\subsection{Why Linear Local Projections Fail in QAR} 
\label{subsec:failure.lp}

From now on, we suppose the empirical researcher directly observes the shock of interest $u_t$, and aims to use local projections to study the effect of shock $u_t$ on outcome variable $y_{t+h}$. This is a common empirical practice, where researchers use previously identified shocks---such as monetary policy shocks obtained via narrative approaches---as regressors in their LP specifications. We begin with a purely linear LP specification, which we refer to as \textit{Linear}:
\be
y_{t+h} = \beta_h u_t + \pi_h^{\prime} W_t + \epsilon_{h,t+h},
\label{spec:pure.linear}
\ee
where $W_t$ is a vector of control variables that are independent of shock $u_{t}$. Since the true DGP is nonlinear, \eqref{spec:pure.linear} is generally misspecified; nevertheless, the population coefficients $(\beta_h,\pi_h)$ are well-defined as the coefficients in the population linear projection (equivalently, the probability limit of the OLS estimator) of $y_{t+h}$ on $(u_t,W_t)$ under the true DGP. These coefficients define the population IRF as the counterfactual response obtained by perturbing the shock in the population regression while holding other regressors fixed.
We use the same definition for every other empirical LP specification considered below.

We first show that linear LP fails to capture any nonlinearities when the true underlying DGP is QAR(1,1), which reproduces the finding in \cite{HerbstJohannsen2025}.
\begin{assumption} \label{assump:qar.dgp}
 Assume that $\left\{y_t, s_t, u_t\right\}$ is generated by the QAR(1,1) model, with the process initialized in the infinite past. 
\end{assumption}

\begin{proposition} \label{prop:failure.lp}
   Under Assumption~\ref{assump:qar.dgp}, the population IRF to a shock of magnitude $\delta$ to $u_t$ implied by \textit{Linear} is given by $\operatorname{IRF}^{Linear}(\delta; h) = \beta_h \delta = \sigma \phi_1^h \delta.$ This expression coincides with the population IRF under the true DGP being an AR(1) process:
$y_{t} = \phi_1 y_{t-1} + \sigma u_t, \quad u_t \overset{\text{i.i.d.}}{\sim} \mathcal{N}(0,1).$
\end{proposition}

\noindent {\bf Comparison with \cite{KolesarPlagborg-Moller2025}.\ } 
KP show that the population coefficient of a linear LP with observed shocks provides a scalar causal summary of the nonlinear causal effects. Formally, they prove that under some mild assumptions, the population coefficient
\be
\beta_h=\int \omega_u(u) \Psi_h^{\prime}(u) d u, \quad \text { where } \quad \omega_u(u) \equiv \frac{\operatorname{Cov}\left(\mathbbm{1}\left\{u_t \geq u\right\}, u_t\right)}{\Var\left(u_t\right)} \ge 0,
\label{eq:causal.weight}
\ee
i.e., $\beta_h$ is the weighted average of marginal effect $\Psi_h^{\prime}(u)$.
At first glance, this appears to contradict our result: when the true DGP is QAR(1,1), Proposition~\ref{prop:failure.lp} suggests that the linear LP coefficient does not provide any information on the nonlinearities in the model, how can it be a summary of the nonlinear effects?

To reconcile this, first note that the relevant causal objects are different. KP characterize the derivative of the unconditional average structural function \(\Psi_h(u)\), whereas our benchmark CAR is defined conditional on the realized state. Because 
$
\Psi_h(u)=\mathbb{E}[\psi_h(u,U_{h,t+h})]
$
averages over \(U_{h,t+h}\), and hence over \(s_{t-1}\), any state dependence is integrated out by construction before KP's causal weighting is applied. In our QAR setting, the only nonlinearity that survives in \(\Psi_h(u)\) is therefore the quadratic term in \(u\). For QAR(1,1), average structural function
$$
\Psi_h(u) = \mathbb{E}\left[\psi_h\left(u, U_{h,t+h}\right)\right]  = \begin{dcases}
     \sigma \phi_1^{h} u + \text{const}, \quad h=0,\\
     \sigma \phi_1^{h} u +  q_h\, u^2+ \text{const}, \quad h\ge 1,\\
\end{dcases} 
$$
where $q_h$ is the higher-order coefficient defined in \eqref{eq:ah_qh_def}.
Proposition 1 in KP implies that, for $h\ge 1$,
\begin{equation*}
\begin{aligned}
\beta_h & =\int_{\mathbb{R}} \omega_u(u) \Psi_h^{\prime}(u) \ d u \\
& = \sigma \phi_1^{h} \underbrace{\int_{\mathbb{R}} \omega_u(u) \ d u}_{=1} + 2 q_h \underbrace{\int_{\mathbb{R}} \omega_u(u) u\ d u}_{=0} = \sigma \phi_1^{h}.
\end{aligned}
\end{equation*}
as $\omega_u(u) = \operatorname{Cov}\left(\mathbbm{1}\left\{u_t \geq u\right\}, u_t\right)/ \Var\left(u_t\right)$: (a) is symmetric around $0$ if $u_t$ is symmetrically distributed around 0, (b) integrates to 1. For $h=0,$ the marginal effect is constant, i.e., $\Psi_h^{\prime}(u)=\sigma$ for all $u$. This is because the second-order effect on $y_{t+h}$ arises from the quadratic term  $s_{t+h-1}^2$, and $u_t$ influences this term only starting from $h=1$ onwards. Therefore, $u_t$ affects $y_t$ solely through the first-order term $(1+\gamma s_{t-1})\sigma u_t$ and we have $\beta_0 = \sigma.$ These results coincide with what we found in our previous calculation. 

The key insight, echoing the discussion in \citet{GoncalvesEtAl2025}, is that the marginal second-order effects---an odd function of the shock---cancels between positive and negative shocks when the causal weight function is symmetric, i.e., $\omega_u(u)=\omega_u(-u).$ In our setting, this symmetry arises under the assumption that shock $u_t$ is symmetrically distributed around zero.
Do empirical shock measures satisfy this symmetry condition? KP estimate causal weight functions for a range of identified shocks in the literature. Several well-cited shock measures---such as the government spending shock in \cite{BlanchardPerotti2002} and the technology shock in \cite{Fernald2014}---exhibit approximately symmetric distributions. In contrast, other shocks, such as the military news shock in \cite{Ramey2011}, display extreme fat-tailed or asymmetric properties. For the former class of shocks, one should be cautious when using a purely linear LP to capture nonlinear effects, whereas for the latter, a linear LP may still provide a reasonable summary of the overall nonlinear response.

\begin{remark}

If the true DGP features not only second-order dynamics but also third-order effects, then the linear LP estimates will provide a causal summary of the first- and third-order effects while missing the second-order effects. This occurs because the marginal third-order effects do not integrate to zero,  i.e. $\int_{\mathbb{R}} \omega_u(u) u^2 \ d u \neq 0$, unlike the second-order effects.
\end{remark}

\section{Performance of State-Dependent Local Projections}
\label{sec:state.lp.how.well}

So far, we have shown that linear LP fails to capture any nonlinearities when the true DGP follows a QAR(1,1) model. In empirical work, researchers often use state-dependent LPs to capture nonlinearities. This raises the central question of the paper: to what extent do commonly used state-dependent LP specifications recover the true impulse responses when the underlying process is nonlinear, specifically in the case of a QAR(1,1) model?

\subsection{State-Dependent Local Projections}
\label{subsec:state.lp}
 In this section, we present four state-dependent LP specifications---two from the existing literature and two that we propose---and derive their implied population IRFs under the QAR(1,1) DGP. The population IRFs, expressed in terms of the building blocks $a_h$ and $q_h$ defined in \eqref{eq:ah_qh_def}, will serve as the basis for evaluating each specification's ability to approximate the true CAR.

\medskip
\noindent{\textbf{Shock-Based Specification.}\ }
We consider the following empirical specification, referred to as \textit{AsymLP}:
\be
y_{t+h}=S_{t}\left[\alpha_h^{(+)} + \beta_h^{(+)} u_{t}+ {\pi_{h}^{(+)}}^{\prime} W_t\right]+\left(1-S_{t}\right)\left[\alpha_h^{(-)} + \beta_h^{(-)} u_{t}+ {\pi_{h}^{(-)}}^{ \prime} W_t\right]+ \epsilon_{h,t+h},
\label{spec:shock.state}
\ee
where binary state $S_t=\mathbbm{1}\{u_{t}>0\}$ indicates whether the contemporaneous shock $u_t$ is positive or not, and $W_t$ is a vector of control variables that are independent of shock $u_{t}$ (e.g., lagged outcome and shocks) but exclude the constant term. $\beta_h^{(+)}$ can be interpreted as the impulse response coefficient when $S_{t}=1,$ and similar for $\beta_h^{(-)}.$ $\alpha_h^{(+)}$ and $\alpha_h^{(-)}$ are state-dependent constant terms. The choice of state variable $S_t$ is intended to capture the asymmetric effects of a shock.

Proposition \ref{prop:irf.shock.state} provides the implied population IRF of  \textit{AsymLP} under the QAR(1,1) DGP.

\begin{proposition} \label{prop:irf.shock.state}
    Under Assumption~\ref{assump:qar.dgp}, the population IRF to a shock of magnitude $\delta$ to $u_t$ implied by \textit{AsymLP} is    
    \be
        \operatorname{IRF}^{AsymLP}(S, \delta; h) = \begin{cases}
            \beta_h^{(+)} \delta, \text{\quad if } S=1, \\
            \beta_h^{(-)} \delta, \text{\quad if } S=0,
        \end{cases}
    \ee
    where, with $q_h$ as in \eqref{eq:ah_qh_def} and $m\equiv \sqrt{2/\pi}/(1-2/\pi)$,
    \be
    \beta_h^{(+)}=\sigma\phi_1^{h} + m \cdot q_h,
    \qquad
    \beta_h^{(-)}=\sigma\phi_1^{h} - m \cdot q_h.
    \label{eq:beta_asym_simplified}
    \ee
\end{proposition}

\begin{remark}
    One should always include a constant term in \textit{AsymLP}. Otherwise, the population coefficient
   $
    \beta_h^{(+)}  = \mathbb{E}[S_t u_t y_{t+h}]/\mathbb{E}[S_t u_t^2]. 
   $
    We can show that when $h\rightarrow \infty,$ $\beta_h^{(+)}$ takes the limit $2\phi_2 \sigma^2 \phi(0)/(1-\phi_1)(1-\phi_1^2).$ Then the population IRF does not converge to zero when $u_t>0$. Similar for the case where $u_t\le 0.$
\end{remark}

\noindent{\textbf{Lag-Based Specification.}\ }
We consider the following empirical specification, referred to as \textit{LagLP}:
\be
y_{t+h}= \beta_h^{(0)} u_{t}+ {\pi_{h}^{(0)}}^{\prime} W_t + y_{t-1} \left(\beta_h^{(1)} u_{t}+ {\pi_{h}^{(1)}}^{\prime} W_t\right) +\epsilon_{h,t+h},
\label{cont.state.spec}
\ee
where lagged outcome variable $y_{t-1}$ serves as a continuous state variable, and $W_t$ is a vector of control variables (possibly including a constant term) that is independent of shock $u_t$. This is the order-$1$ polynomial state-dependent LPs considered in \cite{AuerEtAl2021}. The choice of state variable $y_{t-1}$ is intended to capture different effects of a shock when the lagged outcome variable was at different levels (e.g., the economy is at recession or expansion).

Proposition \ref{prop:irf.lag.state} provides the implied population IRF of \textit{LagLP} under the QAR(1,1) DGP.
\begin{proposition} \label{prop:irf.lag.state}
    Under Assumption~\ref{assump:qar.dgp}, the population IRF to a shock of magnitude $\delta$ to $u_t$ implied by \textit{LagLP} is
    \be
    \operatorname{IRF}^{LagLP}(y, \delta; h) = (\beta_h^{(0)} + \beta_h^{(1)} y)\delta,
    \ee
    where, with $a_h$ as in \eqref{eq:ah_qh_def},
    \be
    \beta_h^{(1)}
    =
    a_h\cdot \frac{\sigma^2/(1-\phi_1^2)}{\Var(y_{t-1})},
    \qquad
    \beta_h^{(0)}=\sigma\phi_1^{h}-\beta_h^{(1)}\EE[y_{t-1}].
    \label{eq:beta_lag_simplified}
    \ee
\end{proposition}

\begin{remark}
    Another popular state-dependent LP specification that constructs the state variable based on lagged variables is:
$$
y_{t+h}=F(z_{t-1})\left[\alpha_h^{(R)} + \beta_h^{(R)} u_{t}+ {\pi_{h}^{(R)}}^{\prime} W_t\right]+\left(1-F(z_{t-1})\right)\left[\alpha_h^{(E)} + \beta_h^{(E)} u_{t}+ {\pi_{h}^{(E)}}^{ \prime} W_t\right]+ \epsilon_{h,t+h},
$$
where $F(z_{t-1})\in [0,1]$ denotes the probability of recession estimated from lagged variables $z_{t-1}$, and $\beta_h^{(R)}$ and $\beta_h^{(E)}$ represent the impulse response coefficients condition on the economy being in a recession or an expansion, respectively. This specification can be viewed as the LP analog of the smooth-transition VAR proposed by \cite{AuerbachGorodnichenko2012}.
\end{remark}

\begin{remark}
    One may combine the shock-based and lag-based specifications by interacting the lagged observable with the binary state $S_t=\mathbbm{1}\{u_{t}>0\}$. This approach has been used in several empirical macroeconomics papers; see, for example, \cite{AlesinaEtAl2018}, \cite{AuerbachGorodnichenko2016}, \cite{BernardiniEtAl2020}, and \cite{BornEtAl2020}. In our setting, we can specify the empirical specification \textit{Mixed} as:
$$
\begin{aligned}
y_{t+h} &= S_{t}\left[\beta_h^{(0, +)} u_{t}+ {\pi_{h}^{(0, +)}}^{\prime} W_t + y_{t-1} \left(\beta_h^{(1, +)} u_{t}+ {\pi_{h}^{(1, +)}}^{\prime} W_t\right)\right] \\
& + (1 - S_{t}) \left[\beta_h^{(0, -)} u_{t}+ {\pi_{h}^{(0, -)}}^{\prime} W_t + y_{t-1} \left(\beta_h^{(1, -)} u_{t}+ {\pi_{h}^{(1, -)}}^{\prime} W_t\right)\right]  +\epsilon_{h,t+h}.
\end{aligned}
$$
\end{remark}

\noindent{\textbf{Targeting the True Impulse Responses.}\ }
The decomposition of the CAR in \eqref{eq:car_three_terms} suggests a general design principle: the nonlinearity in the LP specification must match the nonlinearity in the data-generating process.  Neither \textit{AsymLP} nor \textit{LagLP} includes regressors for \emph{both} the state-dependent effect and the higher-order effect, so neither can fully approximate the CAR.

We now introduce two specifications that are guided by this principle. The first is infeasible but serves as a benchmark; the second is its implementable analogue.

\noindent {\bf Infeasible Specification.\ } 
The three-term structure of the CAR suggests the following (infeasible) empirical specification, referred to as \textit{Infeas}:
\be
    y_{t+h} = \kappa_{h0} + \kappa_{h1} u_t + \kappa_{h2}\, s_{t-1}u_t + \kappa_{h3}\, u_t^2 + \epsilon_{h,t+h}.
\label{infeasible.spec}
\ee
What does specification \textit{Infeas} recover under the QAR(1,1) DGP?

\begin{proposition} \label{prop:irf.infeas}
    Under Assumption~\ref{assump:qar.dgp}, the population IRF to a shock of magnitude $\delta$ to $u_t$ implied by \textit{Infeas} is $\operatorname{IRF}^{Infeas}(s, \delta; h) = \kappa_{h1} \delta + \kappa_{h2} \,s \delta + \kappa_{h3}\, \delta^2$, with $\kappa_{h1}$, $\kappa_{h2}$, and $\kappa_{h3}$  such that the IRF coincides with the true CAR.
\end{proposition}

\noindent {\bf Feasible Specification.\ } Since $s_{t-1}$ is not directly observable, one alternative would be using $y_{t-1}$ as a proxy, i.e.,
\be
    y_{t+h} = \theta_{h0} + \theta_{h1} u_t + \theta_{h2}\, y_{t-1}u_t + \theta_{h3}\, u_t^2 + \pi_h^{\prime} W_t + \epsilon_{h,t+h},
\label{feasible.spec}
\ee
where $W_t$ is a vector of control variables that are independent of shock $u_{t}$.
We refer to this specification as \textit{Feas}.
A natural question is, how well does this feasible specification perform in capturing the nonlinearity? To answer this, we first derive the population IRF implied by \textit{Feas} when the DGP is QAR(1,1), as stated in Proposition~\ref{prop:irf.feas}.

\begin{proposition} \label{prop:irf.feas}
     Under Assumption~\ref{assump:qar.dgp}, the population IRF to a shock of magnitude $\delta$ to $u_t$ implied by \textit{Feas} can be written as
     \be
     \operatorname{IRF}^{Feas}(y, \delta; h) = \beta_h^{(0)} \delta+\beta_h^{(1)} y\delta+q_h \delta^2,
     \label{eq:irf_feas_simplified}
     \ee
     where $\beta_h^{(0)}$ and $\beta_h^{(1)}$ are given in Proposition~\ref{prop:irf.lag.state} and $q_h$ is the higher-order coefficient defined in \eqref{eq:ah_qh_def}.

\end{proposition}

\begin{remark}
   One may also follow \cite{GoncalvesEtAl2024} and estimate the CAR fully nonparametrically; we refer to this approach as \textit{NPLP}. Their idea is to first obtain a nonparametric estimator $\hat{g}_h(s, u)$ of $g_h(s, u) \equiv \mathbb{E}\left[y_{t+h} \mid s_{t-1}=s, u_{t}=u\right]$, and then estimate $\operatorname{CAR}_h(s, \delta)$ as
$$
\widehat{\operatorname{CAR}}_h(s, \delta)=\frac{1}{T} \sum_{t=1}^T\left(\hat{g}_h\left(s, u_{t}+\delta\right)-\hat{g}_h\left(s, u_{t}\right)\right).
$$ Since the latent state variable $s_{t-1}$ is unobservable, one can replace it with its proxy $y_{t-1}$, as in our feasible specification. We implement a control-adjusted version of this approach for comparison in the empirical application in Section~\ref{sec:empiric.exp}.
\label{rem:nplp}
\end{remark}

To summarize the population IRFs: \textit{Linear} captures only the baseline effect $\sigma\phi_1^h\delta$; \textit{AsymLP} additionally captures part of the higher-order effect (through the sign of the shock); \textit{LagLP} additionally captures part of the state-dependent effect (through the observable proxy $y_{t-1}$); and \textit{Feas} captures both the state-dependent and higher-order effects simultaneously. The question remains: how large are these differences quantitatively? We turn to this question next.

\subsection{A Distance Measure for Impulse Responses}
\label{subsec:irf.dist.measure}

To construct a distance measure for comparing IRFs, one might naturally consider the area under the absolute difference between them. However, an important complication arises: the IRFs are defined with respect to different conditioning variables. More precisely, the true CAR is a function of the realized value of latent variable $s_{t-1}$, while the IRFs implied by empirical specifications are functions of the realized value of other observables $z_t$---for example, the binary state $S_t$ in \textit{AsymLP}, the lagged outcome $y_{t-1}$ in \textit{LagLP}, or the pair $\left(y_{t-1}, S_t\right)$ in \textit{Mixed}. As a result, these IRFs cannot be directly compared.

To address this issue, we compare the true CAR and the specification-implied IRF on the same underlying realizations of $\left(s_{t-1},u_t,z_t\right)$ and define distance as their conditional mean squared difference. By varying the conditioning event $\mathcal{C}$, the measure can be evaluated either unconditionally or over particular regions of the state-shock space.

\begin{definition}[Conditional mean-squared approximation error]
\label{def:cond_mse}
Fix a horizon $h\ge 0$, a specification $spec\in\{\textit{Linear},\textit{AsymLP},\textit{LagLP},\textit{Feas}\}$, and a conditioning event $\mathcal{C}$ on $\left(s_{t-1},u_t,z_t\right)$.  The \emph{conditional approximation error} at horizon $h$ is
\be
\mathcal{E}^{spec}_h(\mathcal{C})
\;\equiv\;
\EE\!\left[
\Big(
\operatorname{CAR}_h(s_{t-1},u_t)-\operatorname{IRF}^{spec}(z_t,u_t;h)
\Big)^2
\;\Bigg|\;
\mathcal{C}
\right],
\label{eq:cond_approx_error}
\ee
where $\operatorname{IRF}^{spec}$ is evaluated at the realized conditioning variable of the specification (e.g.\ $y_{t-1}$ for \textit{LagLP} and \textit{Feas}, $S_t$ for \textit{AsymLP}).
The \emph{aggregate distance} over horizons $0,\ldots,H$ is
\be
D^{spec}(\mathcal{C})
\;\equiv\;
\left(\sum_{h=0}^{H}\mathcal{E}^{spec}_h(\mathcal{C})\right)^{1/2}.
\label{eq:aggregate_distance}
\ee
\end{definition}

\begin{remark}[Special cases]\label{rem:distance_special_cases}
The following instantiations of Definition~\ref{def:cond_mse} play distinct roles in the subsequent analysis.
\begin{tlist}
\item \textit{Unconditional distance.\ } Setting $\mathcal{C}=\Omega$ (the entire sample space) gives
\be
D^{spec}(\Omega)
= \mathbb{E}\!\left[\sum_{h=0}^H\Big(\operatorname{CAR}_h(s_{t-1},u_t)-\operatorname{IRF}^{spec}(z_t,u_t;h)\Big)^2\right]^{1/2},
\label{eq:uncond_distance}
\ee
which is the overall distance used in the numerical illustration (Section~\ref{subsec:mc.exp}).
\item \textit{Conditional distance over a region $\mathcal{S}$.\ } Setting $\mathcal{C}=\{(s_{t-1},u_t)\in\mathcal{S}\}$ for a region $\mathcal{S}\subset\RR^2$ gives $D^{spec}(\mathcal{S})$, which is used to pinpoint where in the $(s_{t-1},u_t)$ plane each specification matches the true CAR well (Section~\ref{subsec:mc.exp}, Figure~\ref{fig:avg.dist}).
\item \textit{Shock-conditional MSE.\ } Setting $\mathcal{C}=\{u_t=\delta\}$ and focusing on a single horizon yields
\be
\mathcal{L}^{spec}_h(\delta)
\;\equiv\;
\mathcal{E}^{spec}_h\!\left(\{u_t=\delta\}\right),
\label{eq:def_loss_u_main}
\ee
whose closed-form expressions are derived in Theorem~\ref{thm:cond_u_main}.
\item \textit{State-conditional MSE.\ } Setting $\mathcal{C}=\{s_{t-1}=s\}$ and focusing on a single horizon yields
\be
\mathcal{R}^{spec}_h(s)
\;\equiv\;
\mathcal{E}^{spec}_h\!\left(\{s_{t-1}=s\}\right),
\label{eq:def_loss_s_main}
\ee
whose closed-form expressions are derived in Theorem~\ref{thm:cond_s_main}.
\end{tlist}

\noindent Note that these cases are connected by the law of iterated expectations:
$$
D^{spec}(\Omega)^2 = \sum_{h=0}^{H}\EE\!\left[\mathcal{L}^{spec}_h(u_t)\right] = \sum_{h=0}^{H}\EE\!\left[\mathcal{R}^{spec}_h(s_{t-1})\right].
$$
\end{remark}

The unconditional distance \eqref{eq:uncond_distance} can be estimated through the following procedure. First, we simulate the model once for $T$ periods, obtaining $\{(s_{t-1}, u_t, z_t)\}_{t=1}^{T}$. Next, for each $(s_{t-1}, u_t, z_t)$, we compute $\Delta(s_{t-1}, u_t, z_t):= \sum_{h=0}^H\left(\operatorname{CAR}_h(s_{t-1}, u_t)-\operatorname{IRF}^{spec}(z_t, u_t; h)\right)^2.$
Finally, we calculate the distance as
$
\hat{D}=\left(\frac{1}{T}\sum_{t=1}^{T}\Delta(s_{t-1}, u_t, z_t)\right)^{1/2}.
$

\subsection{Numerical Illustration}
\label{subsec:mc.exp}
In this section, we provide a numerical illustration of the population approximation errors of the four LP specifications, drawing on the population IRFs derived in Section~\ref{subsec:state.lp} and the distance measure introduced in Section~\ref{subsec:irf.dist.measure}. We simulate the QAR(1,1) model in Eq~(\ref{QAR.hetero}) once with a large sample size of $T=10,000$ so that the sample analogue $\hat D$ is numerically close to the underlying population distance. The goal is therefore to isolate population misspecification bias rather than to evaluate finite-sample performance. We set the model parameters as follows: $H=10$, $\phi_1=0.5$, $\sigma=1$, $\phi_2=0.2$, and $\gamma=0.1$, where $\phi_2$ and $\gamma$ govern the degree of nonlinearity in the model. These values are chosen to match the scale of the posterior-median estimates reported by ABS when fitting the QAR model to U.S. data.

Figures~\ref{fig:compare.irf.s} and~\ref{fig:compare.irf.u} compare the IRFs implied by the model and by empirical specifications. In Figure~\ref{fig:compare.irf.s}, we fix the shock size to 1. We observe that the impulse response becomes larger (in absolute value) as \( s_{t-1} \) increases. This reflects the state-dependent nature of the true IRF. The source of this pattern can be seen from the formula for the CAR: the state-dependent loading $a_h$ is strictly positive under our parameterization, so \( s_{t-1} \) acts as an amplifier of the shock's effect. \textit{LagLP} qualitatively captures this relationship: the response increases with \( y_{t-1} \), which serves as a proxy for \( s_{t-1} \). By contrast, the IRF implied by \textit{Linear} is invariant across all values of \( s \). In Figure~\ref{fig:compare.irf.u}, we vary the magnitude of the shock \( \delta \). The true CAR exhibits asymmetric responses: positive and negative shocks of the same magnitude lead to effects of different sizes in absolute value. This asymmetry arises from the higher-order coefficient $q_h$, which is positive under our parameterization. \textit{AsymLP} captures this feature: a positive shock (\( S = 1 \)) leads to a larger response than a negative one (\( S = 0 \)). The IRFs implied by \textit{Linear}, however, are symmetric.

\begin{figure}[t!]
  \setlength{\abovecaptionskip}{0cm}
	\caption{Comparing Impulse Responses: Varying $s$ and $y$}
	\label{fig:compare.irf.s}  
	\begin{center} 
		\begin{tabular}{cc}
			True & \textit{Linear} \\
			\includegraphics[width=3in]{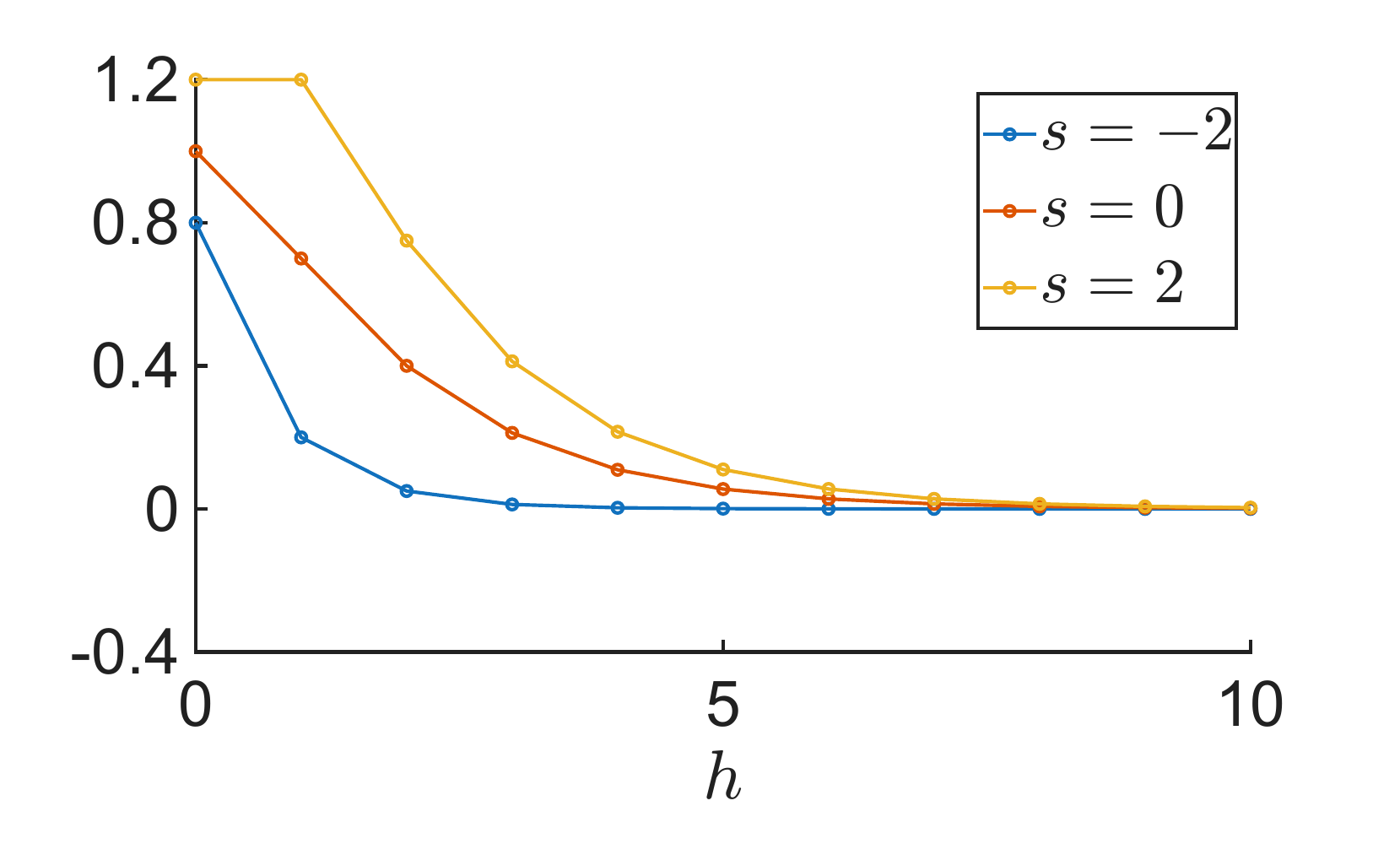} &
			\includegraphics[width=3in]{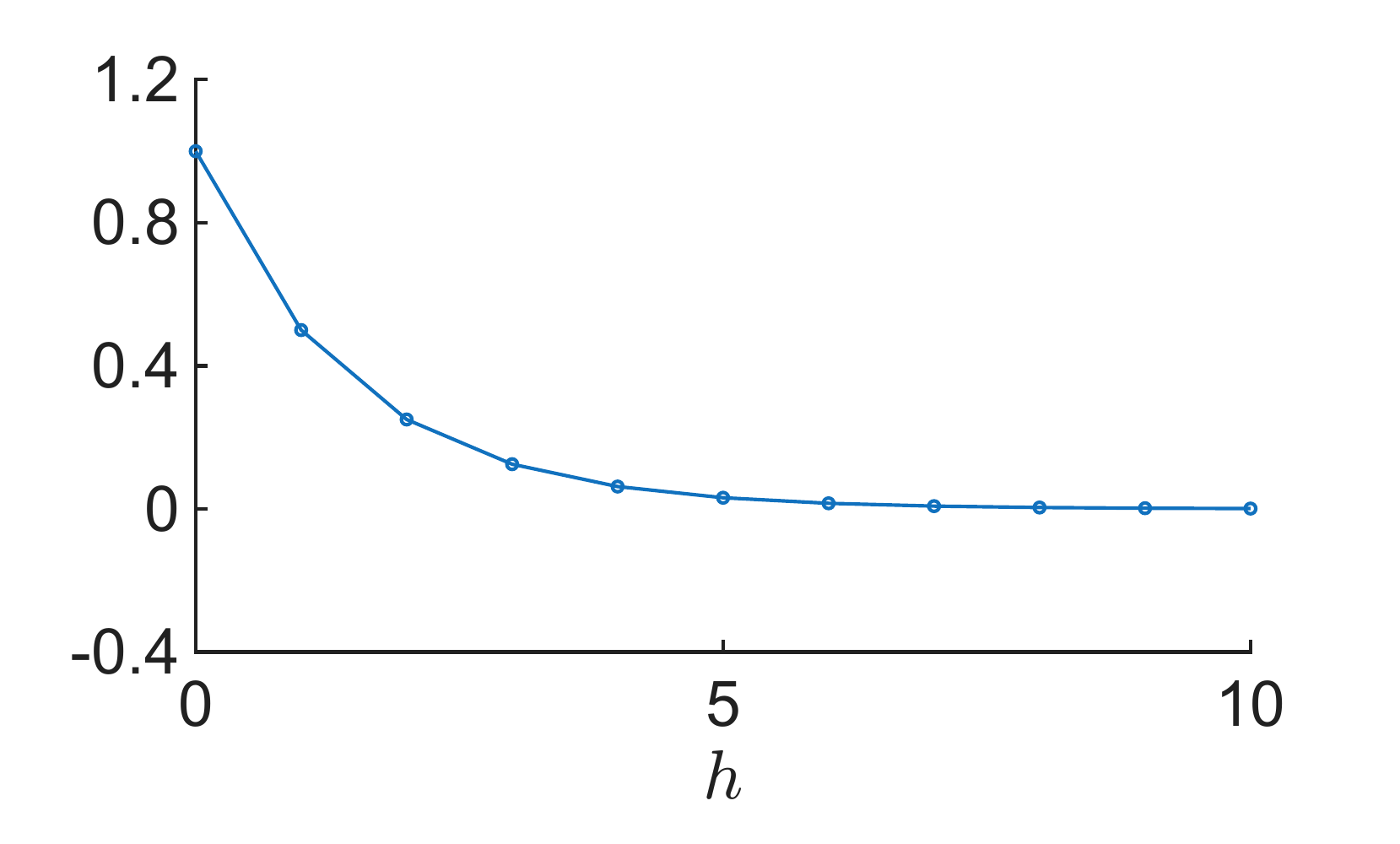} \\
			  \textit{LagLP} & \\
			\includegraphics[width=3in]{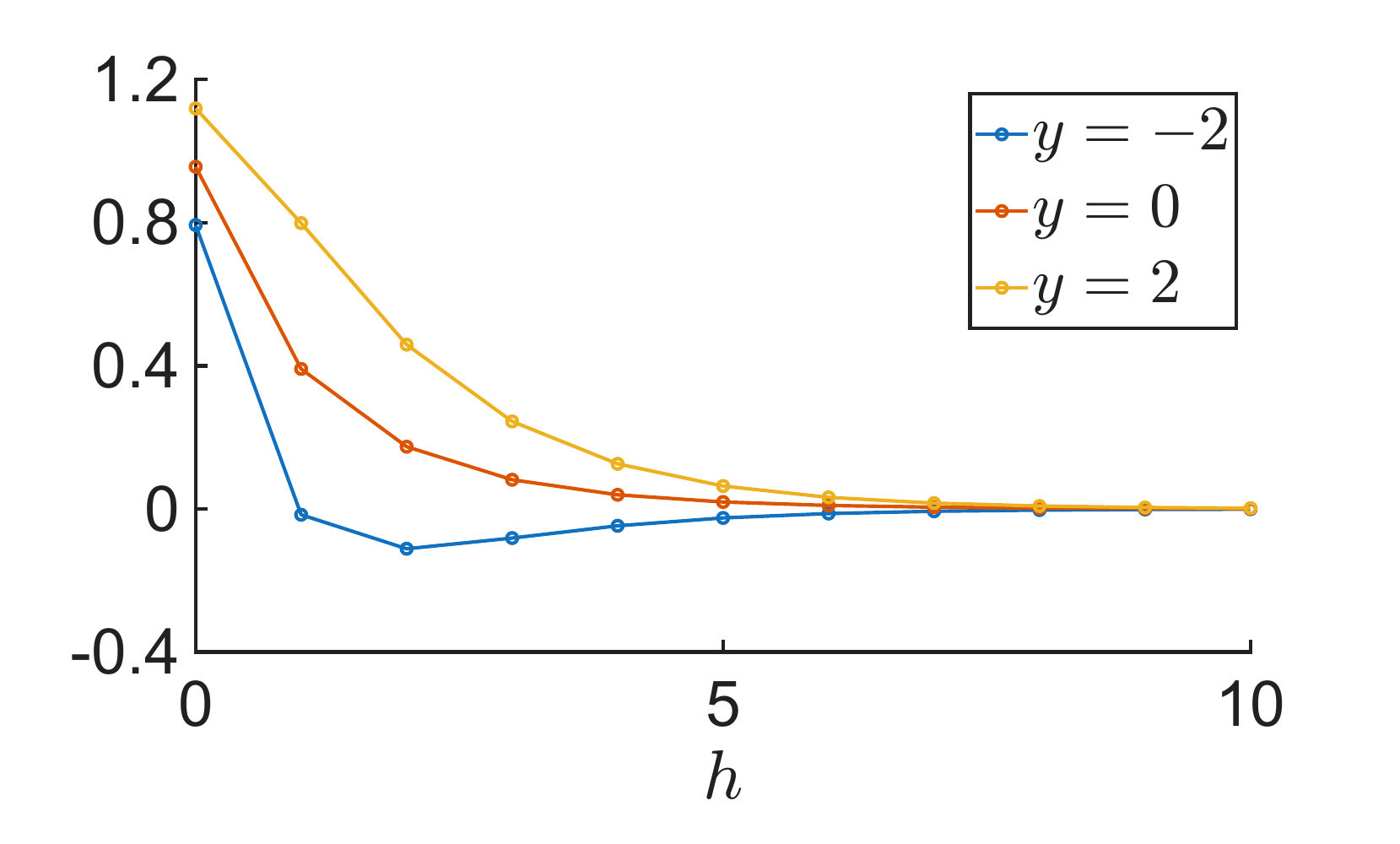} & \\		
		\end{tabular}
	\end{center}
	{\footnotesize {\em Notes}: IRFs when the magnitude of shock is fixed at $\delta=1$. True: $\operatorname{CAR}_h(-2, 1)$, $\operatorname{CAR}_h(0, 1)$, and $\operatorname{CAR}_h(2, 1)$; \textit{Linear}: $\operatorname{IRF}^{Linear}(1; h)$; \textit{LagLP}: $\operatorname{IRF}^{LagLP}(-2, 1; h)$, $\operatorname{IRF}^{LagLP}(0, 1; h)$, $\operatorname{IRF}^{LagLP}(2, 1; h)$.}\setlength{\baselineskip}{4mm}
\end{figure}

\begin{figure}[t!]
  \setlength{\abovecaptionskip}{0cm}
	\caption{Comparing Impulse Responses: Varying $\delta$ and $S$}
	\label{fig:compare.irf.u}  
	\begin{center} 
		\begin{tabular}{cc}
			True & \textit{AsymLP} \\
			\includegraphics[width=3in]{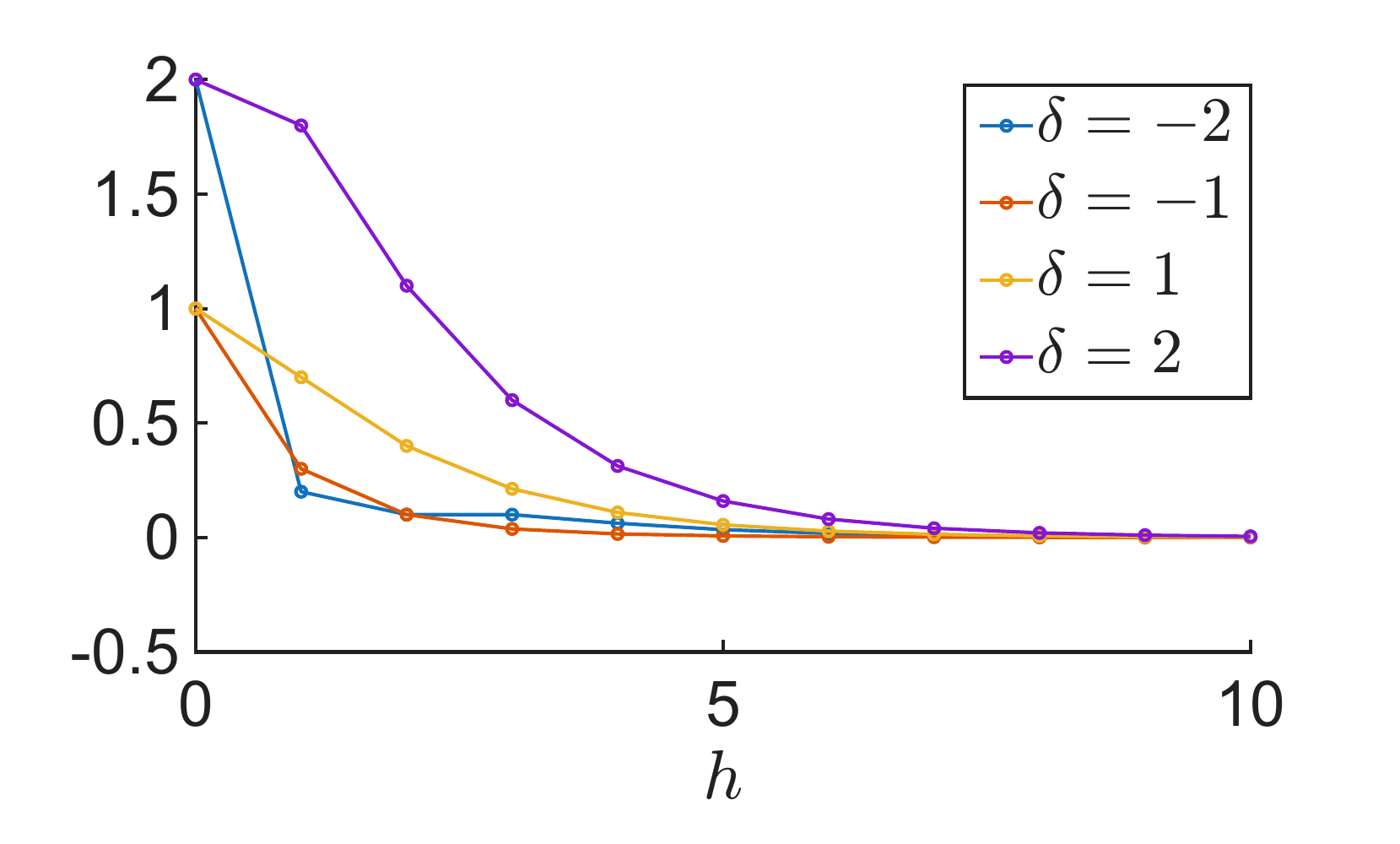} &
			\includegraphics[width=3in]{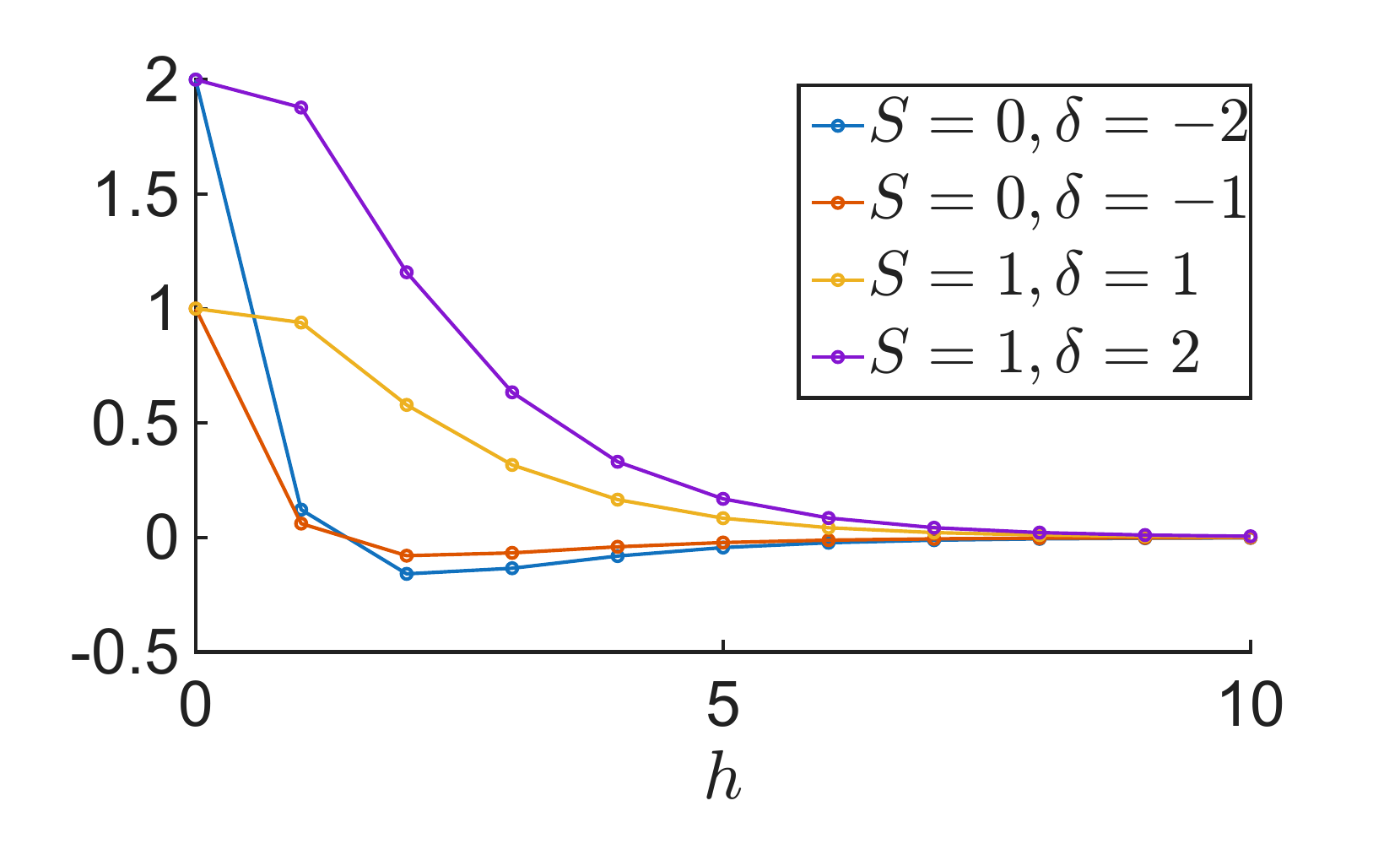} \\
		\end{tabular}
	\end{center}
	{\footnotesize {\em Notes}: We fix $s=0$ for the true CAR. 
    True: $-\operatorname{CAR}_h(0, -2)$, $-\operatorname{CAR}_h(0, -1)$, $\operatorname{CAR}_h(0, 1)$, $\operatorname{CAR}_h(0, 2)$; \textit{AsymLP}: $-\operatorname{IRF}^{AsymLP}(0, -2; h)$, $-\operatorname{IRF}^{AsymLP}(0, -1; h)$, $\operatorname{IRF}^{AsymLP}(1, 1; h)$, $\operatorname{IRF}^{AsymLP}(1, 2; h)$. We negate the impulse response of a negative shock for better comparison.}\setlength{\baselineskip}{4mm}
\end{figure}

The estimated unconditional distance measure \( \hat{D} \) is 0.61 for the purely linear LP, 0.47 for \textit{AsymLP}, 0.50 for \textit{LagLP}, and 0.18 for \textit{Feas}, confirming that both existing state-dependent specifications improve upon the linear benchmark and that \textit{Feas} delivers the best overall approximation.
To further understand where these improvements come from, we estimate the conditional distance over bins of either \( s_{t-1} \) or \( u_t \), as defined in Remark~\ref{rem:distance_special_cases}(ii):
\[
\hat D = \left(\frac{1}{|\mathcal{I}|} \sum_{t \in \mathcal{I}} \Delta(s_{t-1}, u_t, z_t)\right)^{1/2}, \quad \text{where } \mathcal{I} = \{t \mid s^{i-1} \leq s_{t-1} \leq s^i \} \text{ or } \{t \mid u^{i-1} \leq u_t \leq u^i \}.
\]

The first column of Figure~\ref{fig:avg.dist} reports the conditional distance over $s_{t-1}$ bins. \textit{LagLP} outperforms \textit{Linear}, while \textit{AsymLP} yields no material improvement---both \textit{AsymLP} and \textit{Linear} fail to capture state dependence. The second column shows the conditional distance over $u_t$ bins. Here, \textit{AsymLP} outperforms \textit{Linear}, while \textit{LagLP} offers little gain---both \textit{LagLP} and \textit{Linear} fail to capture the higher-order effects. These patterns confirm that each existing state-dependent specification captures a distinct aspect of nonlinearity: \textit{LagLP} captures state dependence since $y_{t-1}$ approximates the latent state $s_{t-1}$, while \textit{AsymLP} with $S_t=\mathbbm{1}\{u_{t}>0\}$ captures the higher-order effects through asymmetric responses to positive and negative shocks.

A closer examination of the first three rows of Figure~\ref{fig:avg.dist} reveals where these improvements concentrate. For \textit{LagLP}, the gains over \textit{Linear} occur primarily in the extreme bins of $s_{t-1}$, with minimal improvement in the middle bins. Similarly, \textit{AsymLP}'s gains over \textit{Linear} appear only in the tails of the shock distribution, with negligible improvement in the center. This pattern reveals an important caveat: while state-dependent specifications improve upon linear LPs on average, these gains concentrate at the extremes. Both specifications perform similarly to the linear LP in the middle regions of their respective distributions---where most observations typically lie. This finding has clear empirical implications. Researchers studying large, rare shocks would benefit from shock-based specifications, as these events fall in the tails where such specifications excel. Those examining responses in extreme economic conditions should employ lag-based specifications, which capture state dependence at the extremes. However, for typical shocks in normal conditions, the added complexity of state-dependent specifications offers limited value over simpler linear LPs.

The bottom row of Figure~\ref{fig:avg.dist} presents the results for \textit{Feas}. The conditional distance to the true CAR, measured over both $s_{t-1}$ bins and $u_t$ bins, is visibly smaller than that under all three preceding specifications.
Crucially, unlike the existing state-dependent specifications that only improve upon linear LPs at the tails, our proposed specification achieves substantial improvements across the entire distribution—both in the extreme bins and, importantly, in the middle regions where most observations lie. This uniform improvement is reflected in the large reduction in conditional distance: \textit{Feas} yields a mean distance of 0.18, which is the smallest across all specifications.

\begin{figure}[t!]
  \setlength{\abovecaptionskip}{0cm}
  \caption{Conditional Distance by $s_{t-1}$ and $u_t$ Bins}
  \label{fig:avg.dist}
  \begin{center}
  \renewcommand{\arraystretch}{0.8} 
  \begin{tabular}{ccc}
      & $s_{t-1}$ bins & $u_t$ bins \\[-10pt]
      \rotatebox{90}{\parbox{5 cm}{\centering \textit{Linear}}}
        & \includegraphics[width=0.42\textwidth]{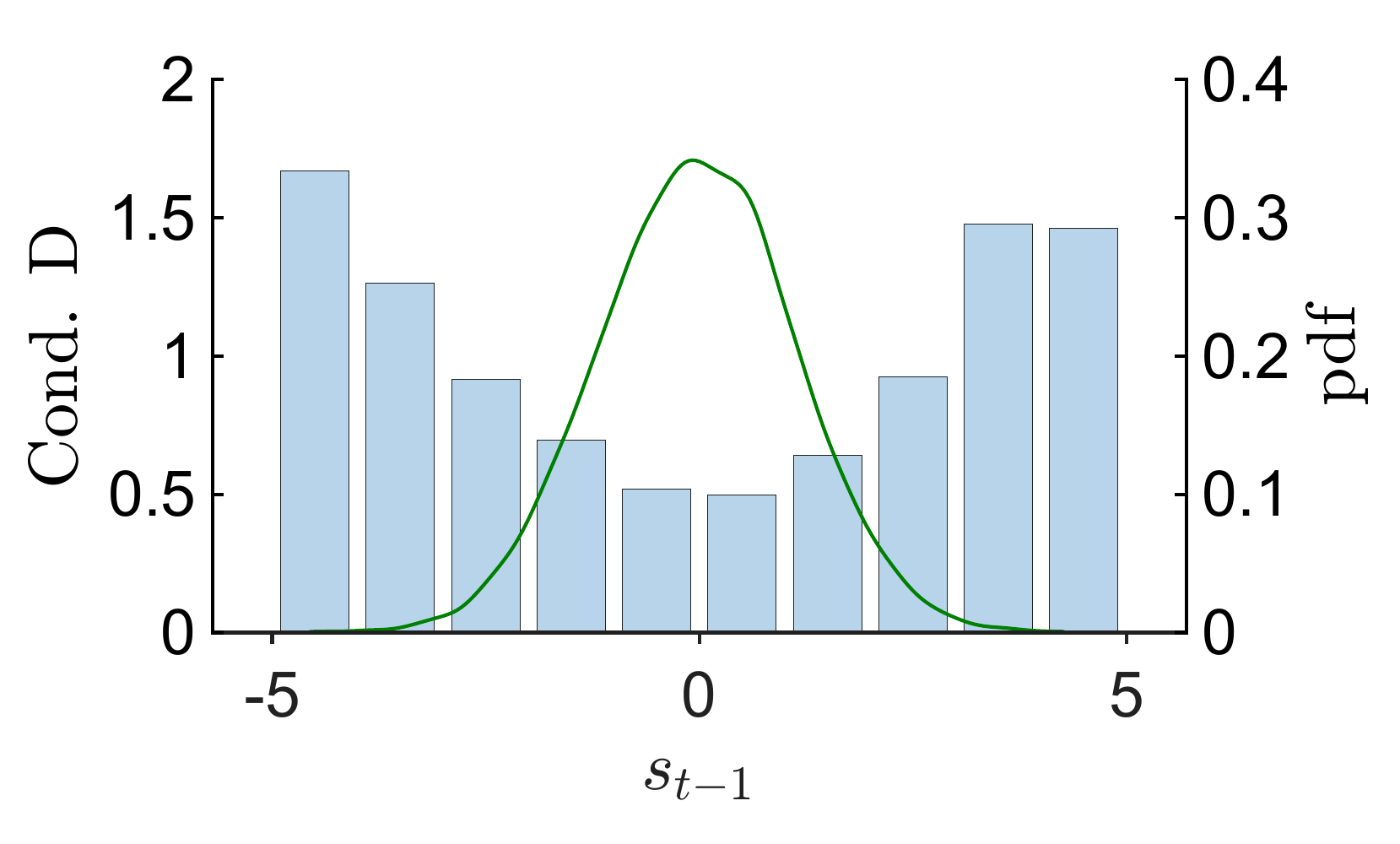}
        & \includegraphics[width=0.42\textwidth]{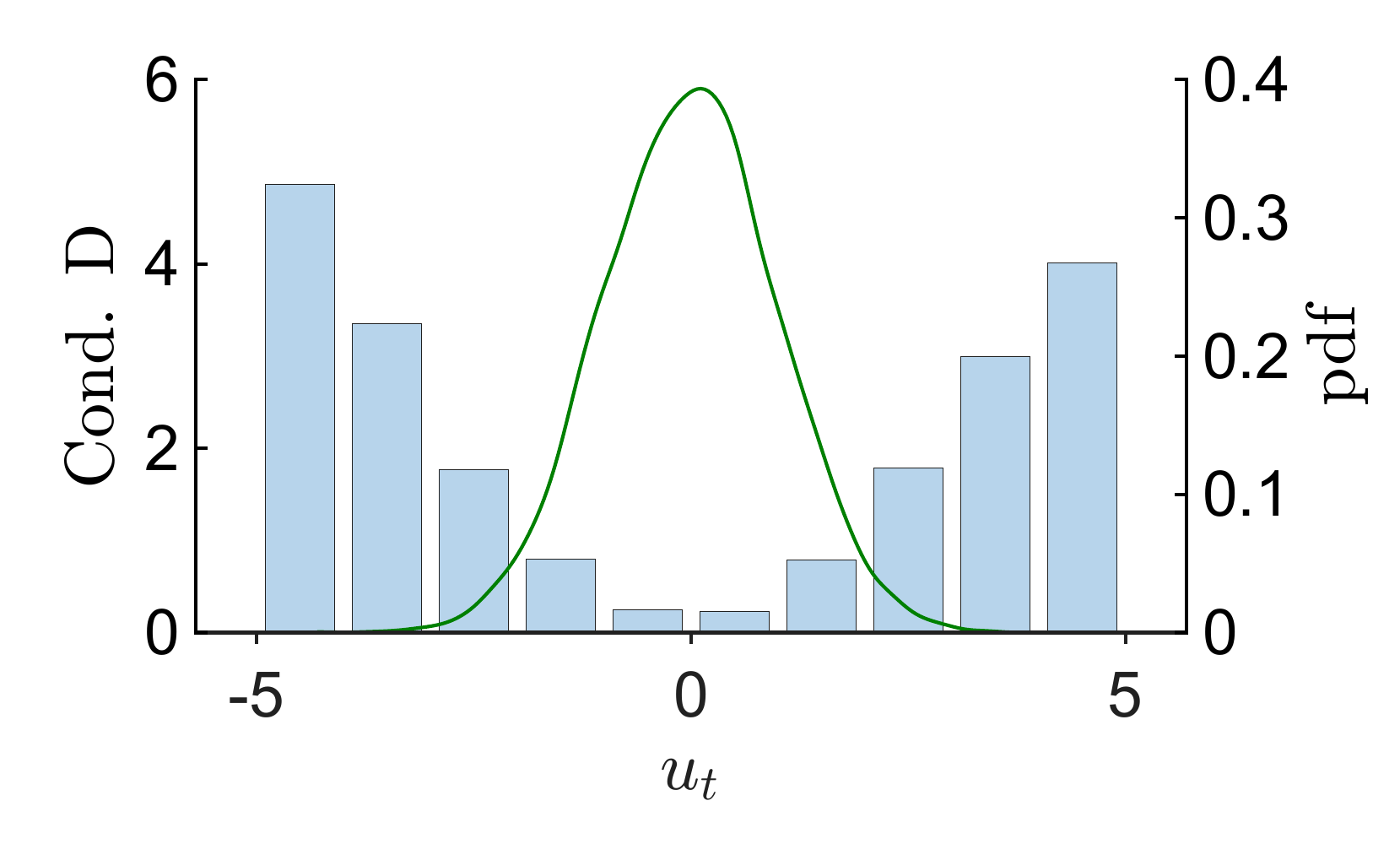} \\[-10pt]
     \rotatebox{90}{\parbox{5 cm}{\centering \textit{AsymLP}}}
        & \includegraphics[width=0.42\textwidth]{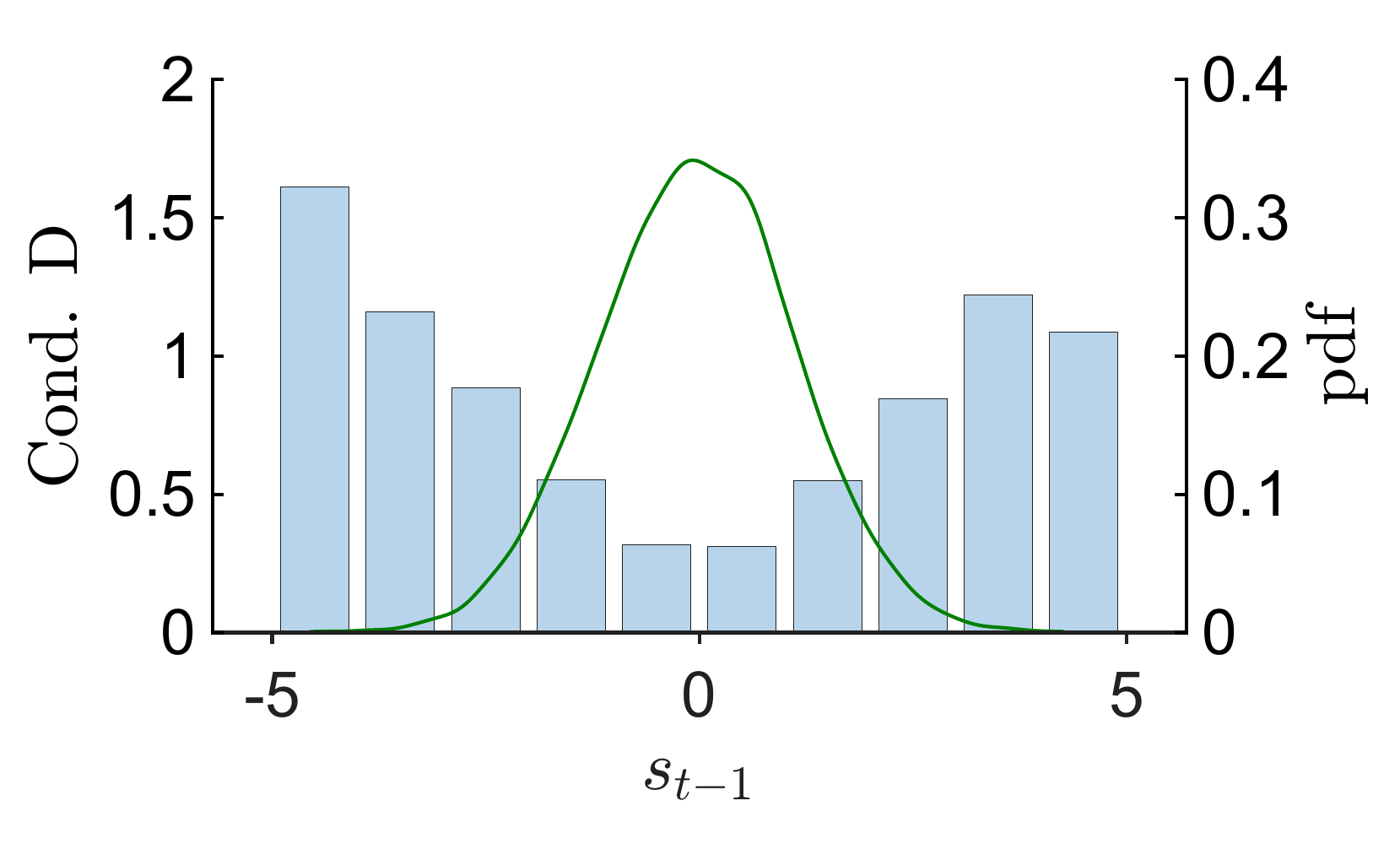}
        & \includegraphics[width=0.42\textwidth]{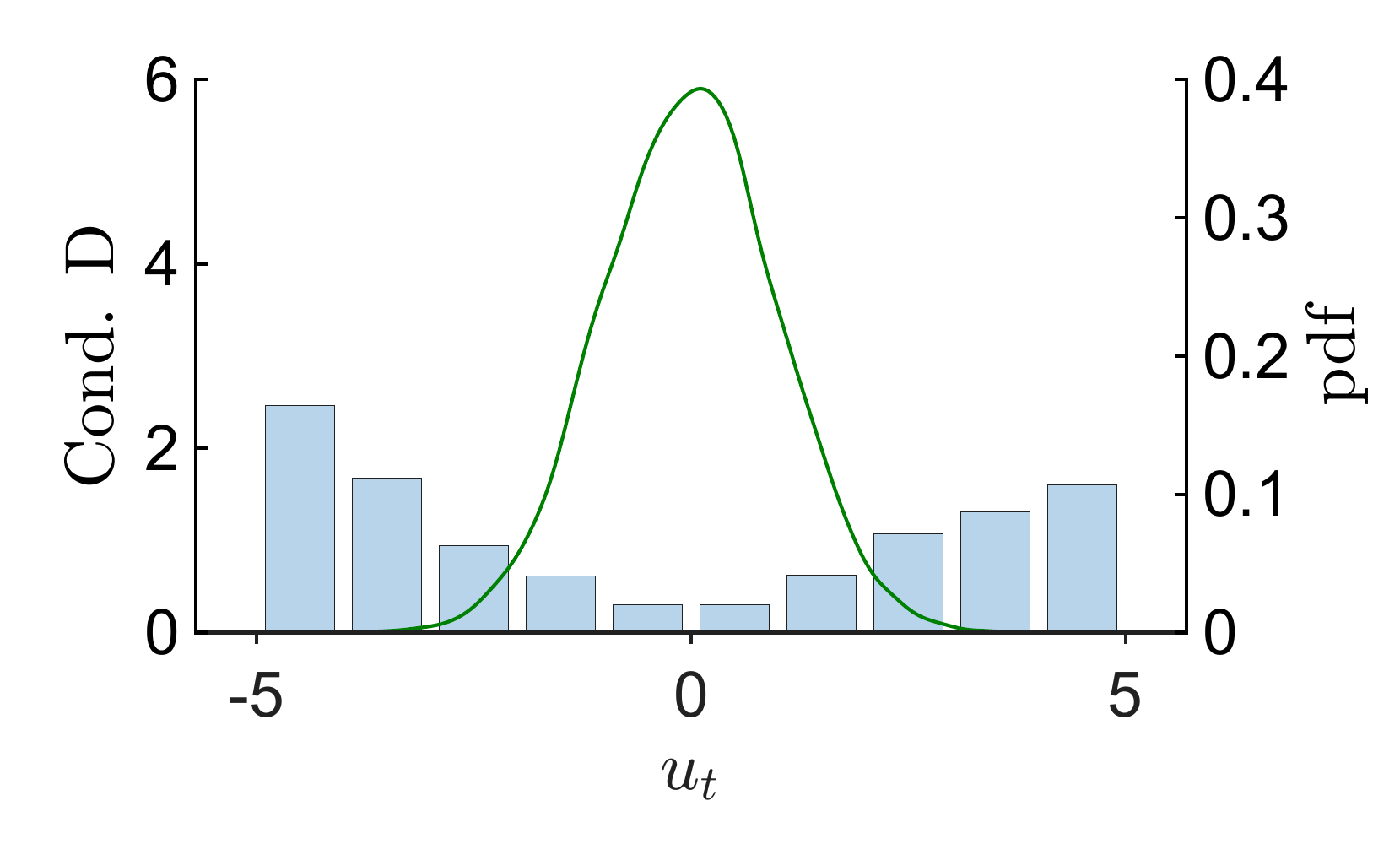} \\[-10pt]
     \rotatebox{90}{\parbox{5 cm}{\centering \textit{LagLP}}}
        & \includegraphics[width=0.42\textwidth]{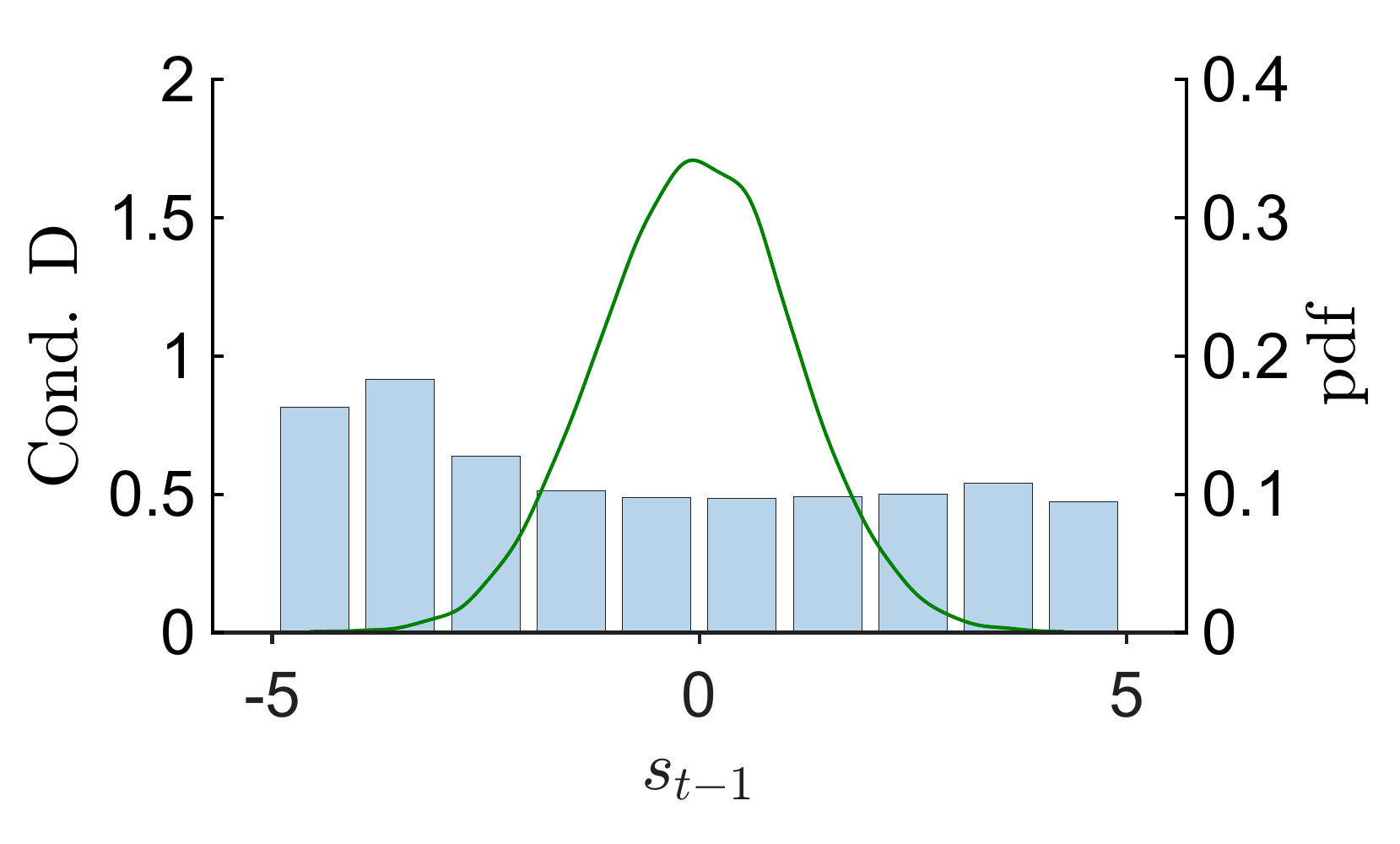}
        & \includegraphics[width=0.42\textwidth]{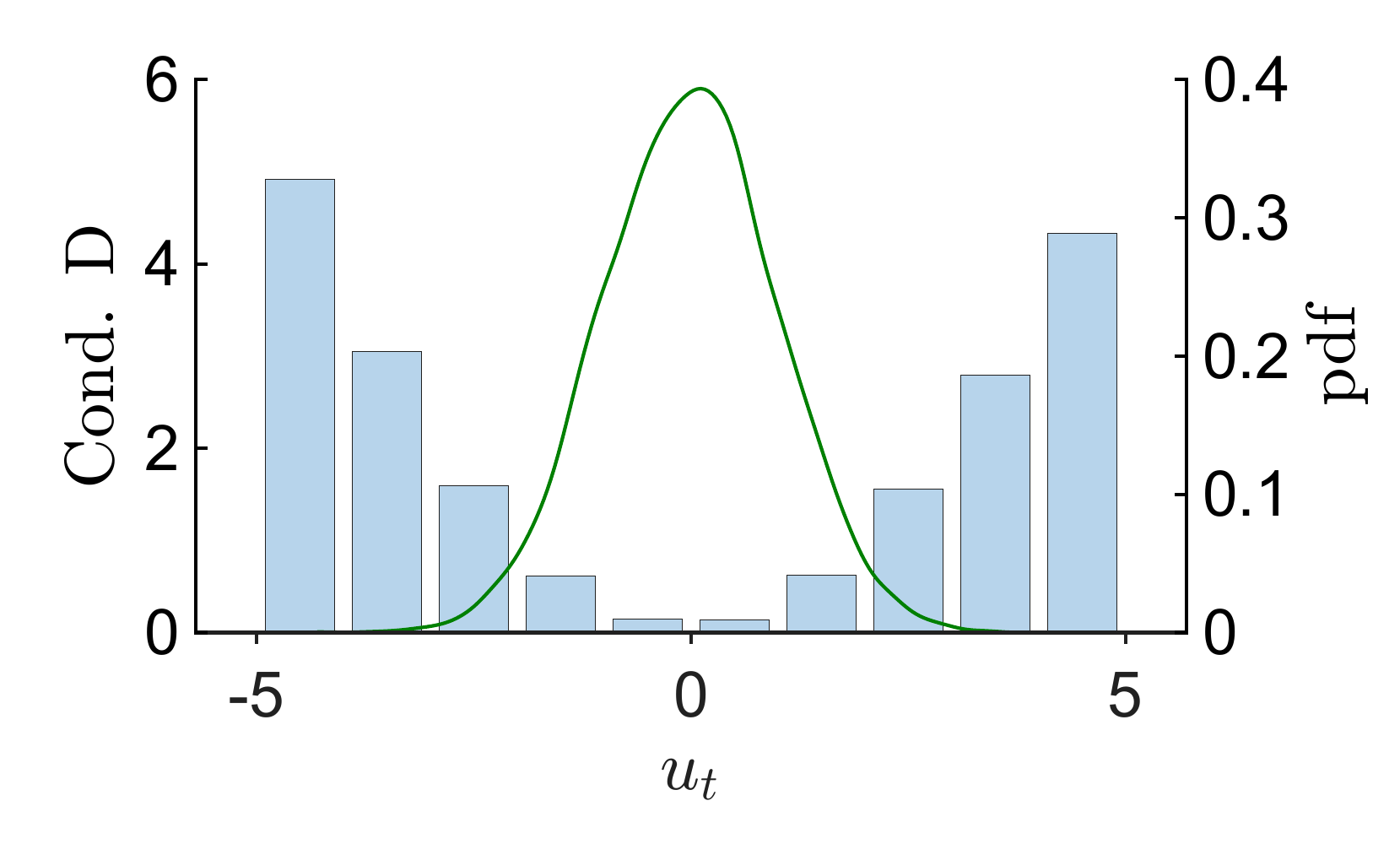} \\[-10pt]
    \rotatebox{90}{\parbox{5 cm}{\centering \textit{Feas}}}
        &  \includegraphics[width=0.42\textwidth]{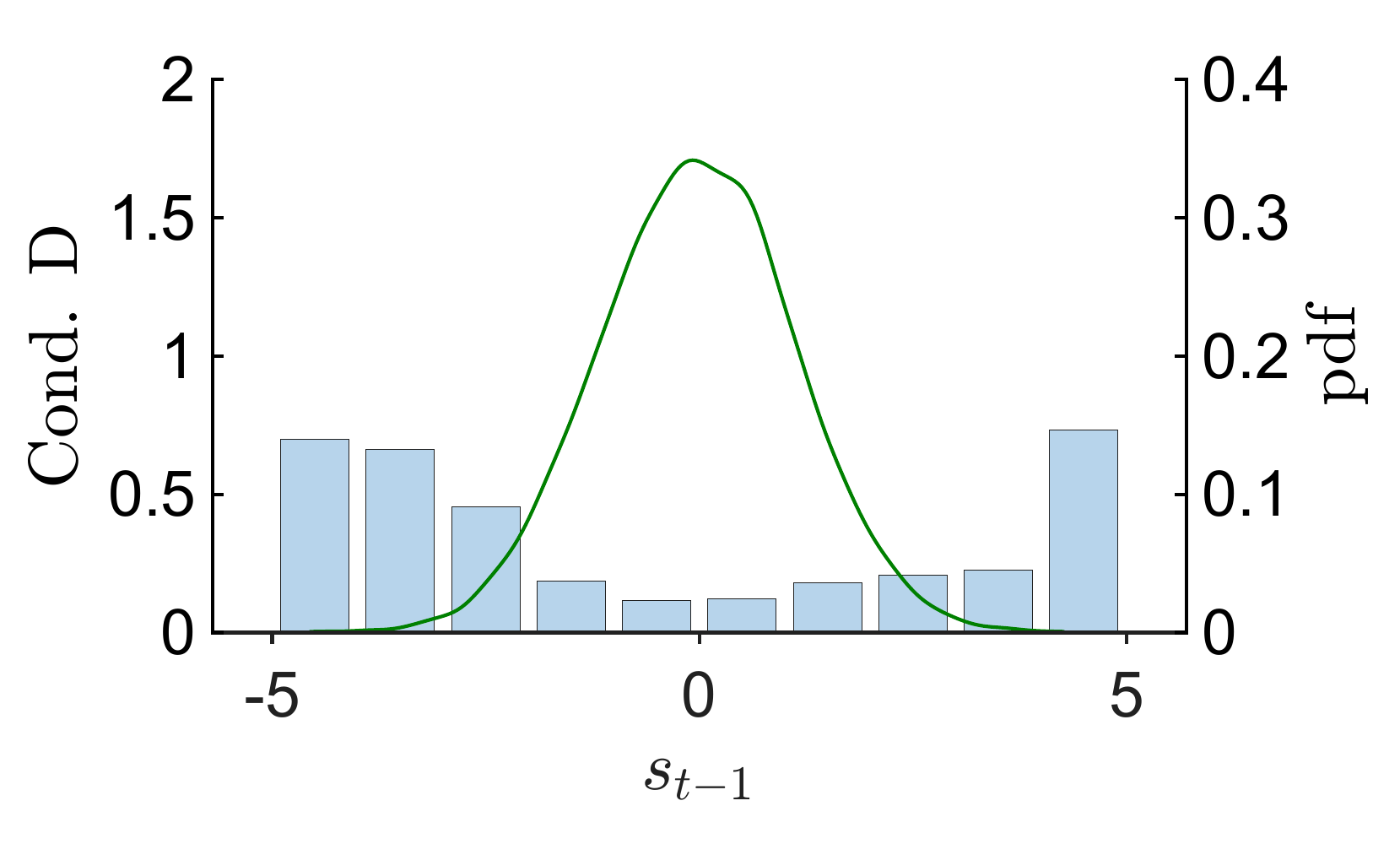} 
        &  \includegraphics[width=0.42\textwidth]{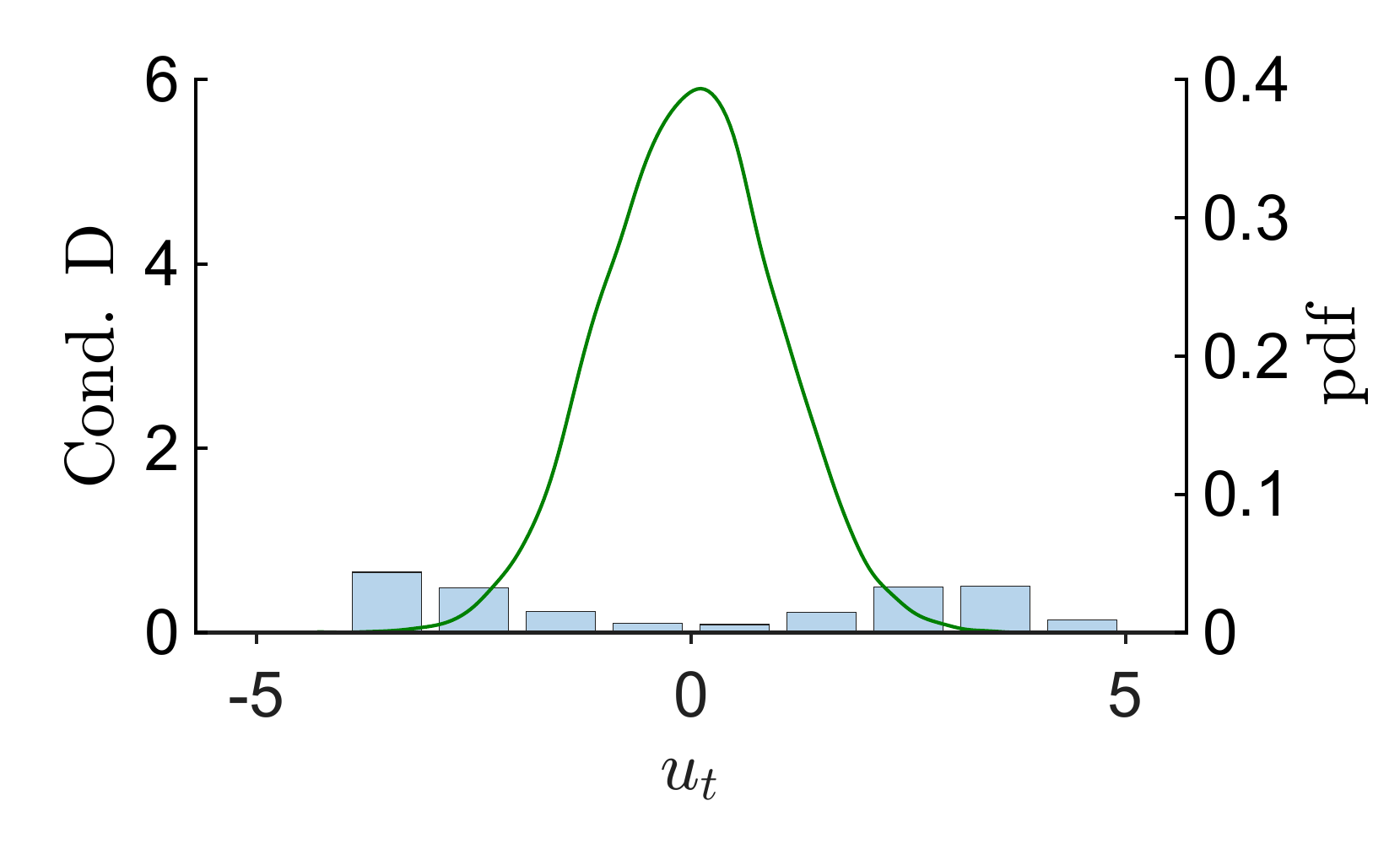} \\
  \end{tabular}
  \end{center}
    {\footnotesize 
    \emph{Notes}: 
    Each row corresponds to one specification.  
    Left column: Blue bars show the conditional distance over $s_{t-1}$ bins; green lines show kernel density estimates for $s_{t-1}$ in the simulated sample.  
    Right column: Blue bars show the conditional distance over $u_t$ bins; green lines show kernel density estimates for $u_t$ in the simulated sample.
    }
  \setlength{\baselineskip}{4mm}
\end{figure}

\clearpage
\subsection{Analytic Results}
\label{subsec:analyytic}
In this section, we establish analytical guarantees for the numerical illustration results within the same QAR(1,1) setting. Specifically, we derive closed-form expressions for the shock-conditional MSE $\mathcal{L}^{spec}_h(\delta)$ and the state-conditional MSE $\mathcal{R}^{spec}_h(s)$ introduced in Remark~\ref{rem:distance_special_cases}(iii)--(iv) (see \eqref{eq:def_loss_u_main} and \eqref{eq:def_loss_s_main}).

The next two theorems summarize the shock- and state-conditioned approximation errors and the induced ranking across specifications.

\begin{theorem}[Conditional MSE comparisons given $u_t=\delta$]
\label{thm:cond_u_main}
Maintain Assumption~\ref{assump:qar.dgp}. Fix $h\ge 0$ and let $\delta\in\RR$.
Let $a_h$ and $q_h$ be as in \eqref{eq:ah_qh_def}, $m$ as in Proposition~\ref{prop:irf.shock.state}, and define $\sigma_s^2 = \Var(s_{t-1})$, $\Cov(s_{t-1},y_{t-1})=\sigma_{sy}$, $\sigma_y^2 = \Var(y_{t-1})$, and 
$
\sigma^2_{s|y} = \sigma_s^2-\sigma_{sy}^2/\sigma_y^2.
$
Then the conditional losses in Definition~\ref{def:cond_mse} satisfy
$$
\mathcal{L}^{Linear}_h(\delta)=a_h^2\delta^2\sigma_s^2 + q_h^2\delta^4,
\quad\quad
\mathcal{L}^{LagLP}_h(\delta)=a_h^2\delta^2\sigma^2_{s|y} + q_h^2\delta^4,
\label{eq:loss_u_lin_lag_main}
$$
$$
\mathcal{L}^{Feas}_h(\delta)=a_h^2\delta^2\sigma^2_{s|y},
\quad\quad
\mathcal{L}^{AsymLP}_h(\delta)=a_h^2\delta^2\sigma_s^2 + q_h^2(\delta^2-m|\delta|)^2.
\label{eq:loss_u_feas_asym_main}
$$
In particular, for every $\delta\in\RR$,
$$
\mathcal{L}^{Feas}_h(\delta)\le \mathcal{L}^{LagLP}_h(\delta)\le \mathcal{L}^{Linear}_h(\delta),
\qquad
\mathcal{L}^{Feas}_h(\delta)\le \mathcal{L}^{AsymLP}_h(\delta),
\label{eq:rank_u_main}
$$
and
$$
\mathcal{L}^{Linear}_h(\delta)-\mathcal{L}^{LagLP}_h(\delta)
=
a_h^2\delta^2\cdot\frac{\sigma_{sy}^2}{\sigma_y^2},
\qquad
\mathcal{L}^{Linear}_h(\delta)-\mathcal{L}^{AsymLP}_h(\delta)
=
q_h^2\big(2m|\delta|^3-m^2\delta^2\big).
\label{eq:gains_u_main}
$$
\end{theorem}

Theorem~\ref{thm:cond_u_main} provides a direct analytical explanation for the shock-bin patterns in Figure~\ref{fig:avg.dist}.
The true QAR response \eqref{eq:car_three_terms} decomposes as
$\operatorname{CAR}_h(s,\delta)=\sigma\phi_1^{h}\delta + a_h\,s\delta + q_h\,\delta^{2}$,
so conditioning on $u_t=\delta$ isolates (i) a state-dependent component $a_h s_{t-1}\delta$ and (ii) a higher-order shock component $q_h\delta^2$.
The theorem's conditional-MSE expressions make clear that the specifications differ in which of these components they can approximate.
\textit{Linear} omits both, and therefore its loss contains the quartic term $q_h^2\delta^4$, which becomes dominant as $|\delta|$ grows.
\textit{LagLP} only alters the state-dependent part---replacing $\sigma_s^2$ with the smaller proxy-based variance $\sigma_{s|y}^2$---but leaves the same $q_h^2\delta^4$ term unchanged.
\textit{AsymLP}, in contrast, leaves the state-dependent part unchanged and targets the higher-order component by approximating $\delta^2$ with $m|\delta|$, so its remaining higher-order error is $q_h^2(\delta^2-m|\delta|)^2$.
Finally, \textit{Feas} combines both ingredients (conditioning on $y_{t-1}$ and including $u_t^2$), which removes the quartic term entirely and yields the smallest shock-conditional MSE for every realized $\delta$.

These expressions map directly into the right column of Figure~\ref{fig:avg.dist}: \textit{Linear} and \textit{LagLP} both have the same leading quartic growth, $q_h^2\delta^4$, so their distances rise rapidly in the tail-shock bins.
The difference between them is of lower order $\delta^2$, arising only from the state-dependent component and therefore small relative to the common quartic term.
By contrast, \textit{AsymLP} improves tail-shock performance through a lower-order correction to the quartic benchmark, and the resulting gain over \textit{Linear} is of order $|\delta|^3$.
Specifically, Theorem~\ref{thm:cond_u_main} shows that $\mathcal{L}^{AsymLP}_h(\delta)\le \mathcal{L}^{Linear}_h(\delta)$ if and only if $|\delta|\ge m/2$, and for $|\delta|>m/2$ the gain
$\mathcal{L}^{Linear}_h(\delta)-\mathcal{L}^{AsymLP}_h(\delta)=q_h^2\,m\,\delta^2(2|\delta|-m)$
is strictly increasing in $|\delta|$ and grows on the order of $|\delta|^3$.

\medskip
\begin{corollary}[Unconditional ranking over the joint distribution of shocks and states]
\label{cor:uncond_rank}
Maintain Assumption~\ref{assump:qar.dgp}. For any finite horizon collection $0,\ldots,H$, the aggregate distance in Definition~\ref{def:cond_mse} satisfies
\[
D^{Feas}(\Omega)\le D^{LagLP}(\Omega)\le D^{Linear}(\Omega),
\qquad
D^{Feas}(\Omega)\le D^{AsymLP}(\Omega).
\]
\end{corollary}

Theorem~\ref{thm:cond_u_main} establishes pointwise (in $\delta$) conditional-MSE dominance of \textit{Feas}; together with the independence $u_t\perp s_{t-1}$ in Assumption~\ref{assump:qar.dgp}, integrating these inequalities over the shock distribution yields the unconditional ranking above.
Equivalently, in the integrated mean-square metric $D^{spec}(\Omega)$, \textit{Feas} is the best-performing specification among \textit{Linear}, \textit{AsymLP}, \textit{LagLP}, and \textit{Feas}.

\medskip
\begin{theorem}[Conditional MSE comparisons given $s_{t-1}=s$]
\label{thm:cond_s_main}
Maintain Assumption~\ref{assump:qar.dgp}. Fix $h\ge 0$ and let $s\in\RR$.
Let $a_h$ and $q_h$ be as in \eqref{eq:ah_qh_def}, and define
\[
\lambda \equiv \frac{\Cov(s_{t-1},y_{t-1})}{\Var(y_{t-1})},
\qquad
\nu_m \equiv \EE\!\left[(u_t^2-m|u_t|)^2\right]
= 3 - 4m\sqrt{2/\pi} + m^2.
\label{eq:nu_m_main}
\]
Let $\mu_y\equiv \EE[y_{t-1}]$ and define the state-conditioned proxy error
\[
\Xi(s)
\equiv
\EE\!\left[
\Big(s_{t-1}-\lambda (y_{t-1}-\mu_y)\Big)^2
\Bigm|\;
s_{t-1}=s
\right].
\]
Then the state-conditioned losses in Definition~\ref{def:cond_mse} satisfy
\[
\mathcal{R}^{Linear}_h(s)=a_h^2 s^2 + 3q_h^2,
\qquad
\mathcal{R}^{LagLP}_h(s)=a_h^2\Xi(s) + 3q_h^2,
\qquad
\mathcal{R}^{Feas}_h(s)=a_h^2\Xi(s),
\label{eq:loss_s_lin_lag_feas_main}
\]
and
\[
\mathcal{R}^{AsymLP}_h(s)=a_h^2 s^2 + \nu_m q_h^2.
\label{eq:loss_s_asym_main}
\]
In particular, for every $s\in\RR$,
\[
\mathcal{R}^{Feas}_h(s)\le \mathcal{R}^{LagLP}_h(s),
\qquad
\mathcal{R}^{AsymLP}_h(s)\le \mathcal{R}^{Linear}_h(s),
\label{eq:rank_s_main}
\]
and the corresponding gaps are
\[
\mathcal{R}^{LagLP}_h(s)-\mathcal{R}^{Feas}_h(s)=3q_h^2,
\qquad
\mathcal{R}^{Linear}_h(s)-\mathcal{R}^{AsymLP}_h(s)=(3-\nu_m)q_h^2.
\label{eq:gains_s_main}
\]
\end{theorem}

Conditioning on the realized state $s_{t-1}=s$ yields the state-side analogue of the shock-conditioned comparison.
For the state-conditional losses $\mathcal{R}^{spec}_h(s)$, the term $a_h^{2}s^{2}$ is the cost of ignoring state dependence (shared by \textit{Linear} and \textit{AsymLP}), whereas the constants in $q_h^{2}$ capture the cost of omitting the quadratic-shock component (shared by \textit{Linear} and \textit{LagLP}).
Thus \textit{AsymLP} only reduces the higher-order part by replacing $3q_h^{2}$ with $\nu_m q_h^{2}$, while \textit{LagLP} can only reduce the state-dependent part by replacing $s^{2}$ with the state-conditioned proxy error $\Xi(s)$.
\textit{Feas} removes the higher-order component altogether and therefore dominates \textit{LagLP} for every state $s$.

\begin{figure}[t!]
  \setlength{\abovecaptionskip}{0cm}
	\caption{Conditional MSE Given $s_{t-1}=s$}
	\label{fig:analytic.loss}  
    \begin{center}
        \includegraphics[width=0.6\linewidth]{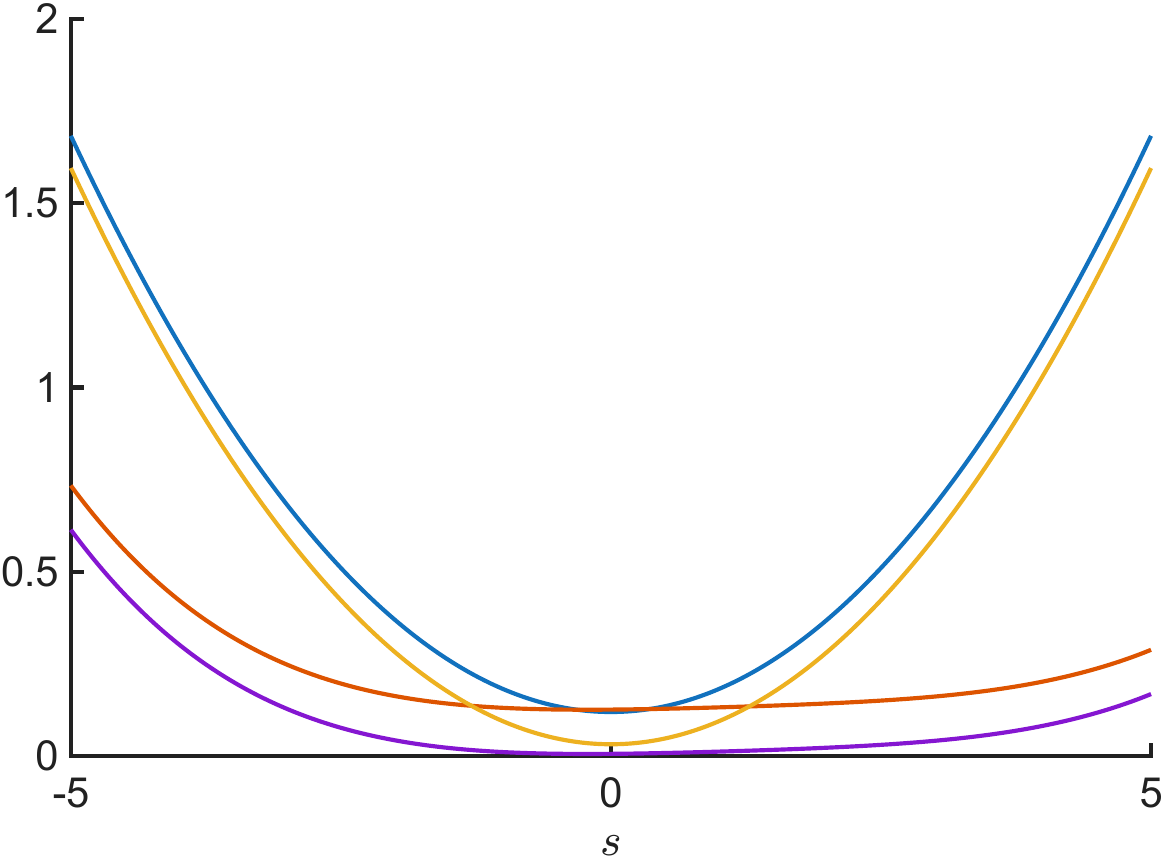}
    \end{center}
	{\footnotesize {\em Notes}: Blue line: \textit Linear; yellow line: \textit AsymLP; orange line: \textit LagLP; purple line: \textit Feas. QAR model parameters are chosen as in Section~\ref{subsec:mc.exp}.}\setlength{\baselineskip}{4mm}
\end{figure}

Theorem~\ref{thm:cond_s_main} also makes precise the caveat behind the lag-based specification:
\[
  \mathcal R^{Linear}_h(s)-\mathcal R^{LagLP}_h(s)
        = a_h^2\big(s^2-\Xi(s)\big).
\]
By definition, $\lambda(y_{t-1}-\EE[y_{t-1}])$ is the best linear prediction for $s_{t-1}$ in mean square, and $\Xi(s)$ is its conditional mean-squared proxy error at state $s$. Since $\EE[s_{t-1}]=0$, the benchmark quantity $s^2$ is the conditional mean-squared error of the unconditional predictor $0$.
Hence $s^2-\Xi(s)$ measures how much conditioning on the proxy $y_{t-1}$ improves the approximation at state $s$; without additional assumptions, it can change sign, so \textit{LagLP} need not dominate \textit{Linear} uniformly across the state space.
In our numerical illustration, $s^2-\Xi(s)$ is negative only in a small neighborhood around $s=0$ (approximately $s\in[-0.32,0.36]$), where proxy noise dominates, and positive for essentially all empirically relevant tail states, which is why \textit{LagLP} gains concentrate in the tails in Figure~\ref{fig:analytic.loss}.

\section{Multivariate Extension: QVAR(1,1)}
\label{sec:multi.extension}

So far, we have focused on the univariate QAR(1,1) model as the DGP. However, this framework is not sufficiently general to capture the rich dynamic interactions among multiple macroeconomic variables. To address this limitation, we consider a multivariate extension, namely the QVAR(1,1) model. This setup offers a more empirically relevant context in which to evaluate LP specifications. We will show that, the main takeaways in the previous sections still hold in the multivariate case.

\medskip
\noindent {\bf Model.\ }  
Let $y_t \in \mathbb{R}^{n}$ be a vector of outcome variables, $s_t \in \mathbb{R}^{n}$ a vector of state variables,
and $\eta_{t} \in \mathbb{R}^{n}$ a vector of reduced-form shocks. The QVAR(1,1) model can be written as
\be
\begin{aligned}
y_t &= \Phi_1 y_{t-1} + \Phi_2 \operatorname{vech}(s_{t-1}s_{t-1}') + (1_n + \mathsf{G} s_{t-1}) \odot \eta_t, \quad \eta_t \overset{\text{i.i.d.}}{\sim} \mathcal{N}(0, \Sigma),\\
s_t &= \Phi_1 s_{t-1} + \eta_t, \quad \rho(\Phi_1)<1,
\end{aligned}
\label{eq:qvar.model}
\ee
where $\vech(A)$ is the operator that stacks the lower-triangular portion of a symmetric matrix $A$ into a vector, and $\odot$ denotes the element-wise product. The coefficient matrix dimensions are
$$
\Phi_1: n \times n; \quad \Phi_2: n \times n(n+1)/2; \quad \mathsf{G}: n \times n; \quad \Sigma: n \times n.
$$

Suppose the researcher is interested in the impulse responses to structural shocks $u_t$, related to the reduced-form shocks by $\eta_t = \Sigma_{tr} \Omega u_t$, where $\Sigma_{tr}$ is the lower-triangular
Cholesky factor of $\Sigma$, $\Omega$ is an orthogonal matrix, and $u_t \overset{\text{i.i.d.}}{\sim} \mathcal{N}(0, I_n).$ We assume $\Omega = I_n$ for simplicity. For notational convenience, write $B\equiv \Sigma_{tr}$ and $b_r\equiv Be_r$ for $r=1,\ldots,n$.

As in the univariate case, we can represent $y_{t+h}$ as a structural function $\tilde{\psi}_h$ of shocks $u_t$ and other variables $U_{h,t+h}=\left(y_{t-1}, s_{t-1}, u_{t+1},\ldots, u_{t+h} \right)$ independent of $u_t.$ We provide the formula of $\tilde{\psi}_h$ in the Online Appendix.

\medskip
\noindent {\bf CAR.\ }
The CAR of the $j$-th variable $y_{jt}$ to a one-time disturbance of magnitude $\delta_i$ in the $i$-th structural shock $u_{it}$ is defined as
$$
\operatorname{CAR}^{(j,i)}_h(\mathcal{F}, \delta_i)
= \mathbb{E}\left[\tilde{\psi}_{jh}(u_{t}+\delta_i e_i, U_{h,t+h})- \tilde{\psi}_{jh}(u_{t}, U_{h,t+h}) \mid \mathcal{F}_{t-1} = \mathcal{F} \right], 
$$
where $\tilde{\psi}_{jh}$ is the $j$-th component of function $\tilde{\psi}_{h}$, and $e_i$ denotes the $i$-th basis vector in $\mathbb{R}^{n}$.

Proposition \ref{prop:qvar.true.car} provides the formula of CAR for the QVAR(1,1) model.
\begin{proposition} \label{prop:qvar.true.car}
The CAR for the QVAR(1,1) model is: for $h\ge 1$,
\[
\begin{aligned}
\operatorname{CAR}^{(j,i)}_h(s, \delta_i)
& = e_j'\Bigg[
\Phi_1^h (1_n +\mathsf{G}\,s)\odot(\delta_i b_i) + \sum_{k=1}^h\Phi_1^{h-k}\,\Phi_2\,\vech\bigl(\delta_i(\Phi_1^k s)(\Phi_1^{k-1} b_i)' + \delta_i(\Phi_1^{k-1} b_i)(\Phi_1^k s)'\bigr) \\
& \qquad\qquad + \sum_{k=1}^h\Phi_1^{h-k}\Phi_2\,\vech\Bigl(\delta_i^2 (\Phi_1^{k-1} b_i)(\Phi_1^{k-1} b_i)'\Bigr)
\Bigg],
\end{aligned}
\]
and for $h=0$,
\[
\operatorname{CAR}^{(j,i)}_0(s, \delta_i)
=  e_j'\Big[(1_n +\mathsf{G}\,s)\odot(\delta_i b_i)\Big],
\]
where $s$ is the realized value of state $s_{t-1}$.
\end{proposition}

As in the univariate case, when $h \geq 1$, the first-order effect depends on the value of the state variable $s_{t-1}$, whereas the second-order effect does not. For $h=0$, the second-order effect is absent.

\medskip
\noindent {\bf Failure of Linear LP.\ }  
As before, we first consider the purely linear LP specification \textit{Linear}:
\be
y_{j,t+h} = \beta_h^{(j,i)} u_{it} + \pi_h^{\prime} W_t + \epsilon_{h,t+h},
\label{multi.linear.spec}
\ee
where $y_{j,t+h}$ is the $j$-th outcome variable, $u_{it}$ is the $i$-th shock, and $W_t$ is a vector of controls that are independent of shocks $u_{t}$. What will the population IRF implied by this specification recover when the true DGP is the QVAR(1,1) model (Assumption \ref{assump:qvar.dgp})?

\begin{assumption} \label{assump:qvar.dgp}
 Assume that $\left\{y_t, s_t, u_t\right\}$ is generated by the QVAR(1,1) model (\ref{eq:qvar.model}), with the process initialized in the infinite past. 
\end{assumption}

The following proposition states that a purely linear LP again fails to capture any nonlinearities in the model.

\begin{proposition} \label{prop:multi.failure.lp}
   Under Assumption~\ref{assump:qvar.dgp}, the population IRF of the $j$-th outcome variable to a shock of magnitude $\delta_i$ to $u_{it}$ implied by Linear is given by $\text{IRF}_{j,i}^{Linear}(\delta_i; h) =  e_j^{\prime} \Phi_1^h \Sigma_{tr} \delta_i e_i,$ where $e_i$ is a $n\times 1$ vector whose $i$-th entry is 1 and all other entries are $0$. This expression coincides with the population IRF under the true DGP being the VAR(1) model, i.e.,
\be
\begin{aligned}
y_t &= \Phi_1 y_{t-1} + \eta_t, \quad \eta_t \sim \mathcal{N}(0, \Sigma), \quad \eta_t = \Sigma_{tr} u_t.
\end{aligned}
\label{eq:var.model}
\ee
\end{proposition}

\medskip
\noindent {\bf Targeting the True IRF. \ } We consider the following infeasible specification \textit{Infeas}:
\be
y_{j,t+h} = \kappa_{h0}^{(j,i)} + \kappa_{h1}^{(j,i)} u_{it} + {\kappa_{h2}^{(j,i)}}^{\prime} s_{t-1} u_{it} + \kappa_{h3}^{(j,i)} u_{it}^2 + \epsilon^{(j,i)}_{h,t+h}.
\label{multi.infeasible.spec}
\ee

Proposition~\ref{prop:irf.multi.infeas} shows that \textit{Infeas} can recover the true impulse response when the true DGP is the QVAR(1,1) model.
\begin{proposition}  \label{prop:irf.multi.infeas}
 Under Assumption~\ref{assump:qvar.dgp}, the population IRF of the $j$-th outcome variable to a shock of magnitude $\delta_i$ to $u_{it}$ implied by \textit{Infeas} is $\operatorname{IRF}_{j,i}^{Infeas}(s, \delta_i; h) = \kappa_{h1}^{(j,i)} \delta_i + {\kappa_{h2}^{(j,i)}}^{\prime} s \delta_i + \kappa_{h3}^{(j,i)} \delta_i^2$, with $\kappa_{h1}^{(j,i)}$, $\kappa_{h2}^{(j,i)}$, and $\kappa_{h3}^{(j,i)}$  such that the IRF exactly recovers the true CAR.
\end{proposition}

In practice, the number of state variables can be very large, so it is empirically infeasible to include all interaction terms $s_{t-1}u_{it}$ (or their proxies) as regressors. Therefore, as suggested by \cite{AndreasenEtAl2018}, a more reasonable option would be conditioning the IRFs on the set $\mathcal{A}$ with a clear economic interpretation. Specifically, define our new causal parameter, conditional CAR, as
$$
\operatorname{cCAR}^{(j,i)}_h(\mathcal{A}, \delta_i) = \mathbb{E}\left[\operatorname{CAR}^{(j,i)}_h(s_{t-1}, \delta_i) \Big|\ s_{t-1}\in \mathcal{A}\right].
\footnote{We interpret $\mathbb{E}[\cdot\mid s_{t-1}\in \mathcal{A}]$ via a regular conditional expectation.
For zero-probability slices $\mathcal{A}=\{s_{t-1} | s_{t-1,I}=c_0\}$ where $s_{t-1,I}$ is a subvector of $s_{t-1}$,
$\mathbb{E}[\cdot\mid s_{t-1}\in \mathcal{A}]=\mathbb{E}[\cdot\mid s_{t-1,I}=c_0]$.} 
$$
For example, suppose the state vector $s_{t-1}$ includes three variables: GDP, the unemployment rate (UR), and the inflation rate (INFL). To study impulse responses in the high-inflation regime, one can define the conditioning set as $\mathcal{A}$ as $\{s_{t-1} = (\text{GDP}_{t-1}, \text{UR}_{t-1}, \text{INFL}_{t-1}) \mid \text{INFL}_{t-1} > 5\%\}$ and focus on cCAR when the state of the economy falls within this regime.

Just as with the CAR, we can design empirical specifications to recover the conditional CAR. Proposition \ref{prop:ccar.multi.infeas} formally states this result.

\begin{proposition} \label{prop:ccar.multi.infeas}
Suppose Assumption~\ref{assump:qvar.dgp} holds.

\noindent (i) If set $\mathcal{A}$ is of positive probability, i.e., $\mathbb{P}(\mathcal{A})>0$.  Consider the following infeasible empirical specification 
\textit{Infeas-Cond1}
\be
y_{j,t+h} = \xi_{h0}^{(j,i)} + \xi_{h1}^{(j,i)} \mathbbm{1}\{s_{t-1}\in \mathcal{A}\} u_{it} + \xi_{h2}^{(j,i)} u_{it}^2+ \epsilon^{(j,i)}_{h,t+h}.
\ee
The conditional CAR can be recovered from the population coefficients of \textit{Infeas-Cond1}:
\begin{equation*}
    \operatorname{cCAR}^{(j,i)}_h(\mathcal{A}, \delta_i) = \xi_{h1}^{(j,i)} \delta_i + \xi_{h2}^{(j,i)} \delta_i^2.
\end{equation*}

\noindent (ii) If $\mathcal{A} = \{s_{t-1} | s_{t-1,I}= c_0\}$ and $\mathbb{P}(\mathcal{A})=0$, where $s_{t-1,I}$ denotes the subvector of $s_{t-1}$ containing the elements indexed by $I \subseteq \{1, \ldots, n\}$. Consider the following (infeasible) empirical specification \textit{Infeas-Cond2}
\be
y_{j,t+h} = \zeta_{h0}^{(j,i)} + \zeta_{h1}^{(j,i)} u_{it} + {\zeta_{h2}^{(j,i)}}^{\prime} s_{t-1,I} \,u_{it} + \zeta_{h3}^{(j,i)} u_{it}^2+ \epsilon^{(j,i)}_{h,t+h}.
\ee
The conditional CAR can be recovered from the population coefficients of \textit{Infeas-Cond2}:
\begin{equation*}
    \operatorname{cCAR}^{(j,i)}_h(c_0, \delta_i) = \zeta_{h1}^{(j,i)} \delta_i + {\zeta_{h2}^{(j,i)}}^{\prime} c_0\,\delta_i + \zeta_{h3}^{(j,i)} \delta_i^2.
\end{equation*}
\end{proposition}

\section{Estimation and Inference}
\label{sec:feas.inference}

The theoretical comparisons in Theorems~\ref{thm:cond_u_main}--\ref{thm:cond_s_main}, together with the numerical evidence in Section~\ref{subsec:mc.exp}, suggest using the feasible specification \textit{Feas} as a practical default.
This section illustrates how to estimate this specification and conduct inference for its implied impulse responses.

Consider the true data-generating process as the QVAR(1,1) model, i.e., Assumption~\ref{assump:qvar.dgp} holds. Fix a horizon $h\ge 0$, a shock index $i\in\{1,\ldots,n\}$, and an outcome index $j\in\{1,\ldots,n\}$.
Let $z_{t-1}\in\mathbb{R}^{k}$ denote a vector of observable state proxies (e.g., a subset of $y_{t-1}$), and let $W_{t-1}$ denote additional controls (typically lags of $y_t$ and $u_t$) that are $\mathcal{F}_{t-1}$-measurable.
For each horizon $h$, we estimate the feasible specification
\be
y_{j,t+h} = x_{it}' \vartheta_h + \epsilon_{h,t+h},
\qquad
x_{it} \equiv \big(1,\,u_{it},\, z_{t-1}^{\prime}u_{it},\,u_{it}^2,\,W_{t-1}^\prime\big)',
\label{eq:feas_reg_vector}
\ee
where $\vartheta_h \equiv (\theta_{h0},\theta_{h1},\theta^\prime_{h2},\theta_{h3},\pi^\prime_h)^\prime$ is the population linear projection coefficient of $y_{j,t+h}$ onto $x_{it}$, and $\epsilon_{h,t+h} \equiv y_{j,t+h}-x_{it}'\vartheta_h$ is the corresponding projection error. 

Throughout this section, we suppress the $(j,i)$ superscript on the coefficients.

Let $T_h\equiv T-h$ and define the OLS estimator and residuals by
\be
\hat\vartheta_h \equiv \arg\min_\vartheta \sum_{t=1}^{T_h} \big(y_{j,t+h} - x_{it}'\vartheta\big)^2,
\qquad
\hat\epsilon_{h,t+h} \equiv y_{j,t+h} - x_{it}'\hat\vartheta_h.
\label{eq:feas_ols_def}
\ee
For a shock of magnitude $\delta_i$ and an evaluation state $z\in\mathbb{R}^{k}$, the implied impulse response is
\be
\widehat{\operatorname{IRF}}^{Feas}_{j,i}(z,\delta_i;h)
\equiv
\hat\theta_{h1}\, \delta_i
+
\hat\theta_{h2}^{\prime}\, z \,\delta_i
+
\hat\theta_{h3}\, \delta_i^2.
\label{eq:feas_irf_hat}
\ee

We impose the following regularity conditions on the state proxies and the controls.

\begin{assumption}[Proxies and Controls]\label{ass:feas.regularity}
There exists an integer $L<\infty$, vectors $a_z$ and $a_W$, and matrices $\{A_{z,\ell},B_{z,\ell},A_{W,\ell},B_{W,\ell}\}_{\ell=1}^L$ such that
\[
z_{t-1}
=
a_z
+
\sum_{\ell=1}^L A_{z,\ell}\,y_{t-\ell}
+
\sum_{\ell=1}^L B_{z,\ell}\,u_{t-\ell},
\]
and
\[
W_{t-1}
=
a_W
+
\sum_{\ell=1}^L A_{W,\ell}\,y_{t-\ell}
+
\sum_{\ell=1}^L B_{W,\ell}\,u_{t-\ell}.
\]
Also the second-moment matrix
$
Q_h\equiv \EE[x_{it}x_{it}']
$
is finite and positive definite.
\end{assumption}

Assumption~\ref{ass:feas.regularity} requires both the state proxy and the control vector to be finite-lag linear filters of observables and identified shocks. This encompasses standard LP lag controls, as well as finite-lag detrending devices such as the real-time Hamilton filter used in the empirical application. This condition ensures that the regression score process has geometric physical dependence and finite moments of all orders.

The following proposition confirms that standard HAC/HAR inference is valid for \textit{Feas} under Assumption~\ref{ass:feas.regularity}.

\begin{proposition}[Asymptotic normality and HAC/HAR inference for \textit{Feas}]\label{prop:hac.feas}
Suppose Assumptions~\ref{assump:qvar.dgp} and \ref{ass:feas.regularity} hold.
Let $\psi_{ht} \equiv x_{it}\,\epsilon_{h,t+h}$. Define
\[
\Gamma_{h,m}\equiv \EE[\psi_{ht}\,\psi_{h,t-m}'],
\qquad
\Omega_h\equiv\sum_{m=-\infty}^{\infty}\Gamma_{h,m},
\qquad
V_h\equiv Q_h^{-1}\Omega_h Q_h^{-1}.
\]
Then, as $T\to\infty$,
\be
\sqrt{T}\big(\hat\vartheta_h-\vartheta_h\big)\xrightarrow{d} \mathcal{N}\!\big(0,\,V_h\big).
\label{eq:asymp_normal_feas}
\ee
As a result, for any fixed $(z,\delta_i)$,
\be
\sqrt{T}\Big\{\widehat{\operatorname{IRF}}^{Feas}_{j,i}(z,\delta_i;h)-\operatorname{IRF}^{Feas}_{j,i}(z,\delta_i;h)\Big\}
\xrightarrow{d}
\mathcal{N}\!\Big(0,\,g(z,\delta_i)^\prime\, V_h\, g(z,\delta_i)\Big),
\label{eq:asymp_normal_irf_feas}
\ee
with gradient $g(z,\delta_i)\equiv \left(0,\,\delta_i,\,(z\delta_i)^\prime,\,\delta_i^2,\,0_{\dim(W_{t-1})}^\prime \right)'$.
Moreover, let $\hat V_h$ be the HAC/HAR estimator based on a bounded symmetric kernel $K$ that is continuous at zero, satisfies $K(0)=1$, and uses a bandwidth sequence $b_T$ with $b_T\to\infty$ and $b_T/\sqrt{T}\to 0$. Then $\hat V_h \to_p V_h.$
\end{proposition}

A natural follow-up question is whether one can use a simpler ``lag-augmentation + EHW'' shortcut: add lagged controls to the LP regression and then report Eicker--Huber--White (EHW) standard errors. 
\citet{MontielOleaEtAl2021} establish that this strategy is valid for correctly specified linear LPs---i.e., when the true DGP is a linear VAR. In that setting, lag augmentation removes the predictable lag component and leaves a score of the form $u_t e_{h,t+h}$, where the residual $e_{h,t+h}$ is free of the current innovation. Since $u_t$ is mean independent of past and future innovations, the score is serially uncorrelated even though the multi-step residual $e_{h,t+h}$ is overlapping.

For the EHW discussion, it is convenient to work with the centered quadratic shock term $c_{it}\equiv u_{it}^2-1$, and define
\[
x_{it}^c\equiv \big(u_{it},\,(z_{t-1}u_{it})',\,c_{it}\big)',
\qquad
\psi_{ht}^c\equiv x_{it}^c \epsilon_{h,t+h}.
\]
As stated in the next proposition, EHW validity is governed by the serial dependence of the slope-score process $\{\psi_{ht}^c\}.$

\begin{proposition}[Criterion for EHW validity]\label{prop:ehw.criterion}
Suppose Assumptions~\ref{assump:qvar.dgp} and \ref{ass:feas.regularity} hold. Define
\[
Q_h^c\equiv \EE[x_{it}^c x_{it}^{c\prime}],
\qquad
\Gamma_{h,\ell}^c\equiv \EE[\psi_{ht}^c\psi_{h,t-\ell}^{c\prime}],
\qquad
\Omega_h^c\equiv \sum_{\ell=-\infty}^{\infty}\Gamma_{h,\ell}^c.
\]
Then the slope block $(\hat\theta_{h1},\hat\theta_{h2}',\hat\theta_{h3})'$ of regression \eqref{eq:feas_reg_vector} satisfies
\[
\sqrt{T}\big((\hat\theta_{h1},\hat\theta_{h2}',\hat\theta_{h3})'-(\theta_{h1},\theta_{h2}',\theta_{h3})'\big)
\xrightarrow{d}
\mathcal N\!\Big(0,\,(Q_h^c)^{-1}\Omega_h^c(Q_h^c)^{-1}\Big).
\]
Consequently, EHW is asymptotically valid for the slope block if and only if
\[
\sum_{\ell\neq0}\Gamma_{h,\ell}^c=0.
\]
\end{proposition}

Proposition~\ref{prop:ehw.criterion} makes clear that the key question is whether lag augmentation renders the slope score $\psi_{ht}^c$ serially uncorrelated. To answer this question, we begin with the zero-proxy-error benchmark $z_{t-1}=s_{t-1}$, which corresponds to the infeasible specification \textit{InFeas} in (\ref{multi.infeasible.spec}). Define
\[
A_{h,t-1}^{(j,i)}
\equiv
\Proj_{L^2(\mathcal F_{t-1})}\!\left[
y_{j,t+h}
-
\kappa_{h1}^{(j,i)}u_{it}
-
\kappa_{h2}^{(j,i)\prime}s_{t-1}u_{it}
-
\kappa_{h3}^{(j,i)}u_{it}^2
\right],
\]
where $\kappa_{h1}^{(j,i)}$, $\kappa_{h2}^{(j,i)}$, and $\kappa_{h3}^{(j,i)}$ are the exact-state coefficients from Proposition~\ref{prop:irf.multi.infeas}. Here $\Proj_{L^2(\mathcal F_{t-1})}$ denotes the orthogonal projection onto the space of square-integrable $\mathcal F_{t-1}$-measurable random variables. The corresponding residual is
\[
e_{h,t+h}^{\star,(j,i)}
\equiv
y_{j,t+h}
-
A_{h,t-1}^{(j,i)}
-
\kappa_{h1}^{(j,i)}u_{it}
-
\kappa_{h2}^{(j,i)\prime}s_{t-1}u_{it}
-
\kappa_{h3}^{(j,i)}u_{it}^2,
\]
and let
\[
x_{it}^{c,\star}\equiv \big(u_{it},\,(s_{t-1}'u_{it})',\,c_{it}\big)'.
\]
Any lag augmentation removes an additive component that is measurable with respect to $\mathcal F_{t-1}$. Residualizing on $L^2(\mathcal F_{t-1})$ therefore provides a maximal benchmark, since it removes all such components at once. Hence, if the score $x_{it}^{c,\star}e_{h,t+h}^{\star,(j,i)}$ remains serially correlated after this residualization, that dependence cannot be eliminated by any lag augmentation, and EHW cannot be restored.

Proposition~\ref{prop:exactstate.pattern} shows that EHW is valid for $h=0$ but generically fails for $h\ge 1$ under the QVAR class. 

\begin{proposition}[Zero proxy error: $h=0$ validity and generic $h\ge 1$ failure]\label{prop:exactstate.pattern}
Maintain Assumptions~\ref{assump:qvar.dgp} and \ref{ass:feas.regularity}. Suppose proxy error is zero so that the interaction term in \eqref{eq:feas_reg_vector} uses the true state, $z_{t-1}=s_{t-1}$. Let
\[
\psi_{ht}^{c,\star}\equiv x_{it}^{c,\star}e_{h,t+h}^{\star,(j,i)},
\qquad
\psi_{ht}^{u,\star}\equiv u_{it}e_{h,t+h}^{\star,(j,i)}.
\]
Write $b_{jr}\equiv e_j'b_r$, and $\mathsf{g}_j'$ as the $j$-th row of $\mathsf{G}$ in (\ref{eq:qvar.model}). Then:
\begin{enumerate}
\item[(i)] At $h=0$, $\{\psi_{0t}^{c,\star}\}$ is a martingale difference sequence. Hence EHW is valid at $h=0$.
\item[(ii)] For each horizon $h\ge 1$ and each $\ell\neq i$, define the horizon-$h$ cross-shock quadratic coefficient
\[
\chi_{h,j,i\ell}
\equiv
\sum_{r=1}^{h} e_j'\Phi_1^{h-r}\Phi_2\,
\vech\!\left((\Phi_1^{r-1}b_i)(\Phi_1^{r-1}b_\ell)' + (\Phi_1^{r-1}b_\ell)(\Phi_1^{r-1}b_i)'\right).
\]
Then, for every $h\ge 1$,
\[
\EE\!\left[\psi_{ht}^{u,\star}\psi_{h,t-h}^{u,\star}\right]
=
(\mathsf{g}_j'\Phi_1^{h-1} b_i)\sum_{\ell\neq i} b_{j\ell}\chi_{h,j,i\ell}.
\]
Consequently, if for a given $h\ge 1$,
\[
(\mathsf{g}_j'\Phi_1^{h-1} b_i)\sum_{\ell\neq i} b_{j\ell}\chi_{h,j,i\ell}\neq 0,
\]
then the lag-$h$ slope-score autocovariance matrix has a nonzero off-zero-lag entry, and EHW fails at that horizon $h$.
\end{enumerate}
\end{proposition}

At $h=0$, the residual contains only contemporaneous shocks $u_{\ell t}$ with $\ell\neq i$, so the score has conditional mean zero and hence forms a martingale difference sequence. For any $h\ge1$, however, the QVAR recursion generates an omitted cross-shock term $\chi_{h,j,i\ell}u_{it}u_{\ell t}$ in the residual. At the same time, the lagged score $\psi_{h,t-h}^{u,\star}$ contains the same $u_{\ell t}$ through the component $u_{i,t-h}e_j'\diag(1_n+\mathsf{G}s_{t-1})Bu_t$. Because $s_{t-1}=\Phi_1^{h-1}Bu_{t-h}+\cdots$ loads on $u_{i,t-h}$ via $\Phi_1^{h-1}b_i$, this component contributes an additional factor of $u_{i,t-h}$. Multiplying by the shock regressor $u_{it}$ therefore yields a term proportional to $u_{it}^2u_{i,t-h}^2u_{\ell t}^2$, which has nonzero expectation. Hence the lag-$h$ score covariance is nonzero whenever the condition in part~(ii) holds.

We now turn to the case with nonzero proxy error. Proposition~\ref{prop:proxyerror.h0} isolates an additional mechanism through which EHW can fail already at $h=0$.

\begin{proposition}[Proxy error creates an $h=0$ failure channel]
\label{prop:proxyerror.h0}
Maintain Assumptions~\ref{assump:qvar.dgp} and \ref{ass:feas.regularity}. Also assume that $z_{t-1}\in L^2$ and $\Var(z_{t-1})$ is nonsingular. Let
\[
\begin{aligned}
&\Lambda \equiv \Cov(s_{t-1},z_{t-1})\Var(z_{t-1})^{-1},\\
&a_\Lambda \equiv \EE[s_{t-1}] - \Lambda\EE[z_{t-1}],\qquad
\xi_{t-1}\equiv s_{t-1}-a_\Lambda-\Lambda z_{t-1}.
\end{aligned}
\]
and consider the feasible regression \eqref{eq:feas_reg_vector}. Let $e_t^{\star,(j,i)}$ denote the population $h=0$ residual obtained after removing $(u_{it}, z_{t-1}u_{it}, u_{it}^2)$ with their population coefficients, and then subtracting the orthogonal projection of the remainder onto $L^2(\mathcal F_{t-1})$. Then
\[
e_{t}^{\star,(j,i)}
=
b_{ji}(\mathsf{g}_j'\xi_{t-1})u_{it}
+
\sum_{\ell\neq i} b_{j\ell}(1+\mathsf{g}_j's_{t-1})u_{\ell t}.
\]
Hence the shock-score component $\psi_{0t}^{u,\star}\equiv u_{it}e_t^{\star,(j,i)}$ satisfies
\[
\EE[\psi_{0t}^{u,\star}\mid \mathcal F_{t-1}] = b_{ji}\mathsf{g}_j'\xi_{t-1}.
\]
If, in addition, $b_{j\ell}=0$ for every $\ell\neq i$, then
\[
\EE\!\left[\psi_{0t}^{u,\star}\psi_{0,t-1}^{u,\star}\right]
=
b_{ji}^2\,
\EE\!\left[(\mathsf{g}_j'\xi_{t-1})(\mathsf{g}_j'\xi_{t-2})u_{i,t-1}^2\right].
\]
Consequently, proxy error makes EHW fail at $h=0$ whenever the right-hand side is nonzero.
\end{proposition}

Even at horizon $0$, proxy error leaves the omitted interaction $b_{ji}(\mathsf{g}_j'\xi_{t-1})u_{it}$ in the residual, so the slope score already has a predictable component.
Under a Cholesky identification, the first row of $B$ has only $b_{11}\neq 0$, so the condition $b_{j\ell}=0$ for all $\ell\neq i$ holds automatically for $(j,i)=(1,1)$. The proposition therefore yields a direct $h=0$ failure channel.

Taken together, Propositions~\ref{prop:exactstate.pattern} and \ref{prop:proxyerror.h0} deliver a sharp message: EHW is not a general inference method for \textit{Feas} under the QVAR class. With zero proxy error, EHW is valid at $h=0$, but fails generically for horizons $h\ge 1$. With proxy error, EHW can fail from $h=0$. HAC/HAR should therefore be the default inference procedure for \textit{Feas}.

\section{Empirical Application}
\label{sec:empiric.exp}
In this section, we reassess the effects of monetary policy shocks following \citet{Ramey2016}. Specifically, we compare three LP specifications: the benchmark linear specification \textit{Linear}, our proposed state-dependent specification \textit{Feas}\footnote{One can view \textit{Feas} as a feasible version of \textit{Infeas-Cond2} in Section \ref{sec:multi.extension}.}, and the nonparametric specification \textit{NPLP} following Remark~\ref{rem:nplp}.

\medskip
\noindent{\textbf{Data, Shocks, and Specifications.}\ } 
For monetary policy shocks, we use the updated monthly \citet{RomerRomer2004} series provided by \citet{WielandYang2020}, which covers the period 1969--2007. We estimate the causal weights $\omega_u(u)$ in (\ref{eq:causal.weight}) using the software provided by KP. Figure~\ref{fig:rr.causal.weight} shows that the estimated weight function is centered near $u=0$ but negatively skewed. This suggests that linear LP may average out some of the nonlinearities, though the asymmetry prevents this averaging from being complete.

\begin{figure}
    \caption{Estimated Causal Weights for R\&R Shock}
    \label{fig:rr.causal.weight}
    \begin{center}
    \includegraphics[width=0.54\linewidth]{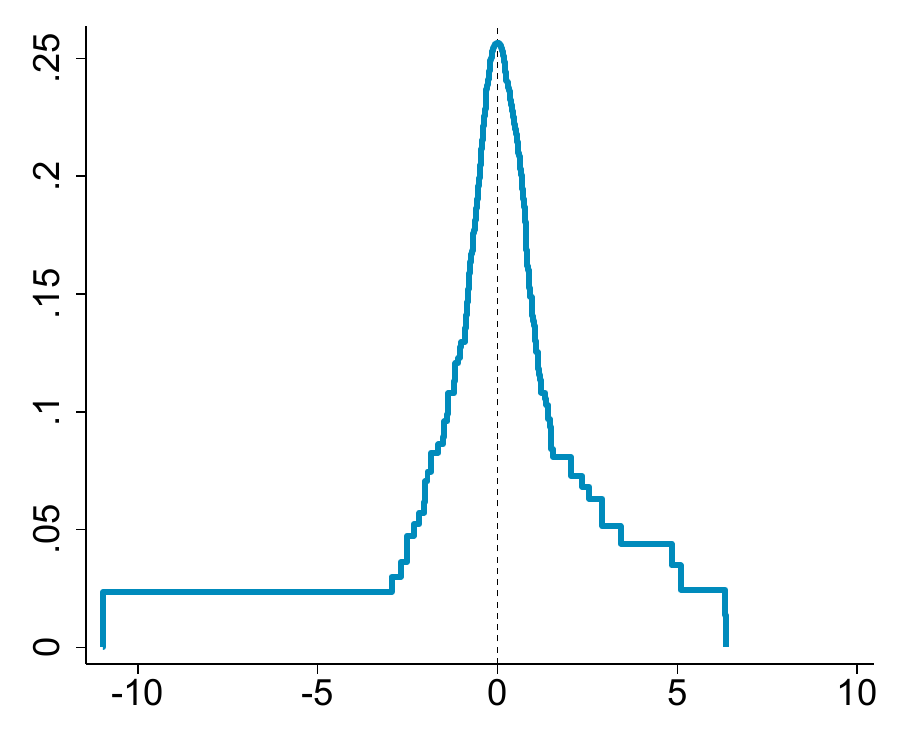}
    \end{center}
    {\footnotesize {\em Notes}: Horizontal axis in units of standard deviations. Total weight $\int_0^{\infty} \omega_u(u) d u$ on positive shocks is 0.47.}
\end{figure}

The \textit{Linear} specification is
\be
y_{t+h} = \beta_{h0} + \beta_{h1} {u}_{t} + \beta_{h2}^{\prime} W_{t} + \epsilon_{h,t+h},
\label{eq:emp.linear.spec}
\ee
the \textit{Feas} specification is
\be
y_{t+h} = \theta_{h0} + \theta_{h1} {u}_{t} + \theta_{h2}^{\prime}\, z_{t-1} \,u_{t} + \theta_{h3}\, u_{t}^2 + \theta_{h4}^{\prime} W_{t} + \epsilon_{h,t+h}.
\label{eq:emp.feas.spec}
\ee
and we consider a control-adjusted nonparametric specification \textit{NPLP}
\be
\EE[y_{t+h}\mid z_{t-1},u_t,W_t]=m_h(z_{t-1}, u_t)+\vartheta_h'W_t,
\label{eq:emp.nplp.spec}
\ee
where $m_h(\cdot,\cdot)$ is an unknown nonparametric function and the controls enter linearly.
Here, $y_{t+h}$ is the outcome variable of interest and $u_t$ is the R\&R shock. The control vector $W_t$ includes monthly lagged shocks and lagged aggregate variables (federal funds rate, industrial production, unemployment rate, CPI, and the commodity price index). Finally, $z_{t-1}$ denotes a vector of two state variables; their construction is described later. The implied impulse responses to a shock of size $\delta$ are
\be
\begin{aligned}
&\text{IRF}^{\textit{Linear}}(\delta; h) = \beta_{h1} \delta, \\
&\text{IRF}^{\textit{Feas}}(z, \delta; h) = \theta_{h1} \delta + \theta_{h2}^{\prime}\, z \delta + \theta_{h3}\, \delta^2, \\
&\text{IRF}^{\textit{NPLP}}(z, \delta; h) = \mathbb{E}\left[m_h(z,u_t+\delta)-m_h(z,u_t)\right].
\end{aligned}
\label{eq:emp.IRF}
\ee
We compute the IRF estimates and the associated confidence bands for \textit{Feas} as described in Section~\ref{sec:feas.inference}. For \textit{NPLP}, the estimation procedure, including the blocked cross-validation bandwidth choice, is described in the Online Appendix. Since we do not report confidence bands for \textit{NPLP}, we use it primarily as a qualitative benchmark.

\medskip
\noindent{\textbf{Choice of State Variables.}\ }
To capture state dependence in the impulse responses, we set $z_{t-1}$ to the cyclical components of industrial production and CPI. Rather than using the raw series, we use detrended series to better capture the cyclical behavior of the economic states. To avoid look-ahead bias, we extract the cyclical components from the log series using a real-time Hamilton filter based on \citet{Hamilton2018}, with monthly lead length $h=24$ and lag length $p=12$, estimated recursively so the lagged state vector is constructed using only information available through $t-1$. The construction details are provided in the Online Appendix. Table~\ref{tab:nber.peak} reports the values of state variables at three major NBER peak--trough episodes in our sample: 1973--1975, 1981--1982, and 2001. Both industrial production and CPI are uniformly higher at peaks than at troughs across the three episodes.

\begin{table}[t!]
  \setlength{\abovecaptionskip}{0cm}
  \caption{NBER Peaks and Troughs}
  \begin{center}
    \begin{tabular}{lcccccc}
    \hline \hline
    No. & Peak  & IP   & CPI  & Trough & IP   & CPI \\
    \hline
    1 & 1973 NOV & 0.096  & 0.014  & 1975 MAR & -0.145  & 0.007  \\
    2 & 1981 JUL & -0.022  & -0.023  & 1982 NOV & -0.110  & -0.058  \\
    3 & 2001 MAR & -0.001  & 0.028  & 2001 NOV & -0.067  & -0.008  \\
    \hline
    \end{tabular}%
   \end{center}
  \label{tab:nber.peak}
  {\footnotesize {\em Notes}: Cyclical components of log industrial production and CPI at NBER peak and trough months. Series are detrended with a real-time Hamilton filter as described in the text.}
  \setlength{\baselineskip}{4mm}
\end{table}%

\medskip
{\noindent\textbf{Results.}\ }  Figure~\ref{fig:irf.state.depend} compares average responses across the three NBER peaks and the three NBER troughs listed in Table~\ref{tab:nber.peak}. For $G\in\{\text{Peak},\text{Trough}\}$, let $z_{G,1},z_{G,2},z_{G,3}$ denote the corresponding benchmark states and let $\bar z_G=\frac{1}{3}\sum_{m=1}^3 z_{G,m}$. The object plotted for \textit{Feas} is
\[
\widehat{\text{IRF}}^{\textit{Feas}}_{G}(\sigma_{MP};h)
=
\frac{1}{3}\sum_{m=1}^{3}\widehat{\text{IRF}}^{\textit{Feas}}(z_{G,m},\sigma_{MP};h)
=
\hat\theta_{h1}\,\sigma_{MP}
+\hat\theta_{h2}^{\prime}\,\bar z_G\,\sigma_{MP}
+\hat\theta_{h3}\,\sigma_{MP}^2,
\]
where $\sigma_{MP}$ is one standard deviation of the R\&R shock (29.7 basis points); a positive shock is contractionary. We overlay the results for \textit{Linear}, \textit{Feas}, and \textit{NPLP}. This exercise is designed to assess the state dependence of the impulse responses---if the true DGP features such dependence, \textit{Feas} and \textit{NPLP} should yield different IRFs across states, whereas \textit{Linear} cannot capture this heterogeneity.

The \textit{Linear} estimates (black dashed line) yield three main findings. First, a contractionary shock leads to a persistent decline in industrial production and a rise in unemployment, both of which gradually revert to their pre-shock levels. Quantitatively, industrial production turns negative after several months and reaches about $-0.62\%$ at a horizon of 26 months, while unemployment peaks at about $0.15\%$ at a horizon of 28 months. Second, CPI increases modestly over short-to-medium horizons, with the point estimate remaining positive through about 26 months, indicating the presence of a price puzzle. Third, a contractionary shock raises the federal funds rate on impact, with the response peaking at about $0.67\%$ after two months, before gradually declining and turning modestly negative at longer horizons.

The \textit{Feas} estimates (red solid line) reveal clear state dependence. Our main finding is that contractionary monetary policy has stronger real effects in troughs than in peaks. Industrial production falls by about $1.46\%$ at a horizon of 26 months in the trough-average state, compared with about $0.45\%$ in the peak-average state, while unemployment rises to about $0.28\%$ in troughs versus about $0.19\%$ in peaks. Consistent with this pattern, the short-run federal funds rate response is substantially stronger in troughs, peaking at about $0.99\%$ at the two-month horizon, compared with about $0.54\%$ in peaks. CPI also shows state dependence at long horizons: by 60 months the response is about $-0.79\%$ in troughs versus about $-0.42\%$ in peaks.

The \textit{NPLP} estimates (green solid line) generate a state-dependence pattern similar to \textit{Feas} in our sample. In particular, trough responses are generally stronger than peak responses for industrial production and unemployment, CPI shows the same long-horizon ordering, and the federal funds rate response is stronger in troughs in the short run. With the cross-validation-tuned bandwidth, \textit{NPLP} yields responses that are close in shape to those of \textit{Feas} and not excessively jagged. We view this as evidence that the main state-dependence patterns are not artifacts of the parametric structure imposed by \textit{Feas}.

\begin{figure}[t!]
  \setlength{\abovecaptionskip}{0cm}
  \caption{State Dependence of Impulse Responses}
  \label{fig:irf.state.depend}
  \begin{center}
  \renewcommand{\arraystretch}{1} 
  \begin{tabular}{ccc}
      & Peak & Trough \\[-10pt]
      \rotatebox{90}{\parbox{5 cm}{\centering \textit{IP}}}
        & \includegraphics[width=0.38\textwidth]{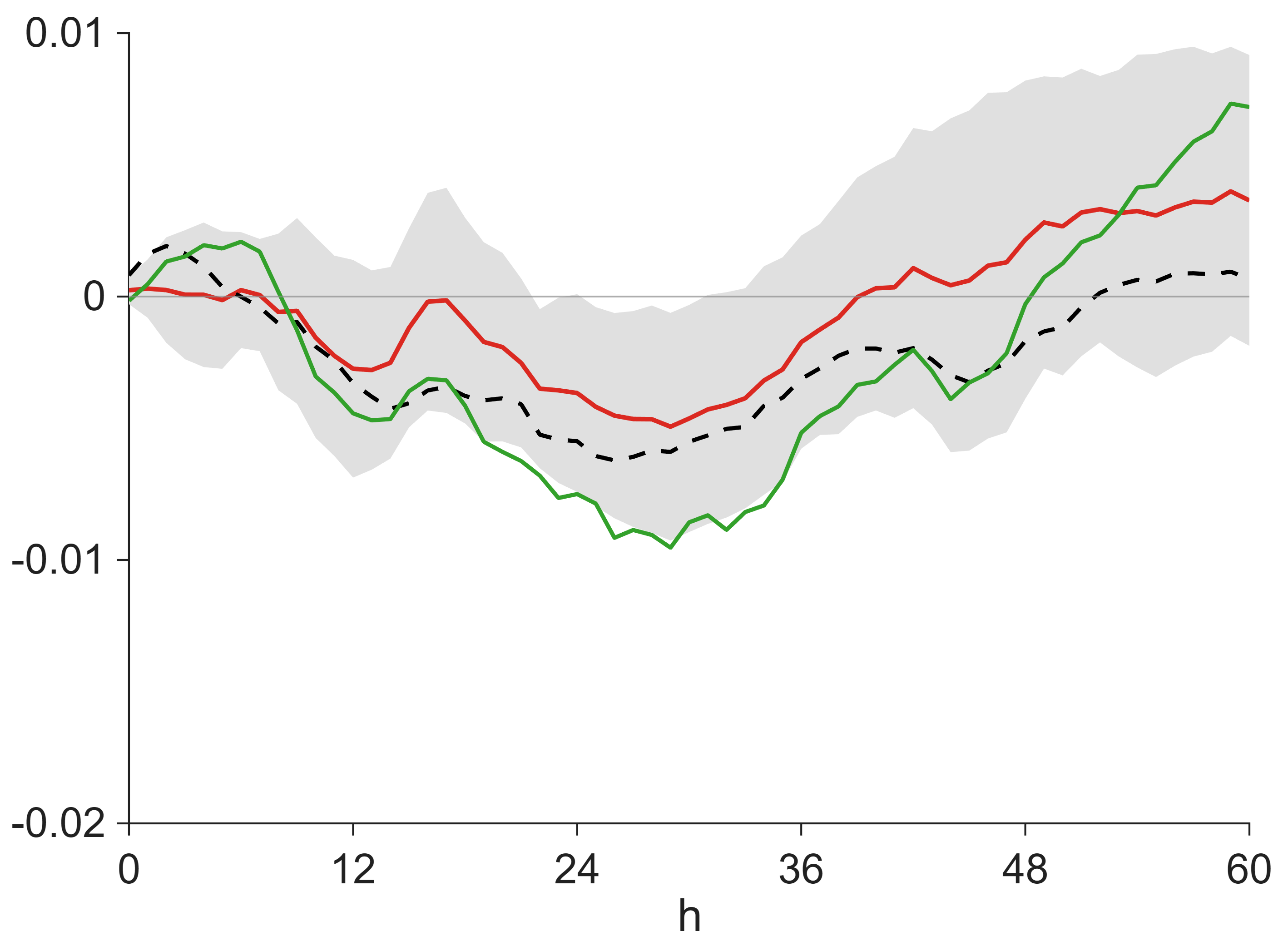}
        & \includegraphics[width=0.38\textwidth]{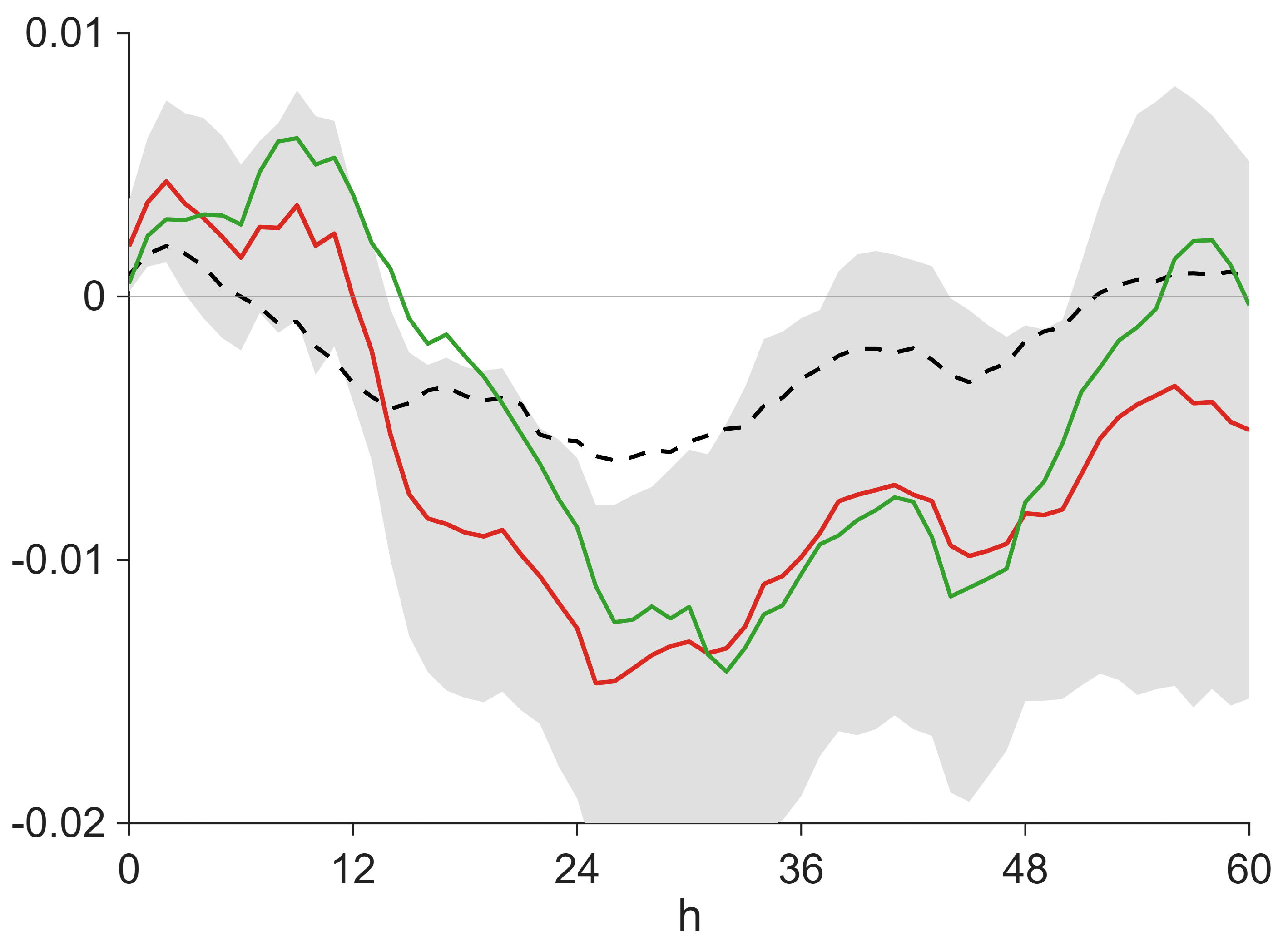} \\[-10pt]
     \rotatebox{90}{\parbox{5 cm}{\centering \textit{Urate}}}
        & \includegraphics[width=0.38\textwidth]{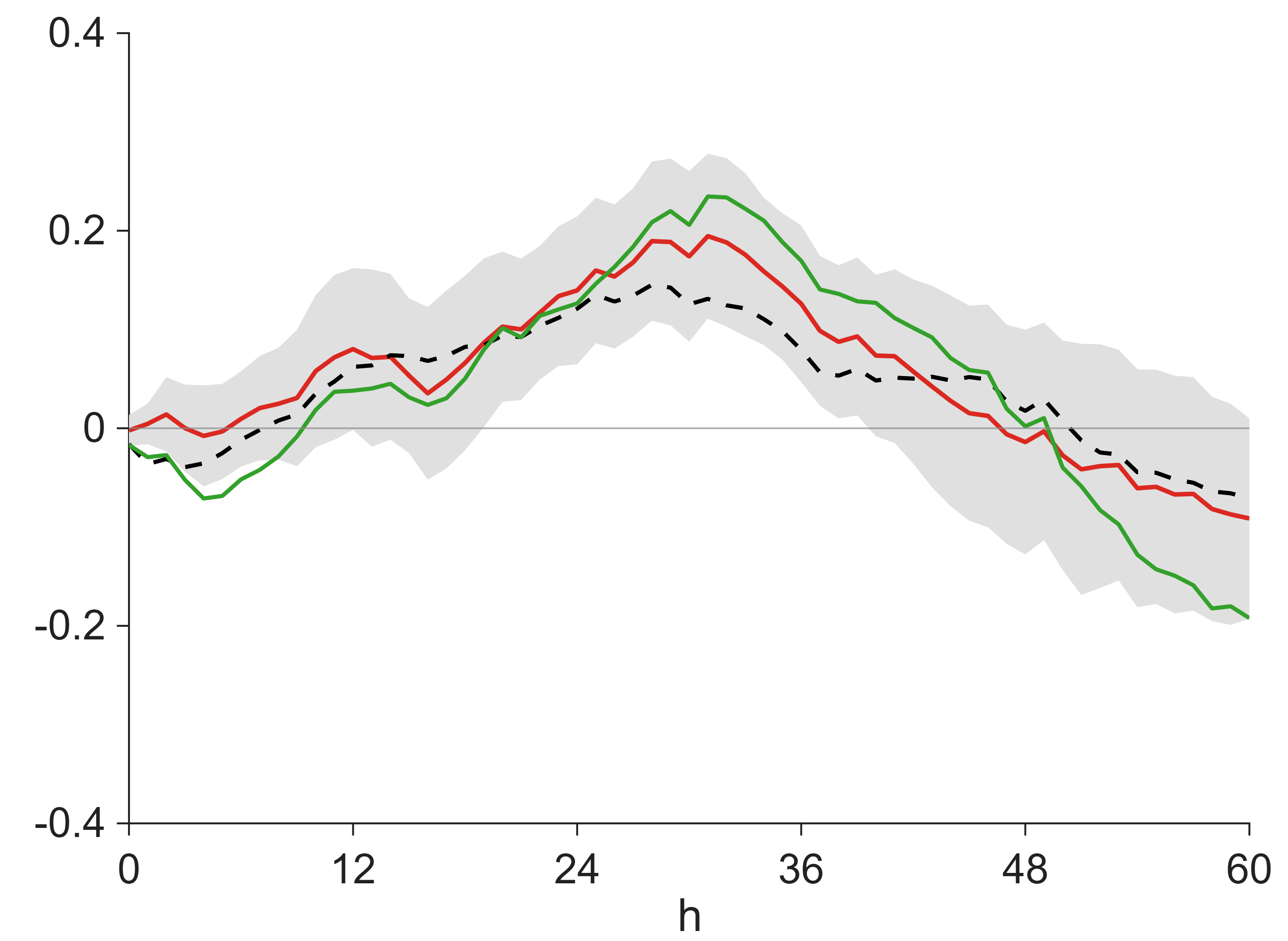}
        & \includegraphics[width=0.38\textwidth]{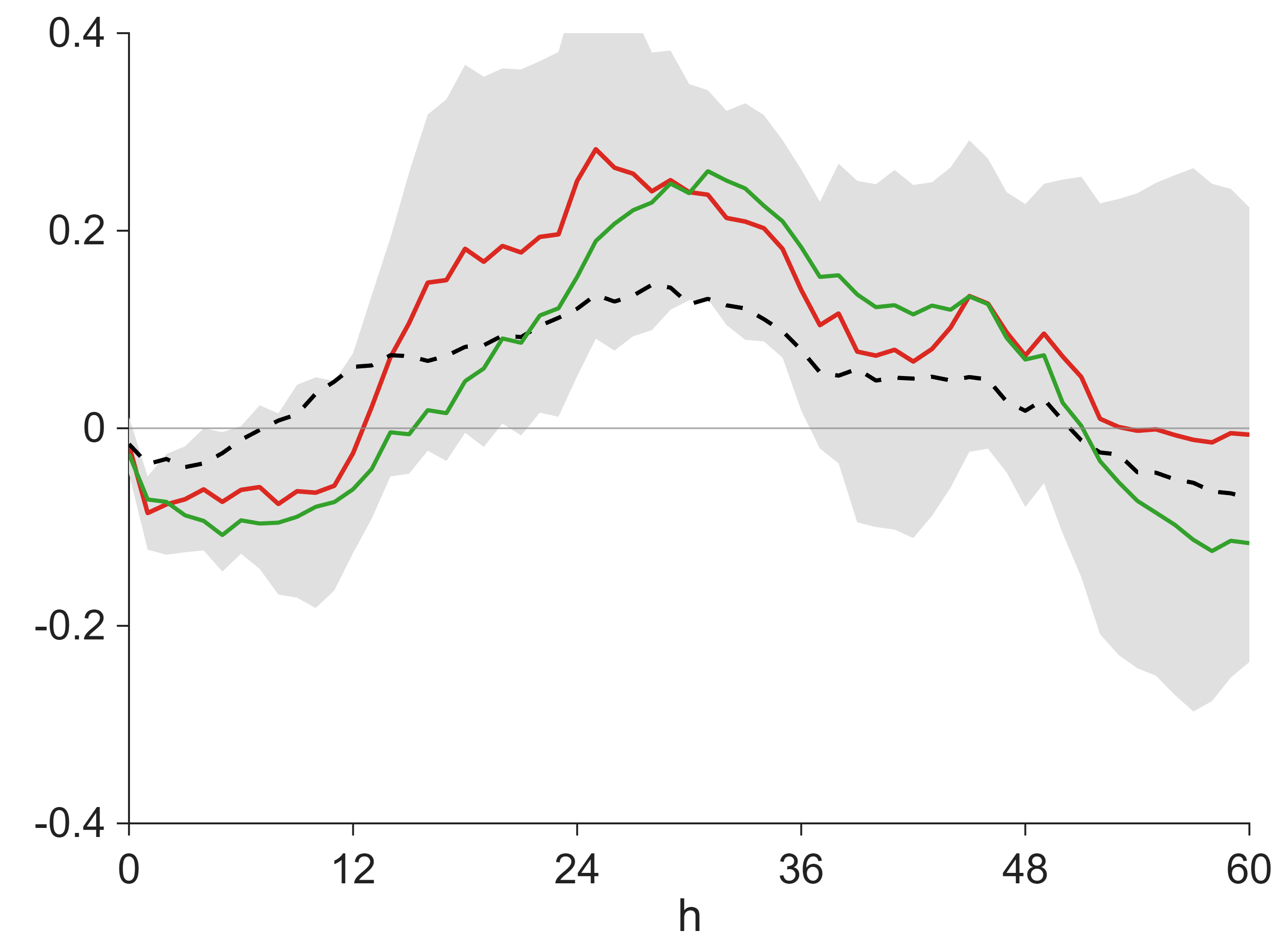} \\[-10pt]
     \rotatebox{90}{\parbox{5 cm}{\centering \textit{CPI}}}
        & \includegraphics[width=0.38\textwidth]{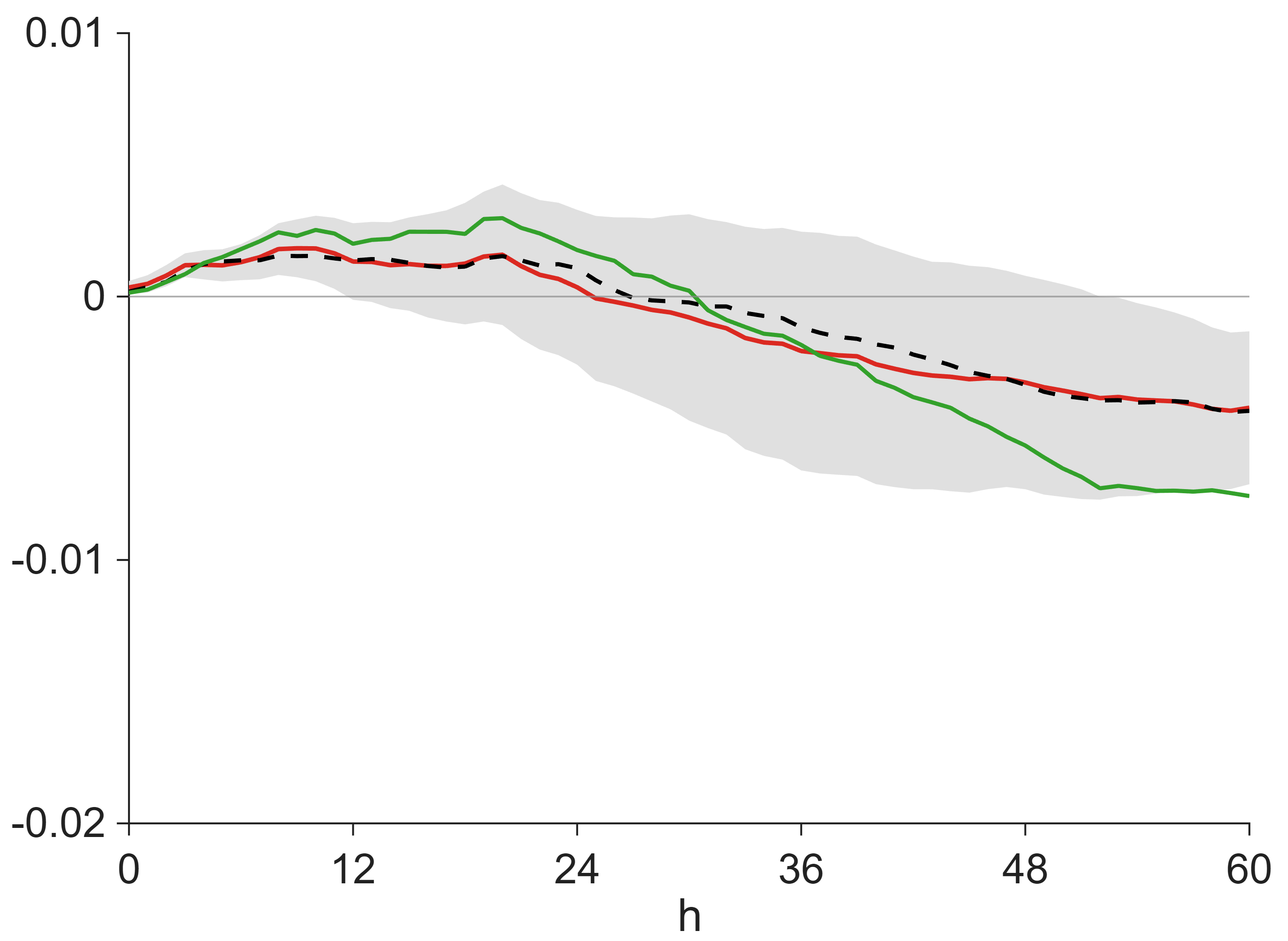}
        & \includegraphics[width=0.38\textwidth]{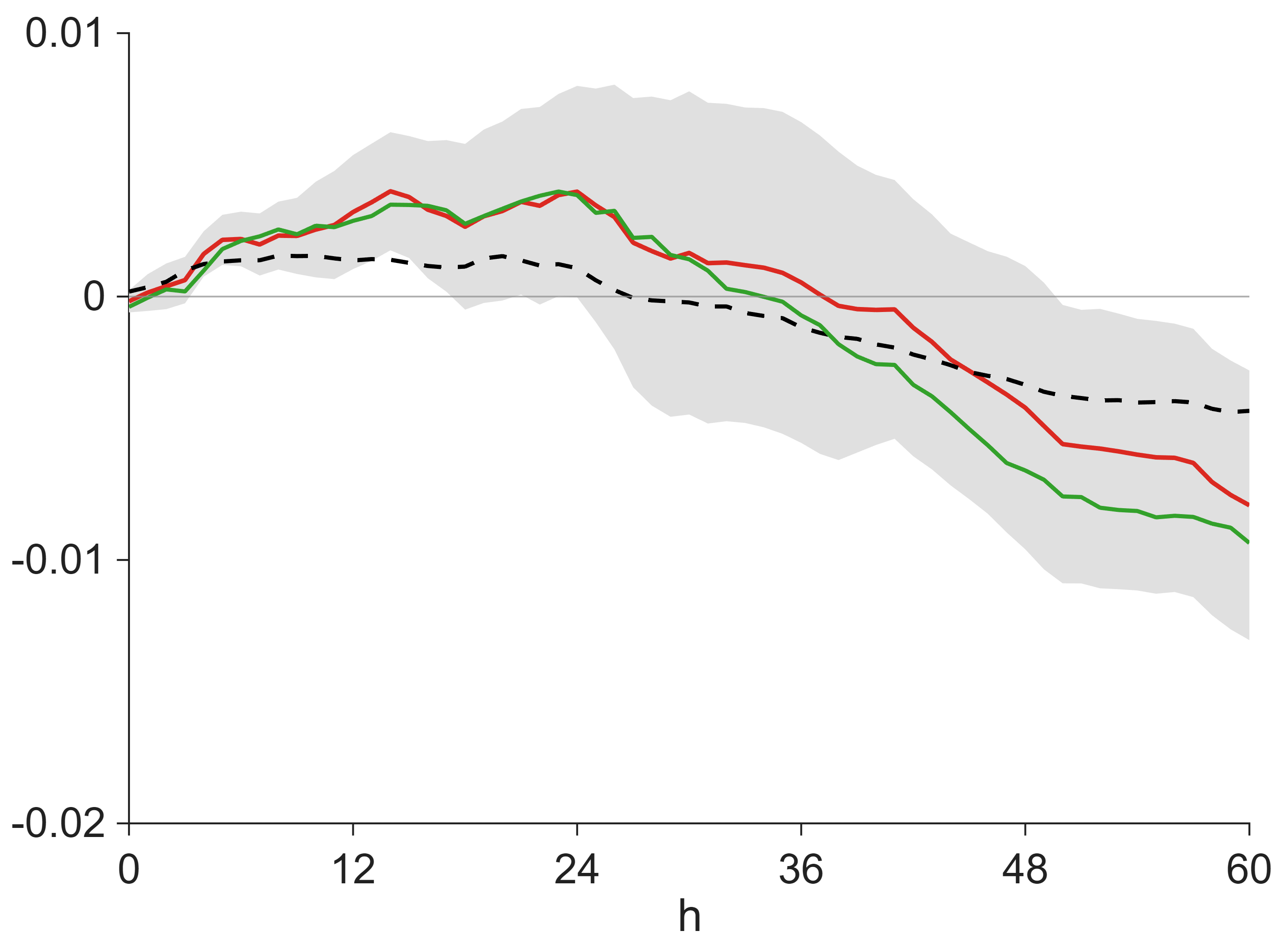} \\[-10pt]
    \rotatebox{90}{\parbox{5 cm}{\centering \textit{FFR}}}
        &  \includegraphics[width=0.38\textwidth]{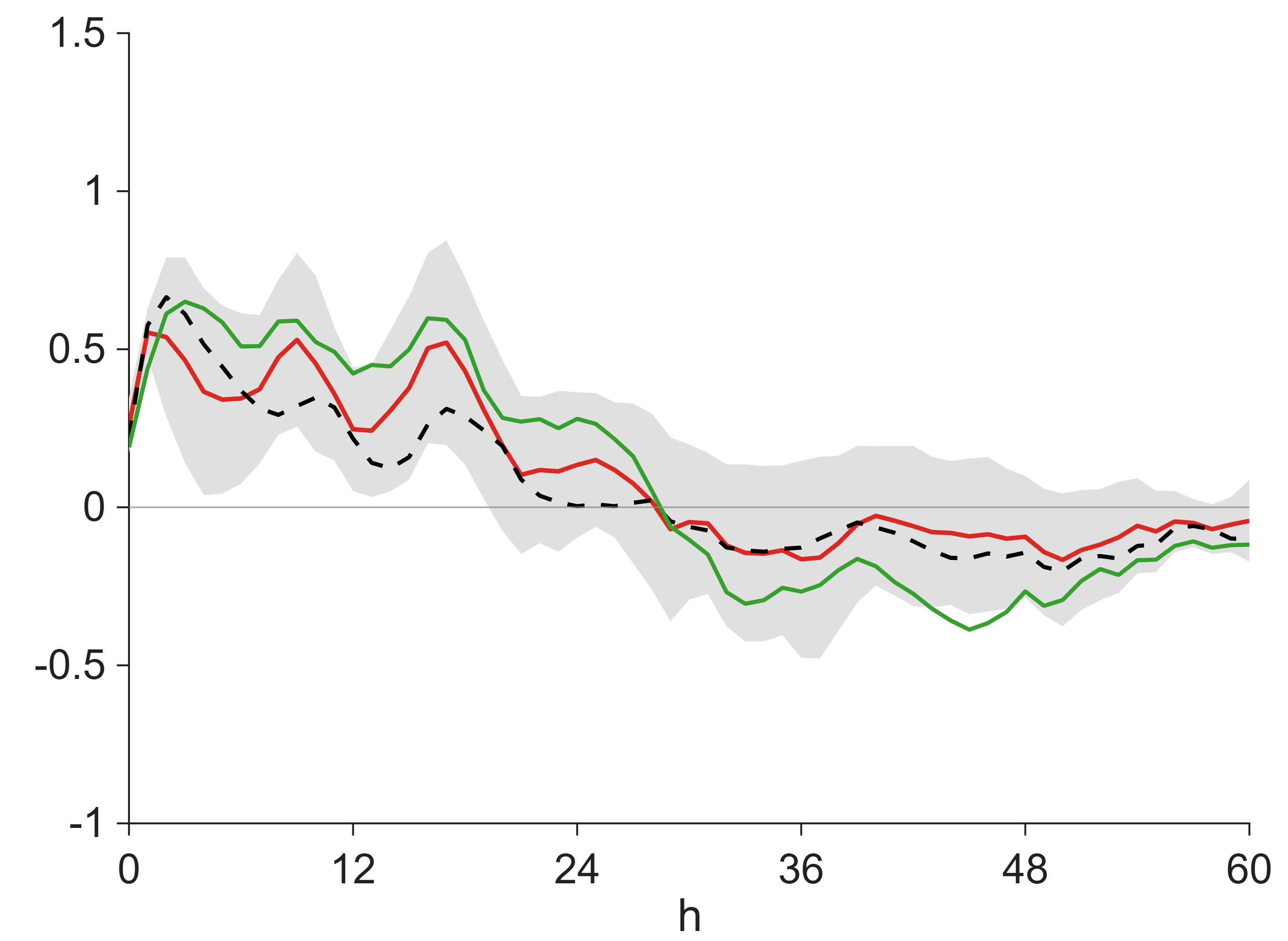} 
        &  \includegraphics[width=0.38\textwidth]{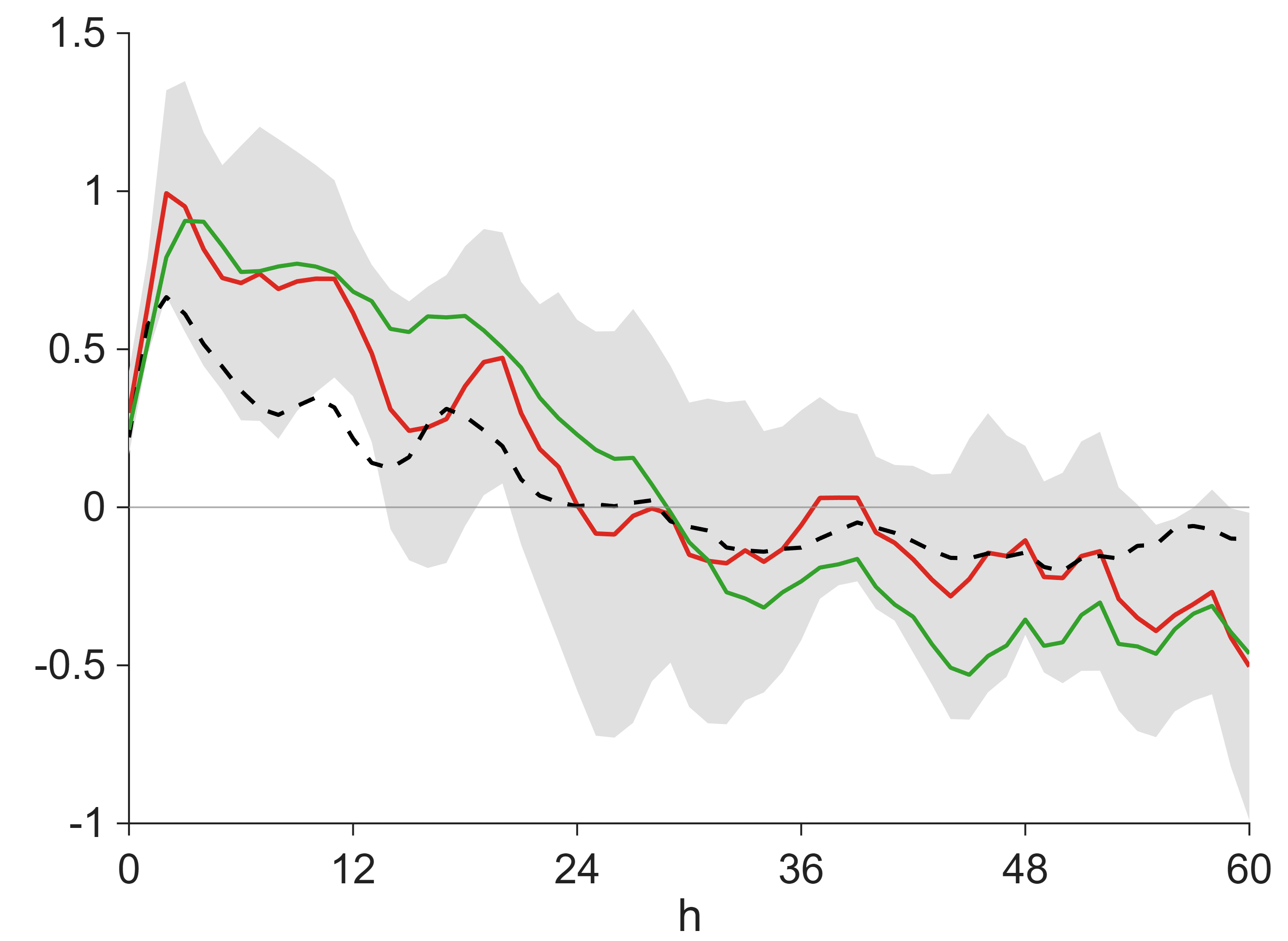} \\
  \end{tabular}
  \end{center}
    {\footnotesize 
    \emph{Notes}: 
    Average peak and trough responses of different outcome variables to a standard deviation contractionary monetary policy shock under different specifications. \textit{IP}, \textit{Urate}, \textit{CPI}, and \textit{FFR} represent industrial production, unemployment rate, CPI, and fed funds rate, respectively. The red solid line denotes \textit{Feas}, the black dashed line denotes \textit{Linear}, and the green solid line denotes \textit{NPLP}. The shaded area represents the 90\% confidence band for \textit{Feas}.
    }
  \setlength{\baselineskip}{4mm}
\end{figure}

To assess higher-order effects, we compare the scaled impulse response 
$$\frac{\text{IRF}^{\textit{Feas}}(z, k\,\sigma_{MP}; h)}{k}$$ for different values of $k$.
Substituting (\ref{eq:emp.IRF}), we can further calculate the scaled responses as
$\theta_{h1}\sigma_{MP} + \theta_{h2}^{\prime}\, z \sigma_{MP} + k\,\theta_{h3}\,\sigma_{MP}^2.$  If the IRF is linear in the shock size, i.e., $\theta_{h3}=0$, these scaled responses should coincide across $k$. Otherwise, the scaled impulse responses should differ, and the larger the higher-order effects, the greater the divergence across $k$. Figure~\ref{fig:irf.high.order} plots the scaled impulse response estimates $$\frac{\widehat{\text{IRF}}^{\textit{Feas}}(z_{ss}, k\,\sigma_{MP}; h)}{k}, \quad k\in\{-1,1,2\},$$ where the state variables are fixed at their steady-state values $z_{ss}=(0,0).$ Differences across $k$ are comparatively small for industrial production and CPI, but more pronounced for unemployment at medium horizons and for the federal funds rate at short-to-medium horizons. This pattern is consistent with the estimated quadratic coefficients (see the Online Appendix).

\begin{figure}[t!]
  \setlength{\abovecaptionskip}{0cm}
	\caption{Higher-Order Effects of Impulse Responses}
	\label{fig:irf.high.order}  
	\begin{center} 
		\begin{tabular}{cc}
			\textit{IP} & \textit{Urate} \\
			\includegraphics[width=0.45\textwidth]{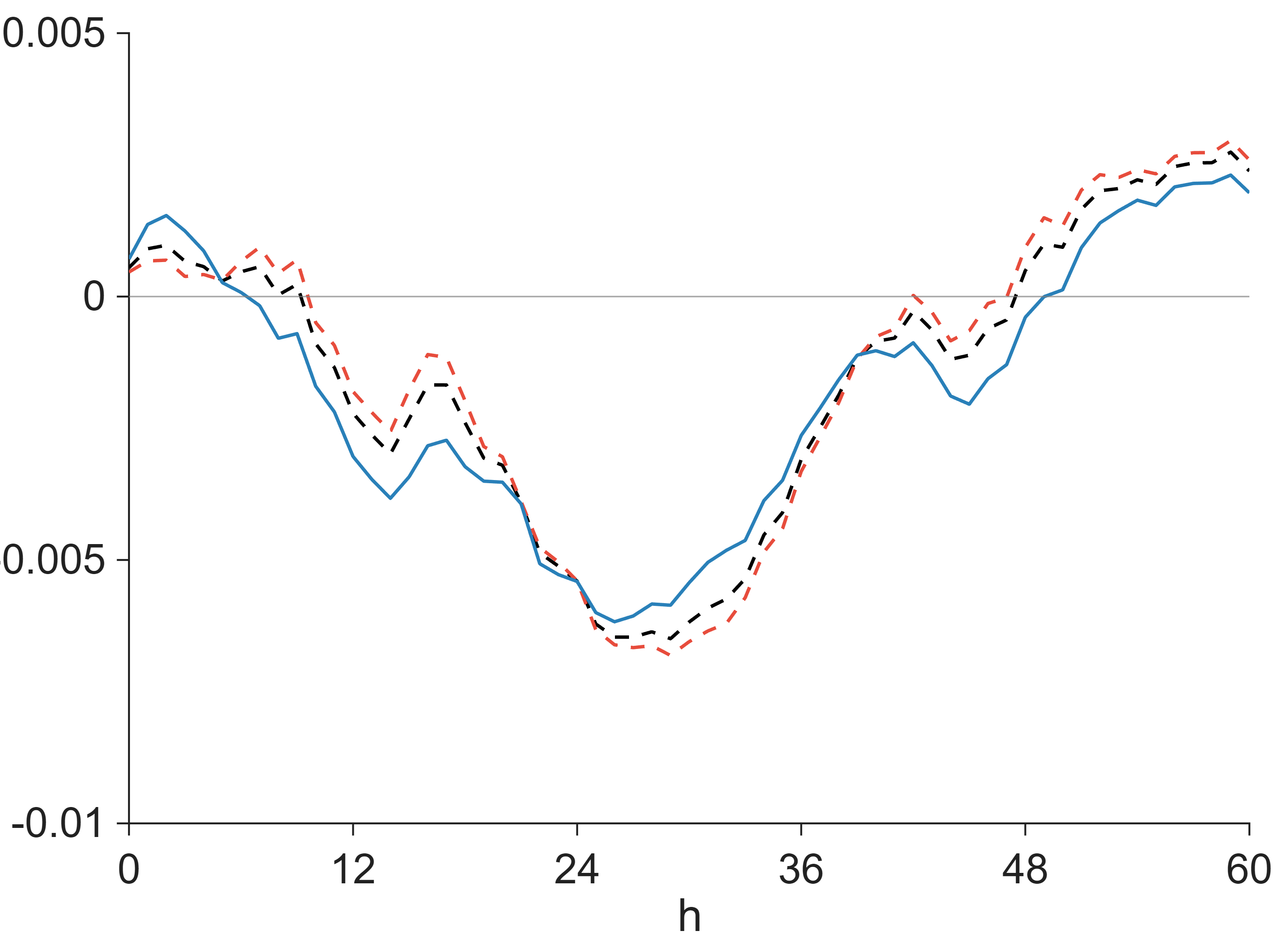} &
			\includegraphics[width=0.45\textwidth]{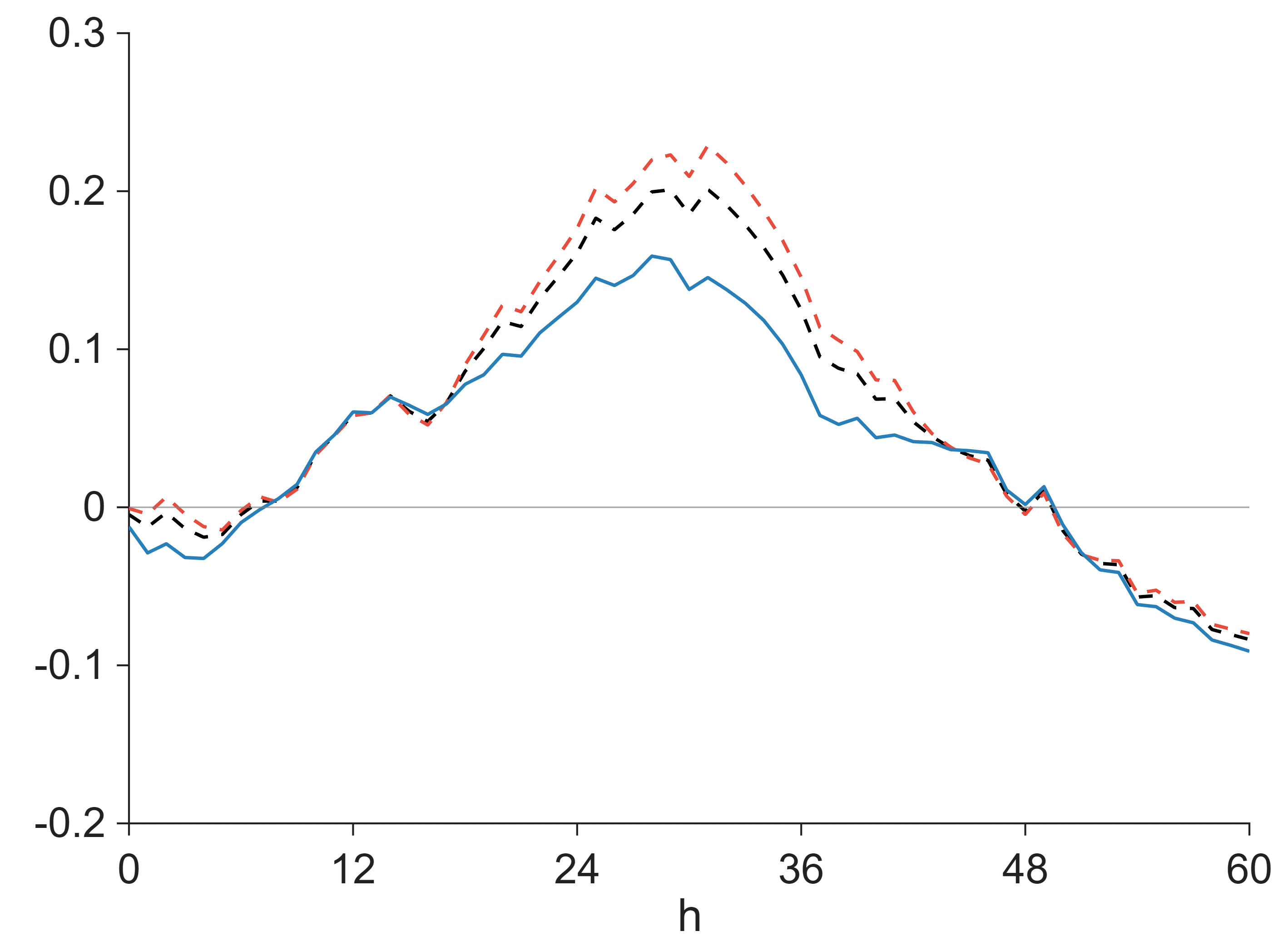} \\
      \textit{CPI} & \textit{FFR}\\ 
      \includegraphics[width=0.45\textwidth]{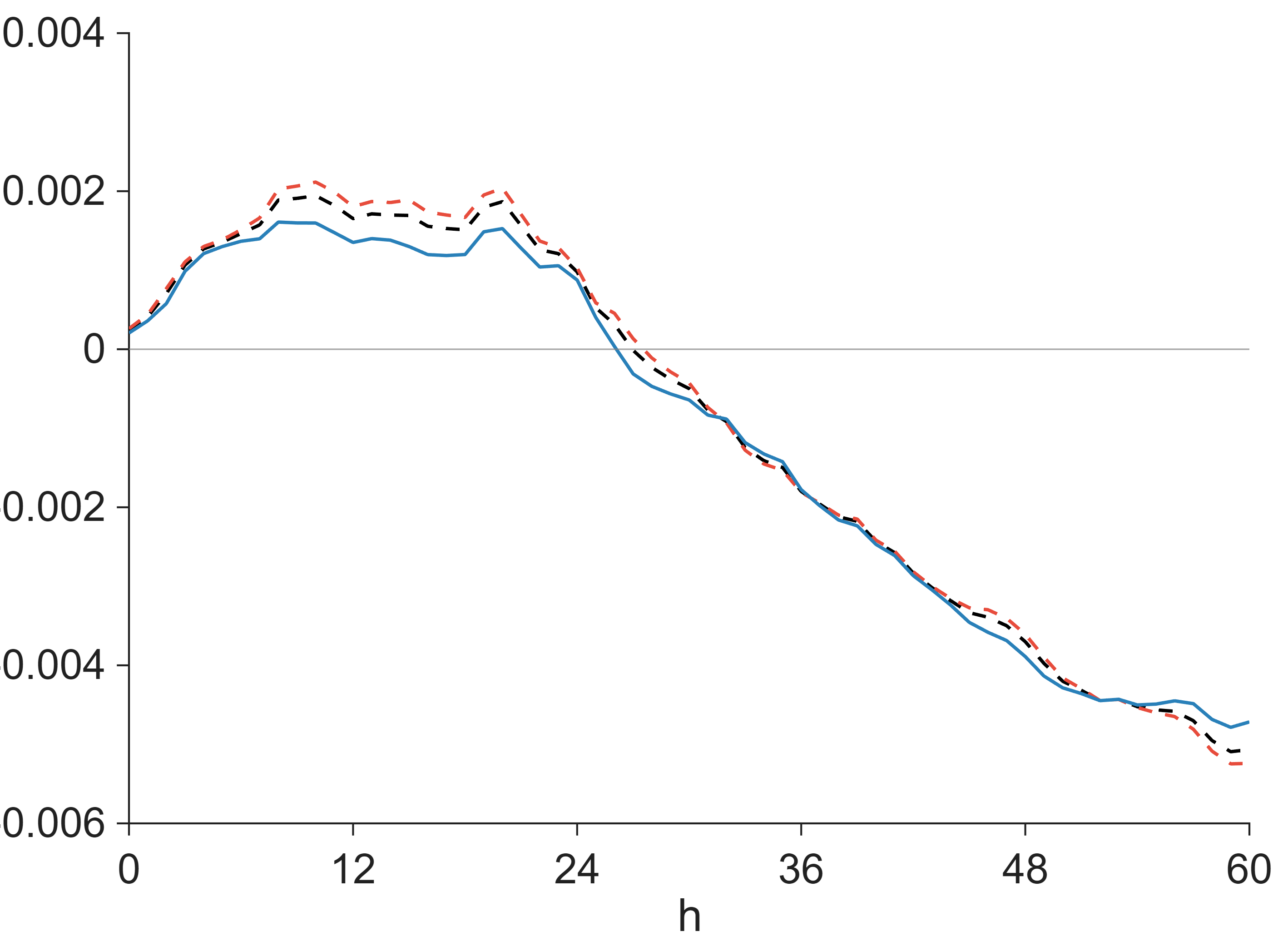} & 
      \includegraphics[width=0.45\textwidth]{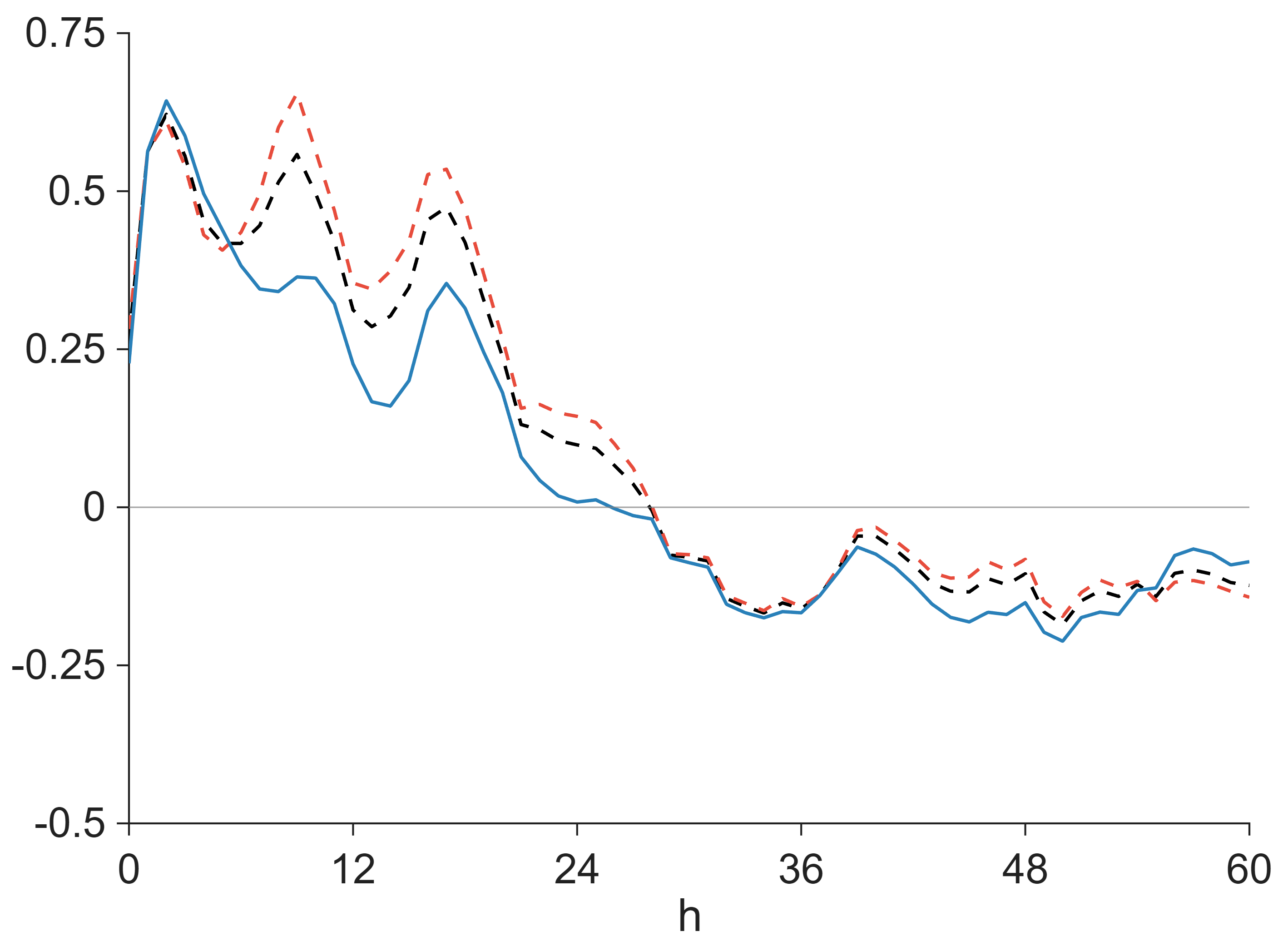}
		\end{tabular}
	\end{center}
	{\footnotesize {\em Notes}: Scaled impulse response estimates for different shock sizes. Black dashed line: $+1$ std shock $\text{IRF}^{\textit{Feas}}(z_{ss}, \sigma_{MP}; h)$; red dashed line: $+2$ std shock $\text{IRF}^{\textit{Feas}}(z_{ss}, 2\sigma_{MP}; h)/2$; blue solid line: $-1$ std shock $-\text{IRF}^{\textit{Feas}}(z_{ss}, -\sigma_{MP}; h)$.}
  \setlength{\baselineskip}{4mm}
\end{figure}

\clearpage
\section{Conclusion}
\label{sec:conclusion}
We use QVARs as a laboratory to assess LP specifications in nonlinear environments and show that linear LP can average out nonlinearities under symmetric shocks, while existing state-dependent variants mainly deliver gains in tail shocks or tail states. For practitioner, we recommend the following workflow. 
\begin{enumerate}
\item \textit{Use \textit{Linear} as the benchmark.} It provides a useful baseline, but be cautious that it may fail to capture nonlinearities, especially when shocks are roughly symmetric. A descriptive plot of the shock distribution, or a plot of the estimated causal weight function like Figure~\ref{fig:rr.causal.weight}, can evaluate whether this concern is relevant.
\item \textit{Use \textit{Feas} as the default nonlinear specification.} Choose interpretable state proxies, report IRFs at a few benchmark states and shock sizes, and use HAC/HAR rather than EHW inference. Diagnostic plots like Figure~\ref{fig:irf.state.depend} and Figure~\ref{fig:irf.high.order} can help assess the form of nonlinearity present in the data.
\item \textit{Use \textit{NPLP} as a complementary benchmark.} \textit{NPLP} is useful for checking whether the main qualitative patterns uncovered by \textit{Feas} persist under a more flexible specification. Agreement between the two increases confidence that the findings are not driven by the structure imposed by \textit{Feas} (as in Figure~\ref{fig:irf.state.depend}); substantial discrepancies can be informative about possible misspecification.

\end{enumerate}

\small  \renewcommand{\baselinestretch}{1.1} \normalsize

\bibliography{ref_qar}

@article{KolesarPlagborg-Moller2025,
  title={Dynamic causal effects in a nonlinear world: the good, the bad, and the ugly},
  author={Koles{\'a}r, Michal and Plagborg-M{\o}ller, Mikkel},
  journal={Journal of Business \& Economic Statistics},
  volume={43},
  number={4},
  pages={737--754},
  year={2025},
  publisher={Taylor \& Francis}
}

@article{HerbstJohannsen2025,
  title={Discussion of “Dynamic Causal Effects in a Nonlinear World: The Good, the Bad, and the Ugly”},
  author={Herbst, Edward P and Johannsen, Benjamin K},
  journal={Journal of Business \& Economic Statistics},
  volume={43},
  number={4},
  pages={761--765},
  year={2025},
  publisher={Taylor \& Francis}
}

@article{Jorda2005,
  title={Estimation and inference of impulse responses by local projections},
  author={Jord{\`a}, {\`O}scar},
  journal={American Economic Review},
  volume={95},
  number={1},
  pages={161--182},
  year={2005},
  publisher={American Economic Association}
}

@article{AruobaEtAl2017,
  title={Assessing DSGE model nonlinearities},
  author={Aruoba, S Bora{\u{g}}an and Bocola, Luigi and Schorfheide, Frank},
  journal={Journal of Economic Dynamics and Control},
  volume={83},
  pages={34--54},
  year={2017},
  publisher={Elsevier}
}

@article{AndreasenEtAl2018,
  title={The pruned state-space system for non-linear DSGE models: Theory and empirical applications},
  author={Andreasen, Martin M and Fern{\'a}ndez-Villaverde, Jes{\'u}s and Rubio-Ram{\'\i}rez, Juan F},
  journal={The Review of Economic Studies},
  volume={85},
  number={1},
  pages={1--49},
  year={2018},
  publisher={Oxford University Press}
}

@article{BenZeevEtAl2023,
  title={Do government spending multipliers depend on the sign of the shock?},
  author={Ben Zeev, Nadav and Ramey, Valerie A and Zubairy, Sarah},
  journal={AEA Papers and Proceedings},
  volume={113},
  pages={382--387},
  year={2023},
  organization={American Economic Association 2014 Broadway, Suite 305, Nashville, TN 37203}
}

@article{FurceriEtAl2018,
  title={The effects of monetary policy shocks on inequality},
  author={Furceri, Davide and Loungani, Prakash and Zdzienicka, Aleksandra},
  journal={Journal of International Money and Finance},
  volume={85},
  pages={168--186},
  year={2018},
  publisher={Elsevier}
}

@article{AlbrizioEtAl2020,
  title={International bank lending channel of monetary policy},
  author={Albrizio, Silvia and Choi, Sangyup and Furceri, Davide and Yoon, Chansik},
  journal={Journal of International Money and Finance},
  volume={102},
  pages={102124},
  year={2020},
  publisher={Elsevier}
}

@article{AuerbachGorodnichenko2013,
  title={Output spillovers from fiscal policy},
  author={Auerbach, Alan J and Gorodnichenko, Yuriy},
  journal={American Economic Review},
  volume={103},
  number={3},
  pages={141--146},
  year={2013},
  publisher={American Economic Association}
}

@article{AuerEtAl2021,
  title={Corporate leverage and monetary policy effectiveness in the euro area},
  author={Auer, Simone and Bernardini, Marco and Cecioni, Martina},
  journal={European Economic Review},
  volume={140},
  pages={103943},
  year={2021},
  publisher={Elsevier}
}

@article{JordaEtAl2013,
  title={When credit bites back},
  author={Jord{\`a}, {\`O}scar and Schularick, Moritz and Taylor, Alan M},
  journal={Journal of Money, Credit and Banking},
  volume={45},
  number={s2},
  pages={3--28},
  year={2013},
  publisher={Wiley Online Library}
}

@article{RameyZubairy2018,
  title={Government spending multipliers in good times and in bad: evidence from US historical data},
  author={Ramey, Valerie A and Zubairy, Sarah},
  journal={Journal of Political Economy},
  volume={126},
  number={2},
  pages={850--901},
  year={2018},
  publisher={University of Chicago Press Chicago, IL}
}

@article{JordaEtAl2020,
  title={The effects of quasi-random monetary experiments},
  author={Jord{\`a}, {\`O}scar and Schularick, Moritz and Taylor, Alan M},
  journal={Journal of Monetary Economics},
  volume={112},
  pages={22--40},
  year={2020},
  publisher={Elsevier}
}

@article{GonccalvesEtAl2024,
  title={State-dependent local projections},
  author={Gon{\c{c}}alves, S{\'\i}lvia and Herrera, Ana Mar{\'\i}a and Kilian, Lutz and Pesavento, Elena},
  journal={Journal of Econometrics},
  volume={244},
  number={2},
  pages={105702},
  year={2024},
  publisher={Elsevier}
}

@article{KoopEtAl1996,
  title={Impulse response analysis in nonlinear multivariate models},
  author={Koop, Gary and Pesaran, M Hashem and Potter, Simon M},
  journal={Journal of Econometrics},
  volume={74},
  number={1},
  pages={119--147},
  year={1996},
  publisher={Elsevier}
}

@article{BlanchardPerotti2002,
  title={An empirical characterization of the dynamic effects of changes in government spending and taxes on output},
  author={Blanchard, Olivier and Perotti, Roberto},
  journal={The Quarterly Journal of Economics},
  volume={117},
  number={4},
  pages={1329--1368},
  year={2002},
  publisher={MIT Press}
}

@article{Fernald2014,
  title={A quarterly, utilization-adjusted series on total factor productivity},
  author={Fernald, John},
  journal={Federal Reserve Bank of San Francisco Working Paper},
  year={2014},
  number={2012-19}
}

@article{Ramey2011,
  title={Identifying government spending shocks: It's all in the timing},
  author={Ramey, Valerie A},
  journal={The Quarterly Journal of Economics},
  volume={126},
  number={1},
  pages={1--50},
  year={2011},
  publisher={MIT Press}
}

@article{AuerbachGorodnichenko2012,
  title={Measuring the output responses to fiscal policy},
  author={Auerbach, Alan J and Gorodnichenko, Yuriy},
  journal={American Economic Journal: Economic Policy},
  volume={4},
  number={2},
  pages={1--27},
  year={2012},
  publisher={American Economic Association}
}

@article{AlesinaEtAl2018,
  title={Is it the “How” or the “When” that Matters in Fiscal Adjustments?},
  author={Alesina, Alberto and Azzalini, Gualtiero and Favero, Carlo and Giavazzi, Francesco and Miano, Armando},
  journal={IMF Economic Review},
  volume={66},
  number={1},
  pages={144--188},
  year={2018},
  publisher={Palgrave Macmillan}
}

@article{AuerbachGorodnichenko2016,
  title={Effects of fiscal shocks in a globalized world},
  author={Auerbach, Alan J and Gorodnichenko, Yuriy},
  journal={IMF Economic Review},
  volume={64},
  number={1},
  pages={177--215},
  year={2016}
}

@article{BernardiniEtAl2020,
  title={Heterogeneous government spending multipliers in the era surrounding the great recession},
  author={Bernardini, Marco and De Schryder, Selien and Peersman, Gert},
  journal={Review of Economics and Statistics},
  volume={102},
  number={2},
  pages={304--322},
  year={2020},
  publisher={MIT Press One Rogers Street, Cambridge, MA 02142-1209, USA journals-info~…}
}

@article{BornEtAl2020,
  title={Does austerity pay off?},
  author={Born, Benjamin and M{\"u}ller, Gernot J and Pfeifer, Johannes},
  journal={Review of Economics and Statistics},
  volume={102},
  number={2},
  pages={323--338},
  year={2020},
  publisher={MIT Press One Rogers Street, Cambridge, MA 02142-1209, USA journals-info~…}
}

@article{Ramey2016,
  title={Macroeconomic shocks and their propagation},
  author={Ramey, Valerie A},
  journal={Handbook of Macroeconomics},
  volume={2},
  pages={71--162},
  year={2016},
  publisher={Elsevier}
}

@article{RomerRomer2004,
  title={A new measure of monetary shocks: Derivation and implications},
  author={Romer, Christina D and Romer, David H},
  journal={American Economic Review},
  volume={94},
  number={4},
  pages={1055--1084},
  year={2004},
  publisher={American Economic Association}
}

@article{WielandYang2020,
  title={Financial dampening},
  author={Wieland, Johannes F and Yang, Mu-Jeung},
  journal={Journal of Money, Credit and Banking},
  volume={52},
  number={1},
  pages={79--113},
  year={2020},
  publisher={Wiley Online Library}
}

@article{KimEtAl2008,
  title={Calculating and using second-order accurate solutions of discrete time dynamic equilibrium models},
  author={Kim, Jinill and Kim, Sunghyun and Schaumburg, Ernst and Sims, Christopher A},
  journal={Journal of Economic Dynamics and Control},
  volume={32},
  number={11},
  pages={3397--3414},
  year={2008},
  publisher={Elsevier}
}

@article{GoncalvesEtAl2024,
  title={Nonparametric Local Projections},
  author={Gon{\c{c}}alves, S{\'\i}lvia and Herrera, Ana Mar{\'\i}a and Kilian, Lutz and Pesavento, Elena},
  year={2024},
  journal={FRB of Dallas Working Paper}
}

@article{BrunnermeierSannikov2014,
  title={A macroeconomic model with a financial sector},
  author={Brunnermeier, Markus K and Sannikov, Yuliy},
  journal={American Economic Review},
  volume={104},
  number={2},
  pages={379--421},
  year={2014},
  publisher={American Economic Association 2014 Broadway, Suite 305, Nashville, TN 37203}
}

@article{EggertssonWoodford2003,
  title={Zero bound on interest rates and optimal monetary policy},
  author={Eggertsson, Gauti B and Woodford, Michael},
  journal={Brookings Papers on Economic Activity},
  volume={2003},
  number={1},
  pages={139--233},
  year={2003},
  publisher={Johns Hopkins University Press}
}

@article{EhrmannEtAl2003,
  title={Regime-dependent impulse response functions in a Markov-switching vector autoregression model},
  author={Ehrmann, Michael and Ellison, Martin and Valla, Natacha},
  journal={Economics Letters},
  volume={78},
  number={3},
  pages={295--299},
  year={2003},
  publisher={Elsevier}
}

@article{MontielOleaEtAl2021,
  title={Local projection inference is simpler and more robust than you think},
  author={Montiel Olea, Jos{\'e} Luis and Plagborg-M{\o}ller, Mikkel},
  journal={Econometrica},
  volume={89},
  number={4},
  pages={1789--1823},
  year={2021},
  publisher={Wiley Online Library}
}

@article{FernandezEtAl2015,
  title={Estimating dynamic equilibrium models with stochastic volatility},
  author={Fern{\'a}ndez-Villaverde, Jes{\'u}s and Guerr{\'o}n-Quintana, Pablo and Rubio-Ram{\'\i}rez, Juan F},
  journal={Journal of Econometrics},
  volume={185},
  number={1},
  pages={216--229},
  year={2015},
  publisher={Elsevier}
}

@article{Andreasen2012,
  title={An estimated DSGE model: Explaining variation in nominal term premia, real term premia, and inflation risk premia},
  author={Andreasen, Martin M},
  journal={European Economic Review},
  volume={56},
  number={8},
  pages={1656--1674},
  year={2012},
  publisher={Elsevier}
}

@article{KimRuge-Murcia2009,
  title={How much inflation is necessary to grease the wheels?},
  author={Kim, Jinill and Ruge-Murcia, Francisco J},
  journal={Journal of Monetary Economics},
  volume={56},
  number={3},
  pages={365--377},
  year={2009},
  publisher={Elsevier}
}

@incollection{FernandezEtAl2016,
  title={Solution and estimation methods for DSGE models},
  author={Fern{\'a}ndez-Villaverde, Jes{\'u}s and Rubio-Ram{\'\i}rez, Juan Francisco and Schorfheide, Frank},
  booktitle={Handbook of macroeconomics},
  volume={2},
  pages={527--724},
  year={2016},
  publisher={Elsevier}
}

@article{GoncalvesEtAl2025,
  title={Discussion of:“Dynamic Causal Effects in a Nonlinear World: the Good, the Bad, and the Ugly”},
  author={Gon{\c{c}}alves, S{\'\i}lvia and Herrera, Ana Mar{\'\i}a and Pesavento, Elena},
  journal={Journal of Business \& Economic Statistics},
  volume={43},
  number={4},
  pages={755--760},
  year={2025},
  publisher={Taylor \& Francis}
}

@article{Hamilton2018,
  title={Why you should never use the Hodrick-Prescott filter},
  author={Hamilton, James D},
  journal={Review of Economics and Statistics},
  volume={100},
  number={5},
  pages={831--843},
  year={2018},
  publisher={MIT Press One Rogers Street, Cambridge, MA 02142-1209, USA journals-info~…}
}

@article{Wu2005,
  title={Nonlinear system theory: Another look at dependence},
  author={Wu, Wei Biao},
  journal={Proceedings of the National Academy of Sciences},
  volume={102},
  number={40},
  pages={14150--14154},
  year={2005},
  doi={10.1073/pnas.0506715102},
  publisher={National Academy of Sciences}
}

@article{Wu2011,
  title={Asymptotic theory for stationary processes},
  author={Wu, Wei Biao},
  journal={Statistics and Its Interface},
  volume={4},
  number={2},
  pages={207--226},
  year={2011},
  doi={10.4310/SII.2011.v4.n2.a15},
  publisher={International Press}
}

\small  \renewcommand{\baselinestretch}{1.3} \normalsize
		

\clearpage

\renewcommand{\thepage}{A.\arabic{page}}
\setcounter{page}{1}

\begin{appendix}
	\markright{Online Appendix -- This Version: \today }
	
	\renewcommand{\theequation}{A.\arabic{equation}}
	\setcounter{equation}{0}	
	\renewcommand*\thetable{A-\arabic{table}}
	\setcounter{table}{0}
	\renewcommand*\thefigure{A-\arabic{figure}}
	\setcounter{figure}{0}
	\renewcommand*\thetheorem{A-\arabic{theorem}}
	\setcounter{theorem}{0}

	\begin{center}
		
		{\large {\bf Online Appendix: How Well Are State-Dependent Local Projections Capturing Nonlinearities?}}
		
		{\bf Zhiheng You}
	\end{center}
	
	\thispagestyle{empty} 
	
	\noindent This Appendix consists of the following sections: 
	
	\begin{itemize}
            \item[A.] Derivation of Structural Functions 
            \item[B.] Alternative IRF Definitions: Details
		    \item[C.] Proofs
            \item[D.] Details and Additional Results for Section \ref{sec:empiric.exp}
	\end{itemize}
	
	\newpage
	
\section{Derivation of Structural Functions}
\subsection{QAR(1, 1)}
In this section, we represent $y_{t+h}$ as a structural function $\psi_h$ of shock $u_t$ and other variables $U_{h,t+h}=(y_{t-1}, s_{t-1}, u_{t+1},\ldots, u_{t+h})$ independent of $u_t$ for the QAR(1,1) model.

For $h\ge1$, 
\begin{equation*}
\begin{aligned}
    y_{t+h} &= \phi_1 y_{t+h-1} + \phi_2 s_{t+h-1}^2 + (1+\gamma s_{t+h-1})\sigma u_{t+h},\\
    &= \phi_1 \left[\phi_1 y_{t+h-2} + \phi_2 s_{t+h-2}^2 + (1+\gamma s_{t+h-2})\sigma u_{t+h-1}\right] + \phi_2 s_{t+h-1}^2 + (1+\gamma s_{t+h-1})\sigma u_{t+h},\\
    &=\ldots\\
    &=\phi_1^{h+1} y_{t-1} + \phi_1^{h}\phi_2 s_{t-1}^2 + \sigma \phi_1^{h} u_t + \gamma\sigma \phi_1^h s_{t-1}u_t  + \phi_2 \sum_{k=0}^{h-1}\phi_1^{h-k-1} s_{t+k}^2 \\
    & + \sigma \sum_{k=0}^{h-1} \phi_1^{h-k-1} u_{t+1+k} + \gamma\sigma \sum_{k=0}^{h-1} \phi_1^{h-1-k} s_{t+k} u_{t+1+k},\\
    &=\phi_1^{h+1} y_{t-1} + \phi_1^{h}\phi_2 s_{t-1}^2 + \sigma \phi_1^{h} u_t + \gamma\sigma \phi_1^h s_{t-1}u_t \\
    & + \phi_2 \sigma^2 \sum_{k=0}^{h-1}\phi_1^{h-k-1} \left(\phi_1^{2k} u_{t}^2 +  \left(\sum_{\substack{j=0 \\ j\neq k}}^{\infty} \phi_1^j u_{t+k-j}\right)^2 + 2 \phi_1^k u_{t} \sum_{\substack{j=0 \\ j\neq k}}^{\infty} \phi_1^j u_{t+k-j}\right) \\
    & + \sigma \sum_{k=0}^{h-1} \phi_1^{h-k-1} u_{t+1+k} + \gamma\sigma^2 \sum_{k=0}^{h-1} \phi_1^{h-1-k} \left(\phi_1^k u_t + \sum_{\substack{j=0 \\ j\neq k}}^{\infty} \phi_1^j u_{t+k-j}\right) u_{t+1+k}.
\end{aligned}
\end{equation*}

Note that we can replace $$\sum_{\substack{j=0 \\ j\neq k}}^{\infty} \phi_1^j u_{t+k-j} = \phi_1^{k+1} \frac{s_{t-1}}{\sigma}+\sum_{j=0}^{k-1} \phi_1^j u_{t+k-j}$$ for $k\ge 1$ and $$\sum_{\substack{j=0 \\ j\neq k}}^{\infty} \phi_1^j u_{t+k-j} = \phi_1^{k+1} \frac{s_{t-1}}{\sigma}$$ for $k=0$ to get rid of past shocks $u_{t-1}, u_{t-2},\ldots$ (The two cases can be combined if we use the empty-sum convention when $k=0.$)

For $h=0$, structural function $\psi_h$ is simply $y_t=\phi_1 y_{t-1}+\phi_2 s_{t-1}^2+ (1+\gamma s_{t-1})\sigma u_t$.

\subsection{QVAR(1,1)}
In this section, we represent $y_{t+h}$ as a structural function $\tilde{\psi}_h$ of shock $u_t$ and other variables $U_{h,t+h}=\bigl(y_{t-1}, s_{t-1}, u_{t+1},\dots,u_{t+h}\bigr)$ independent of $u_t$ for the QVAR(1,1) model.

We start from the QVAR(1,1):
\begin{align}
y_t &= \Phi_1 y_{t-1} + \Phi_2 \operatorname{vech}(s_{t-1}s_{t-1}') + (1_n + \mathsf{G} s_{t-1}) \odot (\Sigma_{tr} u_t), \label{eq:y_def} \\
s_t &= \Phi_1 s_{t-1} + \Sigma_{tr} u_t, \label{eq:s_def}
\end{align}

For each $k=1,\dots,h$,
\[ s_{t-1+k} = \Phi_1^ks_{t-1} + \sum_{j=0}^{k-1} \Phi_1^{k-1-j}\Sigma_{tr} u_{t+j}. \]

Iterating \eqref{eq:y_def} forward,
\[ y_{t+h} = \Phi_1^{h+1}y_{t-1} + \sum_{k=0}^h \Phi_1^{h-k} \Bigl[ \Phi_2 \vech\left(s_{t-1+k}s_{t-1+k}'\right) + \left(1_n+\mathsf{G} s_{t-1+k}\right)\odot (\Sigma_{tr} u_{t+k}) \Bigr]. \]

For $k\ge1$, write
\[ \begin{aligned}
\vech\left(s_{t-1+k}s_{t-1+k}'\right) &= \vech\left( \left(\Phi_1^ks_{t-1}\right) \left(\Phi_1^ks_{t-1}\right)'\right)\\[6pt]
&\quad + \vech\left( \left(\Phi_1^ks_{t-1}\right)\left(\Phi_1^{k-1}\Sigma_{tr}u_t\right)' + \left(\Phi_1^{k-1}\Sigma_{tr}u_t\right)\left(\Phi_1^ks_{t-1}\right)' \right)\\[6pt]
&\quad + \sum_{j=1}^{k-1} \vech\left( \left(\Phi_1^ks_{t-1}\right)\left(\Phi_1^{k-1-j}\Sigma_{tr}u_{t+j}\right)' + \left(\Phi_1^{k-1-j}\Sigma_{tr}u_{t+j}\right)\left(\Phi_1^ks_{t-1}\right)' \right)\\[6pt]
&\quad + \vech\left( \left(\Phi_1^{k-1}\Sigma_{tr}u_t\right) \left(\Phi_1^{k-1}\Sigma_{tr}u_t\right)' \right)\\[6pt]
&\quad + \sum_{j=1}^{k-1} \vech\left( \left(\Phi_1^{k-1}\Sigma_{tr}u_t\right) \left(\Phi_1^{k-1-j}\Sigma_{tr}u_{t+j}\right)' + \left(\Phi_1^{k-1-j}\Sigma_{tr}u_{t+j}\right) \left(\Phi_1^{k-1}\Sigma_{tr}u_t\right)' \right)\\[6pt]
&\quad + \vech\left( \left(\sum_{j=1}^{k-1}\Phi_1^{k-1-j}\Sigma_{tr}u_{t+j}\right) \left(\sum_{j=1}^{k-1}\Phi_1^{k-1-j}\Sigma_{tr}u_{t+j}\right)' \right),
\end{aligned} \]
and
\[ \begin{aligned}
\left(1_n+\mathsf{G} s_{t-1+k}\right)\odot (\Sigma_{tr} u_{t+k}) &=  1_n\odot (\Sigma_{tr} u_{t+k}) + \left(\mathsf{G} \Phi_1^{k}s_{t-1}\right)\odot (\Sigma_{tr} u_{t+k})\\
&\quad + \left(\mathsf{G} \Phi_1^{k-1}\Sigma_{tr} u_t\right)\odot (\Sigma_{tr} u_{t+k}) + \sum_{j=1}^{k-1} \left(\mathsf{G} \Phi_1^{k-1-j}\Sigma_{tr}u_{t+j}\right)\odot \left(\Sigma_{tr} u_{t+k}\right).
\end{aligned} \]

Plugging these into the iterate for $y_{t+h}$ and collecting powers of $u_t$, we obtain for $h\ge 1$,
\[ y_{t+h} = R_{t,h} + L_{t,h}\left(u_t\right) + C_{t,h}\left(u_t,u_{t+1:t+h}\right) + Q_{t,h}\left(u_t\right), \]
where
\[ \begin{aligned}
L_{t,h}\left(u_t\right) &= \Phi_1^h \left(1_n+\mathsf{G} s_{t-1}\right)\odot \left(\Sigma_{tr} u_{t}\right)\\
&\quad + \sum_{k=1}^h \Phi_1^{h-k} \Phi_2 \vech\left( \left(\Phi_1^ks_{t-1}\right) \left(\Phi_1^{k-1}\Sigma_{tr}u_t\right)' + \left(\Phi_1^{k-1}\Sigma_{tr}u_t\right) \left(\Phi_1^ks_{t-1}\right)' \right),\\[6pt]
C_{t,h}\left(u_t,u_{t+1:t+h}\right) &= \sum_{k=1}^h \Phi_1^{h-k} \Phi_2 \sum_{j=1}^{k-1} \vech\left( \left(\Phi_1^{k-1}\Sigma_{tr}u_t\right) \left(\Phi_1^{k-1-j}\Sigma_{tr}u_{t+j}\right)' + \left(\Phi_1^{k-1-j}\Sigma_{tr}u_{t+j}\right) \left(\Phi_1^{k-1}\Sigma_{tr}u_t\right)' \right)\\
&\quad + \sum_{k=1}^h \Phi_1^{h-k} \left(\mathsf{G} \Phi_1^{k-1}\Sigma_{tr} u_t\right)\odot \left(\Sigma_{tr} u_{t+k}\right),\\[6pt]
Q_{t,h}\left(u_t\right) &= \sum_{k=1}^h \Phi_1^{h-k} \Phi_2 \vech\left( \left(\Phi_1^{k-1}\Sigma_{tr}u_t\right) \left(\Phi_1^{k-1}\Sigma_{tr}u_t\right)' \right),
\end{aligned} \]
and $R_{t,h}$ contains only the nuisance variables and no $u_t$. For $h=0$, the structural function $\tilde{\psi}_h$ is simply
\[ y_t = \Phi_1y_{t-1} + \Phi_2 \vech\left(s_{t-1}s_{t-1}'\right) + \left(1_n+\mathsf{G} s_{t-1}\right)\odot \left(\Sigma_{tr} u_{t}\right). \]

\clearpage
\section{Alternative IRF Definitions: Details} \label{app:irf_details}

\subsection{Infinitesimal versus Finite Shocks}
The conditional marginal response (CMR) proposed in \cite{GonccalvesEtAl2024} studies impulse response to an infinitesimal shock. Specifically,  CMR to an infinitesimal shock in $u_{t}$ is defined as 
$$
\operatorname{CMR}_h(\mathcal{F})=\lim _{\delta \rightarrow 0} \frac{\operatorname{CAR}_h(\mathcal{F}, \delta)}{\delta}.
$$
The CMR for QAR(1,1) is
$$
    \operatorname{CMR}_h(s) = \underbrace{\sigma \phi_1^{h} \left(1 + \gamma s + 2 \phi_2 s \frac{1-\phi_1^h}{1-\phi_1} \right)}_{\text{first-order effect}}.
$$
Since the CMR captures only first-order effects and not second-order effects, we do not use it as our IRF definition, to enable a fair comparison across specifications.

Like CAR, the generalized IRF proposed by \cite{KoopEtAl1996} studies impulse response to a finite shock:
$$\operatorname{GIRF}_h (\mathcal{F}, \delta)= \mathbb{E}[y_{t+h}|u_t = \delta, \mathcal{F}_{t-1}=\mathcal{F}] - \mathbb{E}[y_{t+h}|\mathcal{F}_{t-1}=\mathcal{F}].$$
\cite{HerbstJohannsen2025} calculate the GIRF for QAR(1,1) as
\begin{equation*}
  \operatorname{GIRF}_h (s, \delta) = \phi_1^h(1+\gamma s) \sigma \delta+\phi_2 \phi_1^{h-1} \frac{1-\phi_1^h}{1-\phi_1}\left[2 \phi_1 \sigma s \delta+\sigma^2\left(\delta^2-1\right)\right].
\end{equation*}
Hence, for this model, GIRF and CAR only differ by a constant that does not depend on the shock size or initial state:
\begin{equation*}
\operatorname{GIRF}_h\left(s, \delta\right) = \operatorname{CAR}_h\left(s, \delta\right) - \phi_2\phi_1^{h-1}\frac{1-\phi_1^h}{1-\phi_1}\sigma^2.
\end{equation*}

\subsection{Conditional versus Unconditional IRFs}
Our CAR and CMR definitions condition on the realized state $s_{t-1}$. In contrast, KP define unconditional causal parameters that integrate over states. The average marginal effect is:
$$\theta_h(\omega) = \int \omega(u) \Psi_h^{'}(u) \ du,$$
and the average causal effect of a shock with magnitude $\delta>0$ is:
$$\theta_h\left(\delta, \omega\right) \equiv \frac{1}{\delta} \int \omega(u)\left\{\Psi_h(u+\delta)-\Psi_h(u)\right\} d u,$$
where $\omega(\cdot)$ is a nonnegative weight function over baseline shock values that integrates to 1.

The relationship between our IRF definitions and KP's causal parameters can be established as follows:

\noindent (a) \textit{Avg. causal effect versus CAR}: Let $p_u$ be the distribution of $u_t$. Then:
$$\theta_h(\delta, p_u) = \frac{1}{\delta} \int p_u(u)\left\{\Psi_h(u+\delta)-\Psi_h(u)\right\} du = \frac{1}{\delta} \mathbb{E}[\Psi_h(u_t+\delta)-\Psi_h(u_t)] = \frac{\mathbb{E}[\operatorname{CAR}_h(s_{t-1}, \delta)]}{\delta}.$$

\noindent (b) \textit{Avg. marginal effect versus CMR}: Taking the limit as $\delta \rightarrow 0$:
$$\theta_h(p_u) = \lim_{\delta \rightarrow 0} \theta_h(\delta, p_u) = \lim_{\delta \rightarrow 0} \frac{\mathbb{E}[\operatorname{CAR}_h(s_{t-1}, \delta)]}{\delta} = \mathbb{E}[\text{CMR}_h(s_{t-1})].$$

Thus, KP's causal effect definitions rule out state-dependency by construction, averaging over all possible states.

\subsection{Evolving versus Fixed States in Regime-Switching Models}
In models with Markov-switching states $S_t$, IRFs can allow states to evolve naturally or hold them fixed. The regime-dependent impulse response of \cite{EhrmannEtAl2003} fixes the state throughout the horizon:
$$\operatorname{IRF}^{\text{fixed}}_h(s, \delta) = \mathbb{E}[y_{t+h}|u_t = \delta, S_{t} = \cdots = S_{t+h} = s] - \mathbb{E}[y_{t+h}|S_{t} = \cdots = S_{t+h} = s].$$
Fixed-state IRFs isolate within-regime dynamics for mechanism analysis, while evolving-state IRFs capture full propagation including endogenous regime changes, essential for policy assessment.

\clearpage
\section{Proofs}

\subsection{Proof of Proposition \ref{prop:true.car}} 
\begin{proof}
First, using the formula of structural function $\psi_h$,
\begin{align*}
& \psi_h(u_{t}+\delta, U_{h,t+h})- \psi_h(u_{t}, U_{h,t+h}) = \sigma \phi_1^{h} \delta + \gamma\sigma \phi_1^h s_{t-1}\delta + \phi_2 \sigma^2 \sum_{k=0}^{h-1}\phi_1^{h+k-1} \left(2\delta u_{t} + \delta^2\right) \\
& +  \delta \cdot 2 \phi_2 \sigma^2 \sum_{k=0}^{h-1} \phi_1^{h-1} \sum_{\substack{j=0 \\ j\neq k}}^{\infty} \phi_1^{j} u_{t+k-j}  + \delta \cdot \gamma\sigma^2 \sum_{k=0}^{h-1} \phi_1^{h-1} u_{t+1+k} \\
& = \sigma \phi_1^{h} \delta + \gamma\sigma \phi_1^h s_{t-1}\delta + \phi_2 \sigma^2 \cdot \frac{\phi_1^{h-1} - \phi_1^{2h-1}}{1 - \phi_1} \cdot \left(2\delta u_{t} + \delta^2\right) \\
& +  \delta \cdot 2 \phi_2 \sigma^2 \sum_{k=0}^{h-1} \phi_1^{h-1} \left( \phi_1^{k+1} \frac{s_{t-1}}{\sigma}+\sum_{j=0}^{k-1} \phi_1^j u_{t+k-j} \right)  + \delta \cdot \gamma\sigma^2 \sum_{k=0}^{h-1} \phi_1^{h-1} u_{t+1+k}
\end{align*}

Note that the only past variable in the above equation is $s_{t-1}$. Then 
\begin{align*}
\operatorname{CAR}_h(\mathcal{F}, \delta) &= \mathbb{E}\left[\psi_h(u_{t}+\delta, U_{h,t+h})- \psi_h(u_{t}, U_{h,t+h}) \mid \mathcal{F}_{t-1} = \mathcal{F} \right]\\
& = \mathbb{E}\left[\psi_h(u_{t}+\delta, U_{h,t+h})- \psi_h(u_{t}, U_{h,t+h}) \mid s_{t-1}=s\right]\\
& = \sigma \phi_1^{h} \delta + \gamma\sigma \phi_1^h s \delta + \phi_2 \sigma^2 \cdot \frac{\phi_1^{h-1} - \phi_1^{2h-1}}{1 - \phi_1} \cdot \delta^2 + \delta \cdot 2 \phi_2 \sigma^2 \sum_{k=0}^{h-1} \phi_1^{h-1} \left( \phi_1^{k+1} \frac{s}{\sigma} \right) \\
& = \underbrace{\sigma \phi_1^{h} \left(1 + \gamma s + 2 \phi_2 s \frac{1-\phi_1^h}{1-\phi_1} \right) \delta}_{\text{first-order effect}} + \underbrace{\phi_2 \sigma^2 \cdot \frac{\phi_1^{h-1} - \phi_1^{2h-1}}{1 - \phi_1} \cdot \delta^2}_{\text{second-order effect}} \\
& := \operatorname{CAR}_h(s, \delta).
\end{align*}
and
\begin{align*}
\operatorname{CMR}_h(\mathcal{F}) & = \lim _{\delta \rightarrow 0} \frac{\operatorname{CAR}_h(\mathcal{F}, \delta)}{\delta}  = \lim _{\delta \rightarrow 0} \frac{\operatorname{CAR}_h(s, \delta)}{\delta} \\
& = \lim _{\delta \rightarrow 0} \sigma \phi_1^{h} \left(1 + \gamma s + 2 \phi_2 s \frac{1-\phi_1^h}{1-\phi_1} \right) + \phi_2 \sigma^2 \cdot \frac{\phi_1^{h-1} - \phi_1^{2h-1}}{1 - \phi_1} \cdot \delta \\
& =  \underbrace{\sigma \phi_1^{h} \left(1 + \gamma s + 2 \phi_2 s \frac{1-\phi_1^h}{1-\phi_1} \right)}_{\text{first-order effect}} := \operatorname{CMR}_h(s).
\end{align*}
\end{proof}

\subsection{Proof of Proposition \ref{prop:failure.lp}}
\begin{proof}
Using the formula of structural function $y_{t+h} = \psi_h(u_t, U_{h,t+h})$, the linear LP population coefficient for $h\ge1$ can be calculated as 
\begin{equation*}
\begin{aligned}
    \beta_h & = \frac{\mathbb{E}[y_{t+h}u_t]}{\Var(u_t)} \\
    & = \phi_1^{h+1}\mathbb{E}[y_{t-1}u_t]
  + \phi_1^h\phi_2 \mathbb{E}[s_{t-1}^2u_t]
  + \sigma \phi_1^h \mathbb{E}[u_t^2]
  + \gamma\sigma \phi_1^h \mathbb{E}[s_{t-1}u_t^2] \\
&\quad
  + \phi_2\sigma^2 \sum_{k=0}^{h-1}\phi_1^{h-k-1}\phi_1^{2k}\mathbb{E}[u_t^3] 
  + \phi_2\sigma^2 \sum_{k=0}^{h-1}\phi_1^{h-k-1}
  \mathbb{E}\!\left[\left(\sum_{\substack{j=0\\j\neq k}}^{\infty}\phi_1^j u_{t+k-j}\right)^2 u_t\right] \\
&\quad
  + 2\phi_2\sigma^2 \sum_{k=0}^{h-1}\phi_1^{h-k-1}\phi_1^k
  \mathbb{E}\!\left[u_t^2\left(\sum_{\substack{j=0\\j\neq k}}^{\infty}\phi_1^j u_{t+k-j}\right)\right] \\
&\quad
  + \sigma \sum_{k=0}^{h-1}\phi_1^{h-k-1}\mathbb{E}[u_{t+1+k}u_t] 
  + \gamma\sigma^2 \sum_{k=0}^{h-1}\phi_1^{h-k-1}\phi_1^k \mathbb{E}[u_t^2u_{t+1+k}] \\
&\quad
  + \gamma\sigma^2 \sum_{k=0}^{h-1}\phi_1^{h-k-1}
  \mathbb{E}\!\left[\left(\sum_{\substack{j=0\\j\neq k}}^{\infty}\phi_1^j u_{t+k-j}\right)u_{t+1+k}u_t\right]\\
    & = \sigma \phi_1^{h}.
\end{aligned}
\end{equation*}
as $\{u_t\}$ are iid standard normal and $u_t$ is independent of $s_{t-1}$ and $y_{t-1}$.

The linear LP population coefficient for $h=0$ is
\begin{equation*}
    \beta_0 = \frac{\mathbb{E}[y_{t}u_t]}{\Var(u_t)}=\frac{\mathbb{E}[(\phi_1 y_{t-1}+\phi_2 s_{t-1}^2+ (1+\gamma s_{t-1})\sigma u_t)\cdot u_t]}{\Var(u_t)}=\sigma.
\end{equation*}
Therefore, we have $\beta_h=\sigma\phi_1^{h \text{ for } h=0,1,\ldots}.$ The corresponding IRF is $\text{IRF}(\delta; h) = \sigma \phi_1^h \delta.$ 

When the underlying DGP is an AR(1), the population coefficient is 
\begin{equation*}
\beta_h =  \frac{\mathbb{E}[y_{t+h}u_t]}{\Var(u_t)} = \frac{\mathbb{E}[\phi_1^h \sigma u_t \cdot u_t]}{\Var(u_t)} = \phi_1^h \sigma,
\end{equation*}
which coincides with that of QAR(1, 1).
\end{proof}

\subsection{Proof of Proposition \ref{prop:irf.shock.state}}
\begin{proof}
Denote $X_t^{(+)} = \left(S_t, S_tu_t, S_tW_t'\right)^{\prime}$, $X_t^{(-)} = \left(1-S_t, (1-S_t)u_t, (1-S_t)W_t'\right)^{\prime}$, and coefficients
$\gamma_h^{(+)} = \left(\alpha_h^{(+)}, \beta_h^{(+)}, {\pi_{h}^{(+)}}^{\prime} \right)^{\prime}$,
$\gamma_h^{(-)} = \left(\alpha_h^{(-)}, \beta_h^{(-)}, {\pi_{h}^{(-)}}^{\prime} \right)^{\prime}$.
The population coefficients solve the normal equations
\[
\EE[X_t^{(+)} {X_t^{(+)}}']\gamma_h^{(+)}=\EE[X_t^{(+)}y_{t+h}],
\qquad
\EE[X_t^{(-)} {X_t^{(-)}}']\gamma_h^{(-)}=\EE[X_t^{(-)}y_{t+h}].
\]
As in the derivations in the text, solving these normal equations yields
\be
 \beta_h^{(+)} = \frac{\EE\!\left[S_t\right] \cdot \EE\!\left[S_t u_t y_{t+h}\right] - \EE\!\left[S_t u_t \right] \cdot \EE\!\left[S_t y_{t+h}\right]}{\EE\!\left[S_t\right] \cdot \EE\!\left[S_t u_t^2\right] - \EE\!\left[S_t u_t\right]^2},
\label{eq:beta_plus_normal_eq}
\ee
and
\be
 \beta_h^{(-)} = \frac{\EE\!\left[1 - S_t\right] \cdot \EE\!\left[(1-S_t) u_t y_{t+h}\right] - \EE\!\left[(1-S_t) u_t \right] \cdot \EE\!\left[(1-S_t) y_{t+h}\right]}{\EE\!\left[1-S_t\right] \cdot \EE\!\left[(1-S_t) u_t^2\right] - \EE\!\left[(1-S_t) u_t\right]^2}.
\label{eq:beta_minus_normal_eq}
\ee

To simplify \eqref{eq:beta_plus_normal_eq}--\eqref{eq:beta_minus_normal_eq}, note that the structural function of QAR(1,1) implies the $h$-step-ahead outcome can be written as a quadratic function of the time-$t$ innovation,
\be
y_{t+h}
=
R_{t,h}
+
\kappa_h u_t
+
a_h s_{t-1}u_t
+
q_h u_t^2,
\label{eq:y_decomp_quadratic_ut}
\ee
where $R_{t,h}$ is measurable with respect to $\sigma(\mathcal{F}_{t-1},u_{t+1},\ldots,u_{t+h})$ and therefore independent of $u_t$,
\[
\kappa_h=\sigma\phi_1^h,
\qquad
a_h=\sigma \phi_1^{h}\left(\gamma + 2\phi_2\frac{1-\phi_1^h}{1-\phi_1}\right),
\qquad
q_h=\phi_2 \sigma^2 \frac{\phi_1^{h-1}-\phi_1^{2h-1}}{1-\phi_1}.
\]
Since $u_t\perp (R_{t,h},s_{t-1},W_t)$ and $\EE[s_{t-1}]=0$, we obtain
\bea
\EE[S_ty_{t+h}]
&=&
\EE[S_t]\EE[R_{t,h}]
+\kappa_h\EE[S_tu_t]
+q_h\EE[S_tu_t^2],
\nonumber\\
\EE[S_tu_ty_{t+h}]
&=&
\EE[S_tu_t]\EE[R_{t,h}]
+\kappa_h\EE[S_tu_t^2]
+q_h\EE[S_tu_t^3].
\nonumber
\eea
Plugging these expressions into \eqref{eq:beta_plus_normal_eq} shows that all terms involving $\EE[R_{t,h}]$ cancel and
\be
\beta_h^{(+)}
=
\kappa_h
+
q_h\cdot
\frac{\EE[S_t]\EE[S_tu_t^3]-\EE[S_tu_t]\EE[S_tu_t^2]}{\EE[S_t]\EE[S_tu_t^2]-\EE[S_tu_t]^2}.
\label{eq:beta_plus_simplify_ratio}
\ee
Because $u_t\sim\mathcal{N}(0,1)$ and $S_t=\mathbbm{1}\{u_t>0\}$, the truncated moments are
\[
\EE[S_t]=\tfrac12,
\qquad
\EE[S_tu_t]=\frac{1}{\sqrt{2\pi}},
\qquad
\EE[S_tu_t^2]=\tfrac12,
\qquad
\EE[S_tu_t^3]=\sqrt{\frac{2}{\pi}}.
\]
Substituting into \eqref{eq:beta_plus_simplify_ratio} yields
\[
\frac{\EE[S_t]\EE[S_tu_t^3]-\EE[S_tu_t]\EE[S_tu_t^2]}{\EE[S_t]\EE[S_tu_t^2]-\EE[S_tu_t]^2}
=
\frac{\sqrt{2/\pi}}{1-2/\pi}
\equiv m,
\]
so $\beta_h^{(+)}=\kappa_h+m q_h$, which is \eqref{eq:beta_asym_simplified}.

The expression for $\beta_h^{(-)}$ follows analogously from \eqref{eq:beta_minus_normal_eq} (or by symmetry), yielding
$\beta_h^{(-)}=\kappa_h-m q_h$.
\end{proof}
\subsection{Proof of Proposition \ref{prop:irf.lag.state}}
\begin{proof}
Denote $V_t = \left(u_t, W_t', y_{t-1}u_t, y_{t-1}W_t^{\prime}\right)^{\prime}$, and $\gamma_h = \left(\beta_h^{(0)}, {\pi_{h}^{(0)}}^{\prime}, \beta_h^{(1)}, {\pi_{h}^{(1)}}^{\prime}\right)^{\prime}.$ The population coefficients are
\begin{equation*}
\begin{aligned}
\gamma_h & = \mathbb{E}[V_t V_t^{\prime}]^{-1} \mathbb{E}[V_t y_{t+h}] \\
& = \mathbb{E}\left[\begin{pmatrix}
    u_t^2 & u_t W_t' & u_t^2 y_{t-1}  & u_t y_{t-1}W_t^{\prime} \\
    W_t u_t & W_t W_t' & W_t y_{t-1}u_t & W_t y_{t-1}W_t^{\prime}\\
    y_{t-1}u_t^2 & y_{t-1}u_tW_t' & y_{t-1}^2 u_t^2 & y_{t-1}^2 u_t W_t^{\prime}\\
    y_{t-1}W_t u_t & y_{t-1}W_t W_t' & y_{t-1}^2 W_tu_t & y_{t-1}^2 W_t W_t^{\prime}
\end{pmatrix}\right]^{-1} \mathbb{E}\left[\begin{pmatrix} u_t y_{t+h} \\
W_t^{\prime} y_{t+h}\\
y_{t-1} u_t y_{t+h}\\
y_{t-1} W_t^{\prime} y_{t+h}
\end{pmatrix}\right] \\
& = \begin{pmatrix}
    1 & 0 & \mathbb{E}[y_{t-1}]  & 0 \\
    0 & \mathbb{E}[W_t W_t'] & 0 & \mathbb{E}[W_t y_{t-1}W_t^{\prime}]\\
    \mathbb{E}[y_{t-1}] & 0 & \mathbb{E}[y_{t-1}^2] & 0\\
    0 & \mathbb{E}[y_{t-1}W_t W_t'] & 0 & \mathbb{E}[y_{t-1}^2 W_t W_t^{\prime}]
\end{pmatrix}^{-1} \begin{pmatrix} \mathbb{E}[u_t y_{t+h}] \\
\mathbb{E}[W_t^{\prime} y_{t+h}]\\
\mathbb{E}[y_{t-1} u_t y_{t+h}]\\
\mathbb{E}[y_{t-1} W_t^{\prime} y_{t+h}]
\end{pmatrix}
\end{aligned}
\end{equation*}
Then $\beta_h^{(0)}$ and $\beta_h^{(1)}$ can be found by solving
\begin{equation}
\begin{cases}
\mathbb{E}[u_{t} y_{t+h}] = \beta_h^{(0)} + \beta_h^{(1)} \mathbb{E}[y_{t-1}], \\
\mathbb{E}[y_{t-1} u_{t} y_{t+h}] = \beta_h^{(0)} \mathbb{E}[y_{t-1}] + \beta_h^{(1)} \mathbb{E}[y_{t-1}^2].
\end{cases}
\label{eq:laglp.system}
\end{equation}
Note that $ \mathbb{E}[y_{t-1}] = \frac{\phi_2\sigma^2}{(1-\phi_1)(1-\phi_1^2)}$ and $\mathbb{E}[y_{t+h}u_t] = \sigma\phi_1^h$. Also, for $h\ge 1$,
\begin{equation*}
\begin{aligned}
    y_{t+h}y_{t-1} &=\phi_1^{h+1} y_{t-1}^2 + \phi_1^{h}\phi_2 s_{t-1}^2 y_{t-1} + \sigma \phi_1^{h} y_{t-1} u_t + \gamma\sigma \phi_1^h s_{t-1} y_{t-1} u_t \\
    & + \phi_2 \sigma^2 \sum_{k=0}^{h-1}\phi_1^{h-k-1} y_{t-1} \left(\phi_1^{2k} u_{t^2} +  \left(\sum_{\substack{j=0 \\ j\neq k}}^{\infty} \phi_1^j u_{t+k-j}\right)^2 + 2 \phi_1^k u_{t} \sum_{\substack{j=0 \\ j\neq k}}^{\infty} \phi_1^j u_{t+k-j}\right) \\
    & + \sigma \sum_{k=0}^{h-1} \phi_1^{h-k-1} y_{t-1} u_{t+1+k} + \gamma\sigma^2 \sum_{k=0}^{h-1} \phi_1^{h-1-k} y_{t-1} \left(\phi_1^k u_t + \sum_{\substack{j=0 \\ j\neq k}}^{\infty} \phi_1^j u_{t+k-j}\right) u_{t+1+k}.
\end{aligned}
\end{equation*}
Thus 
\begin{equation*}
\begin{aligned}
\mathbb{E}[y_{t+h} y_{t-1} u_{t}] & = \sigma \phi_1^{h} \mathbb{E}[y_{t-1}] + \gamma\sigma \phi_1^h \mathbb{E}[s_{t-1} y_{t-1}] + 2 \phi_2 \sigma^2 \sum_{k=0}^{h-1}\phi_1^{h-1} \mathbb{E}\left[y_{t-1}  u_{t^2} \sum_{\substack{j=0 \\ j\neq k}}^{\infty} \phi_1^j u_{t+k-j} \right] \\
& =  \sigma \phi_1^{h} \frac{\phi_2\sigma^2}{(1-\phi_1)(1-\phi_1^2)} + \gamma \sigma\phi_1^h \frac{\sigma^2}{1-\phi_1^2} + 2 \phi_2 \sigma \phi_1^{h}\frac{1-\phi_1^{h}}{1-\phi_1}\frac{\sigma^2}{1-\phi_1^2}\\
& =\frac{\sigma^3 \phi_1^h}{\left(1-\phi_1\right)\left(1-\phi_1^2\right)}\left[\phi_2\left(3-2 \phi_1^h\right)+\gamma\left(1-\phi_1\right)\right] .
\end{aligned}
\end{equation*}
For $h=0$, 
\begin{equation*}
\mathbb{E}[y_{t} y_{t-1} u_{t}] = \sigma \mathbb{E}[y_{t-1}] + \gamma \sigma \mathbb{E}[s_{t-1} y_{t-1}] = \frac{\sigma^3}{\left(1-\phi_1\right)\left(1-\phi_1^2\right)}\left[\phi_2+\gamma\left(1-\phi_1\right)\right].
\end{equation*}

Besides,
\begin{equation*}
\begin{aligned}
\Var(y_t) & =\frac{1}{1-\phi_1^2}\left[\phi_2^2 \Var\left(s_t^2\right)+\sigma^2\left(1+\gamma^2 \mathbb{E}\left[s_t^2\right]\right)+2 \phi_1 \phi_2 \operatorname{Cov}\left(y_t, s_t^2\right)\right]\\
& = \frac{1}{1-\phi_1^2}\left[\phi_2^2 \frac{2\sigma^4}{(1-\phi_1^2)^2}+\sigma^2\left(1+\gamma^2 \frac{\sigma^2}{1-\phi_1^2}\right)+2 \phi_1 \phi_2 \frac{1}{1-\phi_1^3}\left(\phi_1^2 \phi_2 \frac{2\sigma^4}{(1-\phi_1^2)^2} +2 \phi_1 \gamma\frac{\sigma^4}{1-\phi_1^2}\right) \right]\\
&= \frac{2 \phi_2^2 \sigma^4}{\left(1-\phi_1^2\right)^3}+\frac{\sigma^2}{1-\phi_1^2}+\frac{\gamma^2 \sigma^4}{\left(1-\phi_1^2\right)^2}+\frac{4 \phi_1^3 \phi_2^2 \sigma^4}{\left(1-\phi_1^2\right)^3\left(1-\phi_1^3\right)}+\frac{4 \phi_1^2 \phi_2 \gamma \sigma^4}{\left(1-\phi_1^2\right)^2\left(1-\phi_1^3\right)}.
\end{aligned}
\end{equation*}

Solving the system (\ref{eq:laglp.system}),
\[
\begin{aligned}
\beta_h^{(1)}
& = \frac{\mathbb{E}[y_{t+h} y_{t-1} u_{t}] -\sigma \phi_1^h \cdot \EE[y_{t-1}]}{\Var\left(y_{t-1}\right)},\\
& =
\sigma \phi_1^{h}\left(\gamma + 2\phi_2\frac{1-\phi_1^h}{1-\phi_1}\right)\cdot \frac{\sigma^2/(1-\phi_1^2)}{\Var\left(y_{t-1}\right)}.
\end{aligned}
\]
which is \eqref{eq:beta_lag_simplified}. Finally, using the first normal equation
$\EE[u_ty_{t+h}]=\beta_h^{(0)}+\beta_h^{(1)}\EE[y_{t-1}]$ gives
\[
\beta_h^{(0)}=\sigma\phi_1^h-\beta_h^{(1)}\EE[y_{t-1}],
\]
so the implied population IRF is $\operatorname{IRF}^{LagLP}(y,\delta;h)=(\beta_h^{(0)}+\beta_h^{(1)}y)\delta$.
\end{proof}

\subsection{Proof of Proposition \ref{prop:irf.infeas}}
\begin{proof}
To back out the coefficients $\kappa_{h}^0$ to $\kappa_{h}^3$, recall that we can represent $y_{t+h}$ as
\begin{equation*}
\begin{aligned}
    y_{t+h} 
    &=\phi_1^{h+1} y_{t-1} + \phi_1^{h}\phi_2 s_{t-1^2} + \sigma \phi_1^{h} u_t + \gamma\sigma \phi_1^h  s_{t-1}u_t \\
    & + \phi_2 \sigma^2 \sum_{k=0}^{h-1}\phi_1^{h-k-1} \left(\phi_1^{2k} u_{t^2} +  \left(\phi_1^{k+1} \frac{s_{t-1}}{\sigma}+\sum_{j=0}^{k-1} \phi_1^j u_{t+k-j}\right)^2 + 2 \phi_1^k u_{t} \left(\phi_1^{k+1} \frac{ s_{t-1}}{\sigma}+\sum_{j=0}^{k-1} \phi_1^j u_{t+k-j}\right)\right) \\
    & + \sigma \sum_{k=0}^{h-1} \phi_1^{h-k-1} u_{t+1+k} + \gamma\sigma^2 \sum_{k=0}^{h-1} \phi_1^{h-1-k} \left(\phi_1^k u_t + \phi_1^{k+1} \frac{s_{t-1}}{\sigma}+\sum_{j=0}^{k-1} \phi_1^j u_{t+k-j} \right) u_{t+1+k}.
\end{aligned}
\end{equation*}

Matching the coefficients for terms $u_t$, $s_{t-1}u_t$, and $u_t^2$, and collecting the rest of terms (denoted as $r_{h,t+h}$) in the intercept $\alpha_{h0} = \mathbb{E}[r_{h,t+h}]$ and residual $\epsilon_{h,t+h} = r_{h,t+h} - \mathbb{E}[r_{h,t+h}]$, we obtain
\begin{equation*}
\begin{aligned}
    \kappa_{h1} &= \sigma \phi_1^{h}, \\
    \kappa_{h2} &= \gamma \sigma \phi_1^h 
    + 2 \phi_2 \sigma \phi_1^h \cdot \frac{1 - \phi_1^h}{1 - \phi_1}, \\
    \kappa_{h3} &= 
    \begin{cases}
        \phi_2 \sigma^2 \phi_1^{h-1} \cdot \dfrac{1 - \phi_1^h}{1 - \phi_1}, & \text{if } h \geq 1, \\[6pt]
        0, & \text{if } h = 0.
    \end{cases}
\end{aligned}
\end{equation*}
The matched coefficients coincide with the population coefficients under \textit{Infeas} because $\mathbb{E}[\epsilon_{h,t+h} u_t]= \mathbb{E}[\epsilon_{h,t+h} s_{t-1}u_t]=  \mathbb{E}[\epsilon_{h,t+h} u_t^2] = 0$.
It can also be easily verified that $\operatorname{CAR}(s, \delta) = \kappa_{h1} \delta + \kappa_{h2} s \delta + \kappa_{h3} \delta^2.$ Therefore, the infeasible specification (\ref{infeasible.spec}) recovers the true CAR exactly.
\end{proof}

\subsection{Proof of Proposition \ref{prop:irf.feas}}
\begin{proof}
Denote $R_t = \left(1, u_t, y_{t-1}u_t, u_t^2, W_t^{\prime}\right)^{\prime}$ and $\theta_h  = \left(\theta_{h0}, \theta_{h1}, \theta_{h2}, \theta_{h3}\right)^{\prime}.$ Then population coefficients 
$$
\begin{aligned}
    \left(\theta_h, \pi_h^{\prime}\right)^{\prime}  & =  \mathbb{E}[R_tR_t^{\prime}]^{-1} \mathbb{E}[R_t y_{t+h}]\\
    & = \begin{pmatrix}
        1 & 0 & 0 & 1 & \mathbb{E}[W_t]\\
        0 & 1 & \mathbb{E}[y_{t-1}] & 0 & 0 \\
        0 & \mathbb{E}[y_{t-1}] & \mathbb{E}[y_{t-1}^2] & 0 & 0\\
        1 & 0 & 0 & 3  & \mathbb{E}[W_t]\\
        \mathbb{E}[W_t] & 0  & 0  & \mathbb{E}[W_t] & \mathbb{E}[W_tW_t^{\prime}]
    \end{pmatrix}^{-1} \begin{pmatrix}
        \mathbb{E}\left[y_{t+h}\right]\\
        \mathbb{E}\left[u_t y_{t+h}\right] \\
        \mathbb{E}\left[u_t y_{t-1} y_{t+h}\right]\\
        \mathbb{E}\left[u_t^2 y_{t+h}\right] \\
        \mathbb{E}\left[W_t\right] 
    \end{pmatrix} \\
\end{aligned}
$$

$\theta_{h1}$ to $\theta_{h3}$ can be found by solving
\be
\left\{
\begin{aligned}
&\theta_{h0} + \theta_{h3} + \pi_h^{\prime} \mathbb{E}[W_t] = \mathbb{E}\left[y_{t+h}\right],\\
&\theta_{h1} + \theta_{h2} \mathbb{E}[y_{t-1}]  =  \mathbb{E}\left[u_t y_{t+h}\right], \\
&\theta_{h1} \mathbb{E}[y_{t-1}] + \theta_{h2} \mathbb{E}[y_{t-1}^2]  =  \mathbb{E}\left[u_t y_{t-1} y_{t+h}\right],\\
&\theta_{h0} + 3 \theta_{h3}  + \pi_h^{\prime} \mathbb{E}[W_t] = \mathbb{E}\left[u_t^2 y_{t+h}\right].\\
\end{aligned}
\right.
\label{eq:feas.system}
\ee

Note that the linear system determining $(\theta_{h1},\theta_{h2})$ is identical to the one determining $(\beta_h^{(0)},\beta_h^{(1)})$ in the proof of Proposition~\ref{prop:irf.lag.state}. Then $\theta_{h2}=\beta_h^{(1)}$ and $\theta_{h1}=\beta_h^{(0)}$. 

To solve for $\theta_{h3}$, the first and last equations in \ref{eq:feas.system} imply
$\theta_{h3}  = \frac{1}{2}\mathbb{E}\left[u_t^2 y_{t+h}\right] - \frac{1}{2}\mathbb{E}\left[y_{t+h}\right].$
Using the same decomposition of $y_{t+h}$ in terms of the time-$t$ innovation as in \eqref{eq:y_decomp_quadratic_ut}:
\[
y_{t+h}=R_{t,h}+\kappa_h u_t + a_h s_{t-1}u_t + q_h u_t^2,
\qquad u_t\perp (R_{t,h},s_{t-1}),
\]
where $q_h$ is the coefficient on the second-order term in Proposition~\ref{prop:true.car}.
Using $\EE[u_t^2]=1$, $\EE[u_t^3]=0$, and $\EE[u_t^4]=3$, we obtain
\[
\EE[u_t^2 y_{t+h}]
=
\EE[R_{t,h}]
+3q_h,
\qquad
\EE[y_{t+h}]
=
\EE[R_{t,h}]
+q_h,
\]
so $\theta_{h3}=\frac12\{\EE[u_t^2y_{t+h}]-\EE[y_{t+h}]\}=q_h$.

Therefore $\operatorname{IRF}^{Feas}(y,\delta;h)=\theta_{h1}\delta+\theta_{h2}y\delta+\theta_{h3}\delta^2$ with $(\theta_{h1},\theta_{h2},\theta_{h3})=(\beta_h^{(0)},\beta_h^{(1)},q_h)$, which is Proposition~\ref{prop:irf.feas}.
\end{proof}

\subsection{Lemma \ref{lem:align_y_to_s_main}}
\begin{lemma}[A projection identity]
\label{lem:align_y_to_s_main}
Under Assumption~\ref{assump:qar.dgp},
\be
\Proj(y_{t-1}\mid s_{t-1})=\EE[y_{t-1}] + s_{t-1}.
\label{eq:proj_y_on_s_main}
\ee
\end{lemma}

\begin{oldproof}
By stationarity it suffices to work with $(y_t,s_t)$.
Since $\EE[s_t]=0$, the population projection is
\[
\Proj(y_t\mid s_t)=\EE[y_t]+\frac{\Cov(y_t,s_t)}{\Var(s_t)}\,s_t.
\]

We now show that $\Cov(y_t,s_t)=\Var(s_t)$. Let $C\equiv \EE[s_ty_t]$ (so $C=\Cov(y_t,s_t)$ because $\EE[s_t]=0$).
Using the QAR(1,1) recursions \eqref{QAR.hetero} and the independence of $u_t$ from $\mathcal{F}_{t-1}$,
\bea
C=\EE[s_ty_t]
&=&
\EE\big[(\phi_1 s_{t-1}+\sigma u_t)\,(\phi_1 y_{t-1}+\phi_2 s_{t-1}^2+\sigma u_t+\gamma\sigma s_{t-1}u_t)\big]
\nonumber\\
&=&
\phi_1^2\,\EE[s_{t-1}y_{t-1}]
+
\sigma^2\,\EE[u_t^2].
\label{eq:C_recursion}
\eea
Since $\EE[u_t^2]=1$ and $\EE[s_{t-1}y_{t-1}]=C$ by stationarity, \eqref{eq:C_recursion} implies
$C=\phi_1^2 C+\sigma^2$, hence $C=\sigma^2/(1-\phi_1^2)$.

Finally, $\Var(s_t)=\sigma^2/(1-\phi_1^2)$ for the AR(1) state equation, so $\Cov(y_t,s_t)=\Var(s_t)$ and the projection slope equals one.
This proves \eqref{eq:proj_y_on_s_main}.
\end{oldproof}

\subsection{Proof of Theorem \ref{thm:cond_u_main}}
\begin{proof}
By Proposition~\ref{prop:true.car}, the CAR at horizon $h$ can be written as
\[
\operatorname{CAR}_h(s,\delta)=\kappa_h\delta+a_h s\,\delta+q_h\delta^2,
\]
with $\kappa_h=\sigma\phi_1^h$ and $a_h,q_h$ as defined in \eqref{eq:ah_qh_def}.
We compute $\mathcal{L}^{spec}_h(\delta)$ by substituting the corresponding population IRF and using that
$s_{t-1}$ is independent of $u_t$ and has mean zero.

\noindent \textit{Linear.} Since $\operatorname{IRF}^{Linear}(\delta;h)=\kappa_h\delta$,
\[
\operatorname{CAR}_h(s_{t-1},\delta)-\operatorname{IRF}^{Linear}(\delta;h)
=
a_h s_{t-1}\delta+q_h\delta^2,
\]
so, conditioning on $u_t=\delta$,
\[
\mathcal{L}^{Linear}_h(\delta)=a_h^2\delta^2\sigma_s^2+q_h^2\delta^4.
\]

\noindent \textit{LagLP and Feas.} Recall that $\beta_h^{(0)}=\sigma\phi_1^{h}-\beta_h^{(1)}\EE[y_{t-1}]$. Then by Proposition~\ref{prop:irf.lag.state},
\[
\operatorname{IRF}^{LagLP}(y,\delta;h)=(\beta_h^{(0)} + \beta_h^{(1)} y)\delta=\kappa_h\delta+\beta_h^{(1)}(y-\EE[y_{t-1}])\delta,
\]
and by Proposition~\ref{prop:irf.feas},
\[
\operatorname{IRF}^{Feas}(y,\delta;h)=\beta_h^{(0)} \delta+\beta_h^{(1)} y\delta+q_h \delta^2=\kappa_h\delta+\beta_h^{(1)}(y-\EE[y_{t-1}])\delta+q_h\delta^2.
\]
Using the closed-form coefficient in Proposition~\ref{prop:irf.lag.state} and the fact that $\Cov(s_{t-1},y_{t-1})=\sigma^2/(1-\phi_1^2)=\Var(s_{t-1})$ (see Lemma~\ref{lem:align_y_to_s_main}), we can write
$\beta_h^{(1)}=a_h\lambda$ where $\lambda=\Cov(s_{t-1},y_{t-1})/\Var(y_{t-1}).$
Therefore,
\[
\operatorname{CAR}_h(s_{t-1},\delta)-\operatorname{IRF}^{LagLP}(y_{t-1},\delta;h)
=
a_h\big(s_{t-1}-\lambda(y_{t-1}-\EE[y_{t-1}])\big)\delta
+
q_h\delta^2,
\]
and
\[
\operatorname{CAR}_h(s_{t-1},\delta)-\operatorname{IRF}^{Feas}(y_{t-1},\delta;h)
=
a_h\big(s_{t-1}-\lambda(y_{t-1}-\EE[y_{t-1}])\big)\delta.
\]
Conditioning on $u_t=\delta$ and using $\EE[s_{t-1}-\lambda(y_{t-1}-\EE[y_{t-1}])]=0$ gives
$$
\mathcal{L}^{LagLP}_h(\delta)=a_h^2\delta^2\sigma^2_{s|y} + q_h^2\delta^4, \qquad \mathcal{L}^{Feas}_h(\delta)=a_h^2\delta^2\sigma^2_{s|y}.
$$

\noindent \textit{AsymLP.} If $S=\mathbbm{1}\{\delta>0\}$, then Proposition~\ref{prop:irf.shock.state} implies
$\operatorname{IRF}^{AsymLP}(S,\delta;h)=\kappa_h\delta+m q_h|\delta|$.
Hence
\[
\operatorname{CAR}_h(s_{t-1},\delta)-\operatorname{IRF}^{AsymLP}(S_t,\delta;h)
=
a_h s_{t-1}\delta
+
q_h(\delta^2-m|\delta|),
\]
and conditioning on $u_t=\delta$ yields
$$
\mathcal{L}^{AsymLP}_h(\delta)=a_h^2\delta^2\sigma_s^2 + q_h^2(\delta^2-m|\delta|)^2.
$$

The rank ordering follows immediately because $\sigma^2_{s|y}\le \sigma^2_{s}$ and $q_h^2\delta^4\ge 0$.
Finally,
\[
\mathcal{L}^{Linear}_h(\delta)-\mathcal{L}^{LagLP}_h(\delta)
=
a_h^2\delta^2\big(\sigma^2_{s}-\sigma^2_{s|y}\big)
=
a_h^2\delta^2\cdot\frac{\sigma_{sy}^2}{\sigma^2_{y}},
\]
and
\[
\mathcal{L}^{Linear}_h(\delta)-\mathcal{L}^{AsymLP}_h(\delta)
=
q_h^2\Big(\delta^4-(\delta^2-m|\delta|)^2\Big)
=
q_h^2\big(2m|\delta|^3-m^2\delta^2\big).
\]
\end{proof}

\subsection{Proof of Theorem \ref{thm:cond_s_main}}
\begin{proof}
By Proposition~\ref{prop:true.car}, for a fixed horizon $h$,
\[
\operatorname{CAR}_h(s,u_t)=\kappa_h u_t+a_h s\,u_t+q_h u_t^2,
\]
with $\kappa_h=\sigma\phi_1^h$ and $a_h,q_h$ as in \eqref{eq:ah_qh_def}.
We compute each $\mathcal{R}^{spec}_h(s)$ by substituting the corresponding population IRF.
Under Assumption~\ref{assump:qar.dgp}, $u_t\sim\mathcal{N}(0,1)$ is independent of $(s_{t-1},y_{t-1})$.

\noindent \textit{Linear.}
Since $\operatorname{IRF}^{Linear}(u_t;h)=\kappa_h u_t$,
\[
\operatorname{CAR}_h(s,u_t)-\operatorname{IRF}^{Linear}(u_t;h)=a_h s\,u_t+q_h u_t^2.
\]
Conditioning on $s_{t-1}=s$ and using $\EE[u_t^2]=1$, $\EE[u_t^3]=0$, and $\EE[u_t^4]=3$ gives
\[
\mathcal{R}^{Linear}_h(s)=a_h^2 s^2\EE[u_t^2]+q_h^2\EE[u_t^4]=a_h^2 s^2+3q_h^2.
\]

\noindent \textit{LagLP and Feas.}
By Proposition~\ref{prop:irf.lag.state},
\[
\operatorname{IRF}^{LagLP}(y_{t-1},u_t;h)
=
\kappa_h u_t+\beta_h^{(1)}(y_{t-1}-\EE[y_{t-1}])u_t,
\]
and by Proposition~\ref{prop:irf.feas},
\[
\operatorname{IRF}^{Feas}(y_{t-1},u_t;h)
=
\kappa_h u_t+\beta_h^{(1)}(y_{t-1}-\EE[y_{t-1}])u_t+q_h u_t^2.
\]
Using the closed-form coefficient in Proposition~\ref{prop:irf.lag.state} and Lemma~\ref{lem:align_y_to_s_main},
we can write $\beta_h^{(1)}=a_h\lambda$ with $\lambda=\Cov(s_{t-1},y_{t-1})/\Var(y_{t-1})$.
Therefore,
\[
\operatorname{CAR}_h(s,u_t)-\operatorname{IRF}^{LagLP}(y_{t-1},u_t;h)
=
a_h\Big(s-\lambda(y_{t-1}-\EE[y_{t-1}])\Big)u_t+q_h u_t^2,
\]
and
\[
\operatorname{CAR}_h(s,u_t)-\operatorname{IRF}^{Feas}(y_{t-1},u_t;h)
=
a_h\Big(s-\lambda(y_{t-1}-\EE[y_{t-1}])\Big)u_t.
\]
Conditioning on $s_{t-1}=s$, using independence of $u_t$ from $(s_{t-1},y_{t-1})$, and $\EE[u_t^3]=0$ yields
\[
\mathcal{R}^{LagLP}_h(s)
=
a_h^2\,\EE\!\left[\Big(s-\lambda(y_{t-1}-\EE[y_{t-1}])\Big)^2\mid s_{t-1}=s\right]\EE[u_t^2]
+
q_h^2\EE[u_t^4]
=
a_h^2\Xi(s)+3q_h^2,
\]
and
\[
\mathcal{R}^{Feas}_h(s)
=
a_h^2\,\EE\!\left[\Big(s-\lambda(y_{t-1}-\EE[y_{t-1}])\Big)^2\mid s_{t-1}=s\right]\EE[u_t^2]
=
a_h^2\Xi(s).
\]

\noindent \textit{AsymLP.}
Proposition~\ref{prop:irf.shock.state} implies
$\operatorname{IRF}^{AsymLP}(S_t,u_t;h)=\kappa_h u_t+m q_h|u_t|$.
Hence
\[
\operatorname{CAR}_h(s,u_t)-\operatorname{IRF}^{AsymLP}(S_t,u_t;h)
=
a_h s\,u_t+q_h(u_t^2-m|u_t|).
\]
The cross term satisfies $\EE[u_t(u_t^2-m|u_t|)]=0$ by symmetry, so conditioning on $s_{t-1}=s$,
\[
\mathcal{R}^{AsymLP}_h(s)
=
a_h^2 s^2\EE[u_t^2]+q_h^2\EE[(u_t^2-m|u_t|)^2]
=
a_h^2 s^2+\nu_m q_h^2.
\]

The inequalities in \eqref{eq:rank_s_main} are immediate: $\mathcal{R}^{LagLP}_h(s)-\mathcal{R}^{Feas}_h(s)=3q_h^2\ge 0$,
and $\mathcal{R}^{Linear}_h(s)-\mathcal{R}^{AsymLP}_h(s)=(3-\nu_m)q_h^2\ge 0$ because
$\nu_m=\EE[(u_t^2-m|u_t|)^2]\le \EE[u_t^4]=3$.
The gap expressions in \eqref{eq:gains_s_main} follow from the formulas above.
\end{proof}

\subsection{Proof of Proposition \ref{prop:qvar.true.car}}
\begin{proof}

Using the formula of structural function $\tilde{\psi}_h$, it is convenient to work first with the vector-valued decomposition
\[
{
\begin{aligned}
\mathcal C_h^{(i)}(\mathcal{F}, \delta_i)
& \equiv \underbrace{\mathbb{E}[L_{t,h}(u_t+\delta_i e_i)-L_{t,h}(u_t)| \mathcal{F}_{t-1} = \mathcal{F}]}_{(1)}
+ \underbrace{\mathbb{E}[C_{t,h}(u_t+\delta_i e_i,\cdot)-C_{t,h}(u_t,\cdot)| \mathcal{F}_{t-1} = \mathcal{F}]}_{(2)} \\
&\quad + \underbrace{\mathbb{E}[Q_{t,h}(u_t+\delta_i e_i)-Q_{t,h}(u_t)| \mathcal{F}_{t-1} = \mathcal{F}]}_{(3)},
\end{aligned}
}
\]
so that $\operatorname{CAR}^{(j,i)}_h(\mathcal{F}, \delta_i)=e_j'\mathcal C_h^{(i)}(\mathcal{F}, \delta_i)$.

For $h\ge 1$, term (1):
\[
(1)
= \Phi_1^h (1_n +\mathsf{G}\,s)\odot(\delta_i b_i)
+ \sum_{k=1}^h\Phi_1^{h-k}\,\Phi_2\,\vech\bigl(\delta_i(\Phi_1^k s)(\Phi_1^{k-1}b_i)' + \delta_i(\Phi_1^{k-1}b_i)(\Phi_1^k s)'\bigr),
\]
where $s$ is the realized value of $s_{t-1}$.

Term (2):
\[
\begin{aligned}
(2)
& = \sum_{k=1}^h \Phi_1^{h-k} \Phi_2 \sum_{j=1}^{k-1} \vech\left( \mathbb{E}\left[\delta_i(\Phi_1^{k-1}b_i) \left(\Phi_1^{k-1-j}Bu_{t+j}\right)' + \left(\Phi_1^{k-1-j}Bu_{t+j}\right) \delta_i(\Phi_1^{k-1}b_i)' \,\Big|\, \mathcal{F}_{t-1} = \mathcal{F}\right] \right) \\
& \quad + \sum_{k=1}^h \Phi_1^{h-k} \mathbb{E}\left[\left(\mathsf{G} \Phi_1^{k-1}\delta_i b_i\right)\odot \left(B u_{t+k}\right)\,\Big|\, \mathcal{F}_{t-1} = \mathcal{F}\right] = 0,
\end{aligned}
\]
as future shocks are independent of $\mathcal{F}_{t-1}$.

Term (3):
\[
\begin{aligned}
(3)
&= \sum_{k=1}^h \Phi_1^{h-k}\Phi_2\,
\mathbb{E}\!\left[\vech\!\left(
(\Phi_1^{k-1}Bu_t+\delta_i\Phi_1^{k-1}b_i)
(\Phi_1^{k-1}Bu_t+\delta_i\Phi_1^{k-1}b_i)'
-(\Phi_1^{k-1}Bu_t)(\Phi_1^{k-1}Bu_t)'
\right)\middle| \mathcal{F}_{t-1}=\mathcal{F}\right]\\
&= \sum_{k=1}^h \Phi_1^{h-k}\Phi_2 \vech\Bigl(\mathbb{E}\bigl[
\delta_i(\Phi_1^{k-1}b_i)(\Phi_1^{k-1}Bu_t)'
+ \delta_i(\Phi_1^{k-1}Bu_t)(\Phi_1^{k-1}b_i)'
+ \delta_i^2(\Phi_1^{k-1}b_i)(\Phi_1^{k-1}b_i)'
\mid \mathcal{F}_{t-1}=\mathcal{F}\bigr]\!\Bigr) \\
& = \sum_{k=1}^h\Phi_1^{h-k}\Phi_2\,\vech\Bigl(\delta_i^2 (\Phi_1^{k-1}b_i)(\Phi_1^{k-1}b_i)'\Bigr).
\end{aligned}
\]

Therefore,
\[
\begin{aligned}
\operatorname{CAR}^{(j,i)}_h(\mathcal{F}, \delta_i)
& = e_j'\Bigg[
\Phi_1^h (1_n +\mathsf{G}\,s)\odot(\delta_i b_i)  + \sum_{k=1}^h\Phi_1^{h-k}\,\Phi_2\,\vech\bigl(\delta_i(\Phi_1^k s)(\Phi_1^{k-1} b_i)' + \delta_i(\Phi_1^{k-1} b_i)(\Phi_1^k s)'\bigr) \\
& \qquad\qquad + \sum_{k=1}^h\Phi_1^{h-k}\Phi_2\,\vech\Bigl(\delta_i^2 (\Phi_1^{k-1} b_i)(\Phi_1^{k-1} b_i)'\Bigr)
\Bigg].
\end{aligned}
\]
For $h=0$, one recovers
\[
\operatorname{CAR}^{(j,i)}_0(\mathcal{F}, \delta_i)
= e_j'\Big[(1_n +\mathsf{G}\,s)\odot(\delta_i b_i)\Big].
\]
\end{proof}

\subsection{Proof of Proposition \ref{prop:multi.failure.lp}}
\begin{proof}
Using the formula of structural function 
\[
y_{t+h} = \tilde{\psi}_h(u_t,U_{h,t+h}),
\]
the linear LP population coefficient vector
\(
\boldsymbol{\beta}_h^{(i)}=(\beta_h^{(1,i)},\dots,\beta_h^{(n,i)})'
\)
for the \(i\)-th shock is:
for \(h\ge1\),
\begin{equation*}
\begin{aligned}
\boldsymbol{\beta}_h^{(i)}
&= \mathbb{E}[y_{t+h}u_{it}]\;\Var(u_{it})^{-1}
= \mathbb{E}[y_{t+h}u_{it}]\\
&= \mathbb{E}\bigl[(R_{t,h}+L_{t,h}(u_t)
+ C_{t,h}(u_t,u_{t+1:t+h})
+ Q_{t,h}(u_t))\,u_{it}\bigr]\\
&= \underbrace{\mathbb{E}[R_{t,h}u_{it}]}_{(a)}
+ \underbrace{\mathbb{E}[L_{t,h}(u_t)u_{it}]}_{(b)}
+ \underbrace{\mathbb{E}[C_{t,h}(u_t,u_{t+1:t+h})u_{it}]}_{(c)}
+ \underbrace{\mathbb{E}[Q_{t,h}(u_t)u_{it}]}_{(d)}.
\end{aligned}
\end{equation*}

We now compute terms (a)--(d):

1. Term (a):  
Since \(R_{t,h}\) depends only on nuisance variables independent of \(u_{it}\),
\[
\mathbb{E}[R_{t,h}u_{it}]
= \mathbb{E}[R_{t,h}]\;\mathbb{E}[u_{it}]
= 0.
\]

2. Term (b):  
Recall
\[
L_{t,h}\left(u_t\right) = \Phi_1^h \left(1_n+\mathsf{G} s_{t-1}\right)\odot \left(B u_{t}\right)+ \sum_{k=1}^h \Phi_1^{h-k} \Phi_2 \vech\left( \left(\Phi_1^ks_{t-1}\right) \left(\Phi_1^{k-1}Bu_t\right)' + \left(\Phi_1^{k-1}Bu_t\right) \left(\Phi_1^ks_{t-1}\right)' \right).
\]
Condition on \(s_{t-1}\) and use \(u_t\perp s_{t-1}\):
\[
\begin{aligned}
\mathbb{E}[L_{t,h}(u_t)\,u_{it}\mid s_{t-1}]
&= \Phi_1^h \left(1_n+\mathsf{G} s_{t-1}\right)\odot b_i \\
&\quad + \sum_{k=1}^h \Phi_1^{h-k} \Phi_2 \vech\left( \left(\Phi_1^ks_{t-1}\right) \left(\Phi_1^{k-1}b_i\right)' + \left(\Phi_1^{k-1}b_i\right) \left(\Phi_1^ks_{t-1}\right)' \right).
\end{aligned}
\]
Taking expectations and using \(\mathbb{E}[s_{t-1}]=0\) gives
\[
\mathbb{E}[L_{t,h}(u_t)\,u_{it}]
= \Phi_1^h b_i
= \Phi_1^h B e_i.
\]

3. Term (c):
Each summand in \(C_{t,h}\) is of the form \(u_{l,t+j}\,u_{mt}\) for \(j\ge1\).  
Since \(u_{l,t+j}\) has zero mean and is independent of \(u_t\),
\[
\mathbb{E}[u_{l,t+j}u_{mt}u_{it}]
= \mathbb{E}[u_{l,t+j}]\;\mathbb{E}[u_{mt}u_{it}]
= 0,
\]
so \(\mathbb{E}[C_{t,h}u_{it}]=0\).

4. Term (d): 
Each summand in \(Q_{t,h}\) is of the form \(u_{lt}u_{mt}\), so multiplied by \(u_{it}\) it is
\(u_{lt}u_{mt}u_{it}\). For zero-mean Gaussian shocks,
\(\mathbb{E}[u_{lt}u_{mt}u_{it}]=0\),
hence \(\mathbb{E}[Q_{t,h}u_{it}]=0\).

Summing (a)--(d) yields for $h\ge 1$,
\[
\boldsymbol{\beta}_h^{(i)}
= \Phi_1^h b_i
= \Phi_1^h B e_i.
\]
For $h = 0$, the linear LP population coefficient is
\begin{equation*}
\begin{aligned}
    \boldsymbol{\beta}_0^{(i)} &= \mathbb{E}[y_{t} u_{it}]\Var(u_{it})^{-1}= \mathbb{E}[y_{t} u_{it}]\\
    & = \mathbb{E}\left[\left(\Phi_1y_{t-1}
  + \Phi_2\,\vech\bigl(s_{t-1}s_{t-1}'\bigr)
  + \bigl(1_n+\mathsf{G} s_{t-1}\bigr)\!\odot (B u_t)\right) u_{it} \right]\\
    & = \Phi_1 \mathbb{E}\left[y_{t-1}u_{it}\right] + \Phi_2 \mathbb{E}\left[\,\vech\bigl(s_{t-1}s_{t-1}'\bigr)u_{it} \right] + \mathbb{E}\left[\bigl(1_n+\mathsf{G} s_{t-1}\bigr)\!\odot \left(B u_t u_{it}\right) \right] \\
    & = b_i
    = B e_i.
\end{aligned}
\end{equation*}

Therefore, the linear LP coefficient is $\forall i, \boldsymbol{\beta}_h^{(i)} = \Phi_1^h b_i = \Phi_1^h B e_i \text{ for } h=0,1,\ldots,$ and the associated coefficient for the $j$-th variable is $\beta_h^{(j,i)} = e_j^{\prime} \Phi_1^h B e_i.$

When the underlying DGP is the VAR(1), the population coefficient is 
\begin{equation*}
\boldsymbol{\beta}_h^{(i)} = \mathbb{E}\left[y_{t+h} u_{it}\right] \Var(u_{it})^{-1} = \Phi_1^h B e_i,
\end{equation*}
which coincides with that of the QVAR(1,1).

\end{proof}

\subsection{Proof of Proposition \ref{prop:irf.multi.infeas}}
\begin{proof}
Based on the structural function, we can separate the contribution of $u_{it}$: 
$$
\begin{aligned}
y_{t+h} &= \Phi_1^h \bigl(1_n + \mathsf{G}\,s_{t-1}\bigr)\odot\bigl(\Sigma_{tr}e_i u_{it}\bigr)\\
&\quad  + \sum_{k=1}^h \Phi_1^{\,h-k}\,\Phi_2\,
    \vech\Bigl(
      \bigl(\Phi_1^k s_{t-1}\bigr)\bigl(\Phi_1^{\,k-1}\Sigma_{tr}e_i u_{it}\bigr)'
      + \bigl(\Phi_1^{\,k-1}\Sigma_{tr}e_i u_{it}\bigr)\bigl(\Phi_1^k s_{t-1}\bigr)'
    \Bigr)\\
&\quad
  + \sum_{k=1}^h \Phi_1^{\,h-k}\,\Phi_2\,
    \vech\Bigl(
      \bigl(\Phi_1^{\,k-1}\Sigma_{tr}e_i u_{it}\bigr)\bigl(\Phi_1^{\,k-1}\Sigma_{tr}e_i u_{it}\bigr)'
    \Bigr) + r^{(i)}_{h,t+h}.
\end{aligned}
$$
where the remainder term $r^{(i)}_{h,t+h}$ contains:
(1) terms linear in $u_{\ell t}$ for $\ell \neq i$;
(2) quadratic terms $u_{\ell t}^2$ for $\ell \neq i$;
(3) cross-product terms $u_{i t} u_{\ell t}$ for $\ell \neq i$;
(4) cross-product terms $u_{\ell t} u_{m t}$ for $\ell, m \neq i$;
(5) cross terms $s_{t-1} u_{\ell t}$ for $\ell \neq i$;
(6) purely past-dependent terms involving only $s_{t-1}$ or the $R_{t,h}$ component;
(7) terms involving future shocks.

Define the residual $\epsilon^{(i)}_{h,t+h} = r^{(i)}_{h,t+h} - \mathbb{E}[r^{(i)}_{h,t+h}]$ and intercept $\kappa_{h0}^{(i)} = \mathbb{E}[r^{(i)}_{h,t+h}]$. We further collect the terms $u_{it}$, $s_{t-1}u_{it}$, and $u_{it}^2$, and thus write $y_{t+h}$ equivalently as:
\be
y_{t+h} = \kappa_{h0}^{(i)} + \kappa_{h1}^{(i)} u_{it} + {K_{h2}^{(i)}} s_{t-1} u_{it} + \kappa_{h3}^{(i)} u_{it}^2+ \epsilon^{(i)}_{h,t+h},
\ee
where $\kappa_{h1}^{(i)} = \left(\kappa_{h1}^{(1,i)}, \ldots, \kappa_{h1}^{(n,i)}\right)^{\prime},$ $K_{h2}^{(i)}\in\mathbb R^{n\times n}$ is the matrix whose $j$-th row is ${\kappa_{h2}^{(j,i)}}^{\prime}$, and $\kappa_{h3}^{(i)} = \left(\kappa_{h3}^{(1,i)}, \ldots, \kappa_{h3}^{(n,i)}\right)^{\prime}.$

Since $u_{it}$s are iid standard normal shocks, it is easy to verify that  $$\mathbb{E}[u_{it}\epsilon^{(i)}_{h,t+h}]=\mathbb{E}[s_{t-1}u_{it}\epsilon^{(i)}_{h,t+h}]=\mathbb{E}[u_{it}^2\epsilon^{(i)}_{h,t+h}]=0.$$ Then the population coefficients of \textit{Infeas} coincide with $\kappa_{h1}^{(i)}$, $K_{h2}^{(i)}$, and $\kappa_{h3}^{(i)}$.
 
Finally, note that the implied IRF is
\[
\begin{aligned}
\text{IRF}_{j,i}^{Infeas}(s,\delta_i; h)
& = e_j'\Bigg[
\Phi_1^h (1_n +\mathsf{G}\,s)\odot(\Sigma_{tr} \delta_i e_i) \\
& \qquad + \sum_{k=1}^h\Phi_1^{h-k}\,\Phi_2\,\vech\bigl((\Phi_1^k s)(\Phi_1^{\,k-1}\Sigma_{tr}\,\delta_i e_i)' + (\Phi_1^{\,k-1}\Sigma_{tr} \delta_i e_i)(\Phi_1^k s)'\bigr) \\
& \qquad + \sum_{k=1}^h\Phi_1^{h-k}\Phi_2\,\vech\Bigl(\delta_i^2 (\Phi_1^{\,k-1}\Sigma_{tr} e_i)(\Phi_1^{\,k-1}\Sigma_{tr} e_i)'\Bigr)
\Bigg].
\end{aligned}
\]
One can see that the implied IRF exactly recovers the true CAR.
\end{proof}

\subsection{Lemma \ref{lemma:subvector.predictor}}
\begin{lemma} \label{lemma:subvector.predictor}
Under Assumption \ref{assump:qvar.dgp},
\[
\mathbb E\!\left[s_{t-1}\,\big|\,s_{t-1,I}=c_0\right]
=\mathbb E[s_{t-1}s_{t-1,I}']\,\mathbb E[s_{t-1,I}\, s_{t-1,I}']^{-1}c_0.
\]    
\end{lemma}

\begin{oldproof}
Define
\[
B:=\Cov(s_{t-1},s_{t-1,I})\,\Var(s_{t-1,I})^{-1}.
\]
By construction,
\[
\Cov\!\big(s_{t-1}-B\,s_{t-1,I},\,s_{t-1,I}\big)
=\Cov(s_{t-1},s_{t-1,I})-B\,\Var(s_{t-1,I})=0.
\]

Under Assumption~\ref{assump:qvar.dgp}, $(s_{t-1},s_{t-1,I})$ is jointly Gaussian with mean zero. 
For jointly Gaussian vectors, uncorrelatedness implies independence; hence 
$s_{t-1}-B\,s_{t-1,I}$ is independent of $s_{t-1,I}$. Therefore,
\[
\mathbb E[s_{t-1}\mid s_{t-1,I}=c_0]
=\mathbb E[s_{t-1}-B\,s_{t-1,I}\mid s_{t-1,I}=c_0]+B\,c_0
=\mathbb E[s_{t-1}-B\,s_{t-1,I}]+B\,c_0.
\]
Using $\mathbb E[s_{t-1}]=\mathbb E[s_{t-1,I}]=0$ gives 
$\mathbb E[s_{t-1}-B\,s_{t-1,I}]=0$, and thus
\[
\mathbb E[s_{t-1}\mid s_{t-1,I}=c_0]
= B\,c_0
= \Cov(s_{t-1},s_{t-1,I})\,\Var(s_{t-1,I})^{-1} c_0.
\]

Finally, since means are zero, we have
$$\Cov(s_{t-1},s_{t-1,I})=\mathbb E[s_{t-1}\,s_{t-1,I}'], \quad \Var(s_{t-1,I})=\mathbb E[s_{t-1,I}\,s_{t-1,I}'],$$ yielding the stated identity.
\end{oldproof}

\subsection{Proof of Proposition \ref{prop:ccar.multi.infeas}}
\begin{proof}
    Proposition \ref{prop:irf.multi.infeas} implies that we can equivalently write the model as 
\be
y_{j,t+h} = \kappa_{h0}^{(j,i)} + \kappa_{h1}^{(j,i)} u_{it} + {\kappa_{h2}^{(j,i)}}^{\prime} s_{t-1} u_{it} + \kappa_{h3}^{(j,i)} u_{it}^2+ \epsilon^{(j,i)}_{h,t+h},
\label{eq:equiv.multi.model}
\ee
with $\epsilon^{(j,i)}_{h,t+h}$ uncorrelated with the regressors.
The corresponding IRF is equal to $$ \operatorname{CAR}^{(j,i)}_h(s, \delta_i) =  \kappa_{h1}^{(j,i)} \delta_i + {\kappa_{h2}^{(j,i)}}^{\prime} s \delta_i + \kappa_{h3}^{(j,i)} \delta_i^2.$$ 

{\noindent \textbf{Proof of (i)}\quad } Under case (i), we can calculate the conditional CAR as
\begin{equation*}
\begin{aligned}
    &\operatorname{cCAR}^{(j,i)}_h(\mathcal{A}, \delta_i) = \frac{\mathbb{E}\left[\mathbbm{1}\{s_{t-1}\in \mathcal{A}\} \operatorname{CAR}^{(j,i)}_h(s_{t-1}, \delta_i)\right]}{\mathbb{P}\left(s_{t-1}\in \mathcal{A}\right)}\\
    &\quad \quad \quad= \frac{1}{\mathbb{P}\left(s_{t-1}\in \mathcal{A}\right)} \cdot \int \mathbbm{1}\{s_{t-1}\in \mathcal{A}\} p\left(s_{t-1}\right) \operatorname{CAR}^{(j,i)}_h(s_{t-1}, \delta_i)\ d s_{t-1},\\
    &\quad \quad \quad= \frac{1}{\mathbb{P}\left(s_{t-1}\in \mathcal{A}\right)} \cdot \int \mathbbm{1}\{s_{t-1}\in \mathcal{A}\} p\left(s_{t-1}\right) \left(\kappa_{h1}^{(j,i)} \delta_i + {\kappa_{h2}^{(j,i)}}^{\prime} s_{t-1} \delta_i + \kappa_{h3}^{(j,i)} \delta_i^2\right) d s_{t-1}, \\
    &\quad \quad \quad= \frac{\left(\kappa_{h1}^{(j,i)} \mathbb{E}\left[\mathbbm{1}\{s_{t-1}\in \mathcal{A}\}\right] + {\kappa_{h2}^{(j,i)}}^{\prime} \mathbb{E}\left[s_{t-1} \mathbbm{1}\{s_{t-1}\in \mathcal{A}\}\right] \right) \delta_i + \kappa_{h3}^{(j,i)} \mathbb{E}\left[\mathbbm{1}\{s_{t-1}\in \mathcal{A}\}\right] \delta_i^2}{\mathbb{P}\left(s_{t-1}\in \mathcal{A}\right)}.
\end{aligned}
\end{equation*}

Denote $V_t=\left(1, \mathbbm{1}\{s_{t-1}\in \mathcal{A}\}u_{it}, u_{it}^2 \right)^{\prime}$ and $\xi_h = \left(\xi_{h0}^{(j,i)}, \xi_{h1}^{(j,i)}, \xi_{h2}^{(j,i)}\right)^{\prime}.$
Then the population IRF implied by \textit{Infeas-Cond1} is 
\be
\begin{aligned}
    \xi_h  & =  \mathbb{E}[V_tV_t^{\prime}]^{-1} \mathbb{E}[V_t y_{j,t+h}]\\
    & = \begin{pmatrix}
        1 & 0 & 1\\
        0 & \mathbb{E}[\mathbbm{1}\{s_{t-1}\in \mathcal{A}\}] & 0\\
        1 & 0 & 3
    \end{pmatrix}^{-1} \begin{pmatrix}
        \mathbb{E}\left[y_{j,t+h}\right]\\
        \mathbb{E}\left[\mathbbm{1}\{s_{t-1}\in \mathcal{A}\}u_{it} y_{j,t+h}\right] \\
        \mathbb{E}\left[u_{it}^2 y_{j,t+h}\right]
    \end{pmatrix} \\
\end{aligned}
\ee

Since (\ref{eq:equiv.multi.model}) implies that
\begin{gather*}
    \mathbb{E}\left[y_{j,t+h}\right] = \kappa_{h0}^{(j,i)}+ \kappa_{h3}^{(j,i)},\\
    \mathbb{E}\left[\mathbbm{1}\{s_{t-1}\in \mathcal{A}\}u_{it} y_{j,t+h} \right] = \kappa_{h1}^{(j,i)} \mathbb{E}\left[\mathbbm{1}\{s_{t-1}\in \mathcal{A}\}\right] + {\kappa_{h2}^{(j,i)}}^{\prime} \mathbb{E}\left[s_{t-1} \mathbbm{1}\{s_{t-1}\in \mathcal{A}\}\right],\\
     \mathbb{E}\left[u_{it}^2 y_{j,t+h}\right] = \kappa_{h0}^{(j,i)}+ 3\kappa_{h3}^{(j,i)},
\end{gather*}
we obtain that
\begin{gather*}
     \xi_{h0}^{(j,i)} = \kappa_{h0}^{(j,i)},\\
     \xi_{h1}^{(j,i)} = \mathbb{E}\left[\mathbbm{1}\{s_{t-1}\in \mathcal{A}\}\right]^{-1}\left(\kappa_{h1}^{(j,i)} \mathbb{E}\left[\mathbbm{1}\{s_{t-1}\in \mathcal{A}\}\right] + {\kappa_{h2}^{(j,i)}}^{\prime} \mathbb{E}\left[s_{t-1} \mathbbm{1}\{s_{t-1}\in \mathcal{A}\}\right]\right),\\
    \xi_{h2}^{(j,i)} = \kappa_{h3}^{(j,i)}.
\end{gather*}

Since $\mathbb{E}\left[\mathbbm{1}\{s_{t-1}\in \mathcal{A}\}\right] = \mathbb{P}\left(s_{t-1}\in \mathcal{A}\right)$, we can represent the conditional CAR as
\begin{equation*}
    \operatorname{cCAR}^{(j,i)}_h(\mathcal{A}, \delta_i) = \xi_{h1}^{(j,i)} \delta_i + \xi_{h2}^{(j,i)} \delta_i^2.
\end{equation*}

{\noindent \textbf{Proof of (ii)}\quad } 
Under case (ii), we can calculate the conditional CAR as
$$
\begin{aligned}
\operatorname{cCAR}^{(j,i)}_h(c_0, \delta_i) & = \mathbb{E}\left[\operatorname{CAR}^{(j,i)}_h(s_{t-1}, \delta_i)\Big| s_{t-1,I}= c_0\right] \\
& = \kappa_{h1}^{(j,i)} \delta_i + {\kappa_{h2}^{(j,i)}}^{\prime}\, \mathbb{E}\left[s_{t-1}| s_{t-1,I}= c_0\right]\, \delta_i + \kappa_{h3}^{(j,i)} \delta_i^2
\end{aligned}
$$

Consider the following (infeasible) empirical specification \textit{Infeas-Cond2}
\be
y_{j,t+h} = \zeta_{h0}^{(j,i)} + \zeta_{h1}^{(j,i)} u_{it} + \zeta_{h2}^{(j,i)} s_{t-1,I}\,u_{it} + \zeta_{h3}^{(j,i)} u_{it}^2+ \epsilon^{(j,i)}_{h,t+h}.
\ee
The conditional CAR can be recovered from the population coefficients of \textit{Infeas-Cond2}:
\begin{equation*}
    \operatorname{cCAR}^{(j,i)}_h(c_0, \delta_i) = \zeta_{h1}^{(j,i)} \delta_i + {\zeta_{h2}^{(j,i)}}^{\prime}c_0\,\delta_i + \zeta_{h3}^{(j,i)} \delta_i^2.
\end{equation*}

Denote $U_t=\left(1, u_{it}, s^{\prime}_{t-1,I}\, u_{it}, u_{it}^2 \right)^{\prime}$ and $\zeta_h = \left(\zeta_{h0}^{(j,i)}, \zeta_{h1}^{(j,i)}, {\zeta_{h2}^{(j,i)}}^{\prime}, \zeta_{h3}^{(j,i)}\right)^{\prime}.$
Then the population IRF implied by \textit{Infeas-Cond2} is 
\be
\begin{aligned}
    \zeta_h  & =  \mathbb{E}[U_tU_t^{\prime}]^{-1} \mathbb{E}[U_t y_{j,t+h}]\\
    & = \begin{pmatrix}
        1 & 0 & 0 & 1\\
        0 & 1 & 0 & 0\\
        0 & 0 & \mathbb{E}[s_{t-1,I}\, s^{\prime}_{t-1,I}] & 0\\
        1 & 0 & 0 & 3
    \end{pmatrix}^{-1} \begin{pmatrix}
        \mathbb{E}\left[y_{j,t+h}\right]\\
        \mathbb{E}\left[u_{it} y_{j,t+h}\right]\\
        \mathbb{E}\left[s_{t-1,I}\, u_{it}\, y_{j,t+h}\right] \\
        \mathbb{E}\left[u_{it}^2 y_{j,t+h}\right]
    \end{pmatrix} \\
\end{aligned}
\ee

Since (\ref{eq:equiv.multi.model}) implies that
\begin{gather*}
    \mathbb{E}\left[y_{j,t+h}\right] = \kappa_{h0}^{(j,i)}+ \kappa_{h3}^{(j,i)},\\
    \mathbb{E}\left[u_{it}y_{j,t+h}\right] = \kappa_{h1}^{(j,i)},\\
    \mathbb{E}\left[s_{t-1,I}\, u_{it} y_{j,t+h} \right] = \mathbb{E}\left[s_{t-1,I}\, s_{t-1}^{\prime}\right] \kappa_{h2}^{(j,i)},\\
     \mathbb{E}\left[u_{it}^2 y_{j,t+h}\right] = \kappa_{h0}^{(j,i)}+ 3\kappa_{h3}^{(j,i)},
\end{gather*}
we obtain that
\begin{gather*}
     \zeta_{h0}^{(j,i)} = \kappa_{h0}^{(j,i)},\\
     \zeta_{h1}^{(j,i)} = \kappa_{h1}^{(j,i)},\\
     \zeta_{h2}^{(j,i)} = \mathbb{E}\left[s_{t-1,I}\, s_{t-1,I}^{\prime}\right]^{-1}\, \mathbb{E}\left[s_{t-1,I}\, s_{t-1}^{\prime}\right]\, \kappa_{h2}^{(j,i)},\\
     \zeta_{h3}^{(j,i)} = \kappa_{h3}^{(j,i)}.
\end{gather*}

Lemma \ref{lemma:subvector.predictor} implies that $$ {\zeta_{h2}^{(j,i)}}^{\prime} c_0= {\kappa_{h2}^{(j,i)}}^{\prime}\, \mathbb E\!\left[s_{t-1}\big|\,s_{t-1,I}=c_0\right].$$
Then we can represent the conditional CAR as
\begin{equation*}
    \operatorname{cCAR}^{(j,i)}_h(c_0, \delta_i) = \zeta_{h1}^{(j,i)} \delta_i + {\zeta_{h2}^{(j,i)}}^{\prime} c_0\,\delta_i + \zeta_{h3}^{(j,i)} \delta_i^2.
\end{equation*}
\end{proof}

\subsection{Lemma \ref{lem:qar.regularity}}

\begin{lemma}[Stable QVAR implies geometric physical dependence]\label{lem:qar.regularity}
Suppose Assumptions~\ref{assump:qvar.dgp} and \ref{ass:feas.regularity} hold.
For each integer $m\ge 1$, let $\{u_r^\ast\}_{r\in\mathbb Z}$ be an i.i.d. copy of $\{u_r\}_{r\in\mathbb Z}$ and construct the coupled innovation sequence
\[
u_r^{(m)}=
\begin{cases}
u_r, & r\neq t-m,\\
u_r^\ast, & r=t-m.
\end{cases}
\]
Let $(y_r^{(m)},s_r^{(m)})$ be the QVAR process driven by $\{u_r^{(m)}\}$, and define
$z_{t-1}^{(m)}$, $W_{t-1}^{(m)}$, $x_{it}^{(m)}$, $\epsilon_{h,t+h}^{(m)}$, and $\psi_{ht}^{(m)}$ analogously.
Then, for every $q\ge 1$, there exist constants $C_q<\infty$ and $\rho_q\in(0,1)$ such that for all $m\ge 1$,
\[
\|s_t-s_t^{(m)}\|_q + \|y_t-y_t^{(m)}\|_q \le C_q \rho_q^m,
\]
and
\[
\|z_{t-1}-z_{t-1}^{(m)}\|_q
+
\|W_{t-1}-W_{t-1}^{(m)}\|_q
+
\|x_{it}-x_{it}^{(m)}\|_q
+
\|\epsilon_{h,t+h}-\epsilon_{h,t+h}^{(m)}\|_q
+
\|\psi_{ht}-\psi_{ht}^{(m)}\|_q
\le
C_q \rho_q^m.
\]
In particular, $(y_t,s_t)$, $x_{it}$, and, for each fixed $h$, the shifted score sequence $\bar\psi_{ht}\equiv \psi_{h,t-h}$ are stationary ergodic causal Bernoulli shifts (in the sense of \citet{Wu2005}) with finite moments of all orders and absolutely summable physical-dependence coefficients.
\end{lemma}

\begin{proof}
Fix $q\ge 1$ and choose a matrix norm. Since $\rho(\Phi_1)<1$, there exist constants $C_\Phi<\infty$ and $r\in(0,1)$ such that
\[
\|\Phi_1^\ell\|\le C_\Phi r^\ell
\qquad\text{for all }\ell\ge 0.
\]
Write $\eta_t=\Sigma_{tr}u_t$ and $\eta_t^\ast=\Sigma_{tr}u_t^\ast$.
The state recursion admits the stationary representation
\[
s_t=\sum_{\ell=0}^{\infty}\Phi_1^\ell \eta_{t-\ell},
\]
which converges almost surely and in $L^q$ because $\sum_{\ell\ge 0}\|\Phi_1^\ell\|<\infty$ and $\eta_t\in L^q$.
Since $u_t$ is Gaussian, $s_t$ is Gaussian and therefore has finite moments of all orders.

Define
\[
G_t
\equiv
\Phi_2\vech(s_{t-1}s_{t-1}')
+
\diag(1_n+\mathsf{G} s_{t-1})\,\eta_t.
\]
Then $G_t\in L^q$ for every $q\ge 1$.
The $y_t$ recursion can therefore be iterated backward to obtain
\begin{equation}
y_t=\sum_{\ell=0}^{\infty}\Phi_1^\ell G_{t-\ell},
\label{eq:y_rep_qvar_gmc}
\end{equation}
with convergence almost surely and in $L^q$.
Hence $y_t$ also has finite moments of all orders.

Because only the innovation at time $t-m$ is replaced,
\[
s_t-s_t^{(m)}
=
\Phi_1^m(\eta_{t-m}-\eta_{t-m}^\ast).
\]
Hence
\[
\|s_t-s_t^{(m)}\|_q
\le
\|\Phi_1^m\|\,\|\eta_{0}-\eta^\ast_{0}\|_q
\le
C_q r^m.
\]

Let
\[
G_t^{(m)}
\equiv
\Phi_2\vech(s_{t-1}^{(m)}{s_{t-1}^{(m)}}')
+
\diag(1_n+\mathsf{G} s_{t-1}^{(m)})\,\eta_t^{(m)},
\]
where $\eta_t^{(m)}=\Sigma_{tr}u_t^{(m)}$.
Using \eqref{eq:y_rep_qvar_gmc},
\[
y_t-y_t^{(m)}
=
\sum_{\ell=0}^{\infty}\Phi_1^\ell\bigl(G_{t-\ell}-G_{t-\ell}^{(m)}\bigr).
\]
If $0\le \ell\le m-1$, then $\eta_{t-\ell}^{(m)}=\eta_{t-\ell}$, so only the state term differs:
\[
G_{t-\ell}-G_{t-\ell}^{(m)}
=
\Phi_2\!\left\{\vech(s_{t-\ell-1}s_{t-\ell-1}')-\vech\bigl(s_{t-\ell-1}^{(m)}{s_{t-\ell-1}^{(m)}}'\bigr)\right\}
+
\diag\!\bigl(\mathsf{G}(s_{t-\ell-1}-s_{t-\ell-1}^{(m)})\bigr)\eta_{t-\ell}.
\]
Because
\[
aa'-bb' = a(a-b)' + (a-b)b',
\]
H\"older's inequality yields
\[
\Big\|\vech(s_{t-\ell-1}s_{t-\ell-1}')-\vech\bigl(s_{t-\ell-1}^{(m)}{s_{t-\ell-1}^{(m)}}'\bigr)\Big\|_q
\le
C_q\|s_{t-\ell-1}-s_{t-\ell-1}^{(m)}\|_{2q},
\]
and similarly
\[
\Big\|\diag\bigl(\mathsf{G}(s_{t-\ell-1}-s_{t-\ell-1}^{(m)})\bigr)\eta_{t-\ell}\Big\|_q
\le
C_q\|s_{t-\ell-1}-s_{t-\ell-1}^{(m)}\|_{2q}.
\]
Applying the bound for the state process at time $t-\ell-1$ gives
\[
\|G_{t-\ell}-G_{t-\ell}^{(m)}\|_q \le C_q r^{m-\ell}
\qquad (0\le \ell\le m-1).
\]
For $\ell=m$, the state at time $t-m-1$ is unchanged, so
\[
G_{t-m}-G_{t-m}^{(m)}
=
\diag(1_n+\mathsf{G} s_{t-m-1})\,(\eta_{t-m}-\eta_{t-m}^\ast),
\]
and therefore
\[
\|G_{t-m}-G_{t-m}^{(m)}\|_q \le C_q.
\]
If $\ell\ge m+1$, then the innovation at time $t-m$ cannot affect $G_{t-\ell}$, so
\[
G_{t-\ell}=G_{t-\ell}^{(m)}.
\]
Therefore
\[
\|y_t-y_t^{(m)}\|_q
\le
\sum_{\ell=0}^{m-1}\|\Phi_1^\ell\|\, C_q r^{m-\ell}
+
\|\Phi_1^m\|\,C_q
\le
C_q\bigl(mr^m+r^m\bigr).
\]
Choose any $\rho_q\in(r,1)$. Since $mr^m\le C_{\rho_q}\rho_q^m$, this implies
\[
\|y_t-y_t^{(m)}\|_q \le C_q \rho_q^m.
\]

Now use the finite-lag linear representations of $z_{t-1}$ and $W_{t-1}$ from Assumption~\ref{ass:feas.regularity}. Because only finitely many lagged $y$'s and $u$'s enter, the previous bounds imply
\[
\|z_{t-1}-z_{t-1}^{(m)}\|_q + \|W_{t-1}-W_{t-1}^{(m)}\|_q \le C_q \rho_q^m,
\]
after enlarging $C_q$ if necessary to cover the finitely many cases $m\le L$.
Since current $u_{it}$ is unchanged when $m\ge 1$,
\[
x_{it}-x_{it}^{(m)}
=
\Bigl(0,\;0,\;u_{it}(z_{t-1}-z_{t-1}^{(m)})',\;0,\;(W_{t-1}-W_{t-1}^{(m)})'\Bigr)',
\]
so H\"older's inequality gives
\[
\|x_{it}-x_{it}^{(m)}\|_q \le C_q \rho_q^m.
\]

Because $h$ is fixed, applying the bound for $y_t$ at time $t+h$ gives
\[
\|y_{j,t+h}-y_{j,t+h}^{(m)}\|_q \le C_q \rho_q^{m+h}\le C_q \rho_q^m,
\]
after another relabeling of the constant.
Therefore
\[
\epsilon_{h,t+h}-\epsilon_{h,t+h}^{(m)}
=
\bigl(y_{j,t+h}-y_{j,t+h}^{(m)}\bigr)
-
\bigl(x_{it}-x_{it}^{(m)}\bigr)'\vartheta_h,
\]
and hence
\[
\|\epsilon_{h,t+h}-\epsilon_{h,t+h}^{(m)}\|_q \le C_q \rho_q^m.
\]
Finally,
\[
\psi_{ht}-\psi_{ht}^{(m)}
=
(x_{it}-x_{it}^{(m)})\epsilon_{h,t+h}
+
x_{it}^{(m)}\bigl(\epsilon_{h,t+h}-\epsilon_{h,t+h}^{(m)}\bigr),
\]
so another application of H\"older's inequality gives
\[
\|\psi_{ht}-\psi_{ht}^{(m)}\|_q \le C_q \rho_q^m.
\]

The representations above show that $(y_t,s_t)$ and $x_{it}$ are measurable functions of the i.i.d. shock sequence $\{u_r\}_{r\le t}$, while $\psi_{ht}$ depends only on the finite lead block $u_{t+1},\ldots,u_{t+h}$ in addition to the past and current shocks. Hence, for fixed $h$, the shifted process $\bar\psi_{ht}\equiv\psi_{h,t-h}$ is a causal Bernoulli shift in the sense of \citet{Wu2005}. Writing $\bar\psi_{ht}^{(m)}\equiv \psi_{h,t-h}^{(m)}$, its physical-dependence coefficients are $\delta_q(m)\equiv \|\bar\psi_{ht}-\bar\psi_{ht}^{(m)}\|_q$, which are absolutely summable because the coupled $L^q$ differences are summable. Stationarity and ergodicity follow because the process is a measurable function of the i.i.d. shock sequence.
\end{proof}

\subsection{Proof of Proposition \ref{prop:hac.feas}}
\begin{proof}
Fix $h$ and a pair $(j,i)$.
Let $T_h\equiv T-h$, and write the population regression as
\[
y_{j,t+h}=x_{it}'\vartheta_h+\epsilon_{h,t+h},
\qquad
\psi_{ht}\equiv x_{it}\epsilon_{h,t+h}.
\]
By construction of the population linear projection, $\EE[\psi_{ht}]=0$.

Define the shifted score sequence
\[
\bar\psi_{ht}\equiv \psi_{h,t-h},
\qquad t\in\mathbb Z.
\]
Because $h$ is fixed, the asymptotics of
\[
\frac{1}{\sqrt{T_h}}\sum_{t=1}^{T_h}\psi_{ht}
\qquad\text{and}\qquad
\frac{1}{\sqrt{T_h}}\sum_{t=h+1}^{T}\bar\psi_{ht}
\]
are identical.
By Lemma~\ref{lem:qar.regularity}, for every $q>4$ the process $\{\bar\psi_{ht}\}$ is a stationary ergodic causal Bernoulli shift with finite $q$-th moment and physical-dependence coefficients
\[
\delta_q(m)\equiv \|\bar\psi_{ht}-\bar\psi_{ht}^{(m)}\|_q
\]
satisfying $\delta_q(m)\le C_q\rho_q^m$.
Let $\mathcal F_t\equiv \sigma(u_s:s\le t)$ and $P_j Z\equiv \EE[Z\mid\mathcal F_j]-\EE[Z\mid\mathcal F_{j-1}]$. For any conformable vector $a$, define the scalar process $X_t(a)\equiv a'\bar\psi_{ht}$. Its physical-dependence coefficients satisfy
\[
\delta_{2,a}(m)\equiv \|a'(\bar\psi_{ht}-\bar\psi_{ht}^{(m)})\|_2 \le \|a\|\,\delta_2(m),
\]
so $\sum_{m\ge 0}\delta_{2,a}(m)<\infty$. By \citet[Theorem~1(ii)]{Wu2005}, the corresponding projective dependence coefficients satisfy
\[
\theta_{2,a}(m)\equiv \|P_0 X_m(a)\|_2 \le \delta_{2,a}(m),
\]
and hence $\sum_{m\ge 0}\theta_{2,a}(m)<\infty$. Therefore \citet[Theorem~3(ii)]{Wu2011} yields
\[
\frac{1}{\sqrt{T_h}}\sum_{t=h+1}^{T} a'\bar\psi_{ht}
\xrightarrow{d}
\mathcal N\!\bigl(0,\,a'\Omega_h a\bigr),
\qquad
 a'\Omega_h a = \sum_{m=-\infty}^{\infty} a'\Gamma_{h,m}a.
\]
Since this holds for every $a$, the Cram\'er--Wold device implies
\[
\frac{1}{\sqrt{T_h}}\sum_{t=1}^{T_h}\psi_{ht}
\xrightarrow{d}
\mathcal{N}\!\big(0,\Omega_h\big),
\qquad
\Omega_h=\sum_{m=-\infty}^{\infty}\Gamma_{h,m},
\]
where $\Gamma_{h,m}\equiv \EE[\psi_{ht}\psi_{h,t-m}']$.
The same conditions imply that the infeasible HAC estimator based on the true scores is consistent:
if
\[
\tilde\psi_{ht}\equiv x_{it}\epsilon_{h,t+h},
\qquad
\tilde\Gamma_{h,m}\equiv \frac{1}{T_h}\sum_{t=m+1}^{T_h}\tilde\psi_{ht}\tilde\psi_{h,t-m}',
\]
and
\[
\tilde\Omega_h
\equiv
\tilde\Gamma_{h,0}
+
\sum_{m=1}^{b_T}K(m/b_T)\bigl(\tilde\Gamma_{h,m}+\tilde\Gamma_{h,m}'\bigr),
\]
then
\[
\tilde\Omega_h\to_p \Omega_h
\]
for every bounded symmetric kernel $K$ with $K(0)=1$ and every bandwidth sequence satisfying $b_T\to\infty$ and $b_T/\sqrt{T}\to 0$.

Next define
\[
\hat Q_h\equiv \frac{1}{T_h}\sum_{t=1}^{T_h}x_{it}x_{it}'.
\]
By Lemma~\ref{lem:qar.regularity}, $\{x_{it}x_{it}'\}$ is stationary and ergodic with finite mean, so the ergodic theorem yields
\[
\hat Q_h\to_p Q_h\equiv \EE[x_{it}x_{it}'].
\]
By Assumption~\ref{ass:feas.regularity}, $Q_h$ is positive definite.
Using the OLS formula and $y_{j,t+h}=x_{it}'\vartheta_h+\epsilon_{h,t+h}$,
\[
\hat\vartheta_h-\vartheta_h
=
\Big(\sum_{t=1}^{T_h}x_{it}x_{it}'\Big)^{-1}\sum_{t=1}^{T_h}x_{it}\epsilon_{h,t+h}.
\]
Therefore, since $h$ is fixed and $T_h/T\to 1$,
\[
\sqrt{T}\big(\hat\vartheta_h-\vartheta_h\big)
=
\hat Q_h^{-1}\,
\frac{1}{\sqrt{T_h}}\sum_{t=1}^{T_h}\psi_{ht}
+o_p(1)
\xrightarrow{d}
\mathcal{N}\!\big(0,V_h\big),
\]
with $V_h=Q_h^{-1}\Omega_h Q_h^{-1}$.
This proves \eqref{eq:asymp_normal_feas}.

For the feasible HAC/HAR estimator, let
\[
\hat\psi_{ht}\equiv x_{it}\hat\epsilon_{h,t+h},
\qquad
\hat\Gamma_{h,m}\equiv \frac{1}{T_h}\sum_{t=m+1}^{T_h}\hat\psi_{ht}\hat\psi_{h,t-m}',
\]
and define
\[
\hat\Omega_h
\equiv
\hat\Gamma_{h,0}
+
\sum_{m=1}^{b_T}K(m/b_T)\bigl(\hat\Gamma_{h,m}+\hat\Gamma_{h,m}'\bigr).
\]
It remains to show that replacing $\epsilon_{h,t+h}$ by $\hat\epsilon_{h,t+h}$ is asymptotically negligible.
Let $\Delta_h\equiv \hat\vartheta_h-\vartheta_h$.
Then
\[
\hat\epsilon_{h,t+h}-\epsilon_{h,t+h}=-x_{it}'\Delta_h,
\qquad
\hat\psi_{ht}-\tilde\psi_{ht}= -x_{it}x_{it}'\Delta_h.
\]
From the asymptotic normality just established, $\Delta_h=O_p(T^{-1/2})$.
Moreover, Lemma~\ref{lem:qar.regularity} gives finite moments of all orders for $x_{it}$; in particular,
\[
\EE\|x_{it}\|^4<\infty.
\]
Hence
\[
\frac{1}{T_h}\sum_{t=1}^{T_h}\|\hat\psi_{ht}-\tilde\psi_{ht}\|^2
\le
\|\Delta_h\|^2\,
\frac{1}{T_h}\sum_{t=1}^{T_h}\|x_{it}\|^4
=
O_p(T^{-1}),
\]
so
\[
\Big(\frac{1}{T_h}\sum_{t=1}^{T_h}\|\hat\psi_{ht}-\tilde\psi_{ht}\|^2\Big)^{1/2}
=
O_p(T^{-1/2}).
\]
Using Cauchy--Schwarz and
\[
\frac{1}{T_h}\sum_{t=1}^{T_h}\|\tilde\psi_{ht}\|^2=O_p(1),
\qquad
\frac{1}{T_h}\sum_{t=1}^{T_h}\|\hat\psi_{ht}\|^2=O_p(1),
\]
we obtain, uniformly over $0\le m\le b_T$,
\[
\|\hat\Gamma_{h,m}-\tilde\Gamma_{h,m}\|
=
O_p(T^{-1/2}).
\]
Therefore,
\[
\|\hat\Omega_h-\tilde\Omega_h\|
\le
C\sum_{m=0}^{b_T}\|\hat\Gamma_{h,m}-\tilde\Gamma_{h,m}\|
=
O_p\!\left(\frac{b_T}{\sqrt{T}}\right)
=o_p(1),
\]
where $C<\infty$ bounds the kernel weights.
Combining this with $\tilde\Omega_h\to_p\Omega_h$ yields
\[
\hat\Omega_h\to_p \Omega_h.
\]
Since $\hat Q_h\to_p Q_h$, it follows that
\[
\hat V_h\equiv \hat Q_h^{-1}\hat\Omega_h\hat Q_h^{-1}\to_p V_h.
\]

Finally, for any fixed $(z,\delta_i)$, the estimator \eqref{eq:feas_irf_hat} is the linear functional
\[
\widehat{\operatorname{IRF}}^{Feas}_{j,i}(z,\delta_i;h)=g(z,\delta_i)'\hat\vartheta_h,
\]
so \eqref{eq:asymp_normal_irf_feas} follows by the delta method.
\end{proof}

\subsection{Proof of Proposition \ref{prop:ehw.criterion}}
\begin{proof}
Let $a_{t-1}\equiv (1,W_{t-1}')'$ and let
\[
r_{it}\equiv \big(u_{it},\,(z_{t-1}u_{it})',\,u_{it}^2\big)'.
\]
Because $Q_h$ is positive definite, its principal submatrix $\EE[a_{t-1}a_{t-1}']$ is also positive definite, so the population linear projection of each block of $r_{it}$ on $a_{t-1}$ is unique. Since $u_{it}$ is independent of $\mathcal F_{t-1}$ with $\EE[u_{it}]=0$ and $\EE[u_{it}^2]=1$, we have
\[
\EE[a_{t-1}u_{it}]=0,\qquad
\EE[a_{t-1}(z_{t-1}u_{it})']=0,\qquad
\EE[a_{t-1}(u_{it}^2-1)]=0.
\]
Hence
\[
\Proj(u_{it}\mid a_{t-1})=0,\qquad
\Proj(z_{t-1}u_{it}\mid a_{t-1})=0,\qquad
\Proj(u_{it}^2\mid a_{t-1})=1,
\]
and therefore the residualized slope regressor is exactly
\[
r_{it}-\Proj(r_{it}\mid a_{t-1}) = x_{it}^c
= \big(u_{it},\,(z_{t-1}u_{it})',\,c_{it}\big)'.
\]

By the Frisch--Waugh--Lovell theorem, the slope block $(\theta_{h1},\theta_{h2}',\theta_{h3})'$ of the original regression can be analyzed using the partialled-out regression with regressor $x_{it}^c$ and residual $\epsilon_{h,t+h}$. The argument in Proposition~\ref{prop:hac.feas} therefore applies to this partialled-out regression and yields
\[
\sqrt{T}\big((\hat\theta_{h1},\hat\theta_{h2}',\hat\theta_{h3})'-(\theta_{h1},\theta_{h2}',\theta_{h3})'\big)
\xrightarrow{d}
\mathcal N\!\Big(0,\,(Q_h^c)^{-1}\Omega_h^c(Q_h^c)^{-1}\Big),
\]
where
\[
Q_h^c=\EE[x_{it}^c x_{it}^{c\prime}],
\qquad
\Omega_h^c=\sum_{\ell=-\infty}^{\infty}\Gamma_{h,\ell}^c,
\qquad
\Gamma_{h,\ell}^c=\EE[\psi_{ht}^c\psi_{h,t-\ell}^{c\prime}].
\]
An EHW estimator replaces the long-run variance $\Omega_h^c$ by the zero-lag covariance $\Gamma_{h,0}^c$. It is therefore asymptotically valid if and only if
\[
\Omega_h^c=\Gamma_{h,0}^c
\qquad\Longleftrightarrow\qquad
\sum_{\ell\neq0}\Gamma_{h,\ell}^c=0.
\]
This proves the criterion.
\end{proof}

\subsection{Proof of Proposition \ref{prop:exactstate.pattern}}
\begin{proof}
Let
\[
B\equiv \Sigma_{tr},\qquad b_r\equiv Be_r,\qquad D(s)\equiv \diag(1_n+\mathsf{G} s),
\]
and let $\mathsf{g}_j'$ denote the $j$-th row of $\mathsf{G}$.

For $h=0$, \eqref{eq:qvar.model} gives
\[
y_{jt}
=
\underbrace{e_j'\!\left[\Phi_1y_{t-1}+\Phi_2\vech(s_{t-1}s_{t-1}')\right]}_{A_{0,t-1}^{(j,i)}}
+
b_{ji}u_{it}
+
b_{ji}\mathsf{g}_j's_{t-1}u_{it}
+
\sum_{\ell\neq i} b_{j\ell}(1+\mathsf{g}_j's_{t-1})u_{\ell t}.
\]
Hence
\[
e_{0t}^{\star,(j,i)}=\sum_{\ell\neq i} b_{j\ell}(1+\mathsf{g}_j's_{t-1})u_{\ell t}.
\]
Now $x_{it}^{c,\star}=(u_{it},(s_{t-1}'u_{it})',c_{it})'$ is $\mathcal F_t$-measurable, and for each $\ell\neq i$,
\[
\EE[u_{it}u_{\ell t}\mid\mathcal F_{t-1}]=0,
\qquad
\EE[s_{t-1}u_{it}u_{\ell t}\mid\mathcal F_{t-1}]=s_{t-1}\EE[u_{it}u_{\ell t}]=0,
\]
\[
\EE[c_{it}u_{\ell t}\mid\mathcal F_{t-1}]=\EE[(u_{it}^2-1)u_{\ell t}]=0.
\]
Therefore
\[
\EE[\psi_{0t}^{c,\star}\mid\mathcal F_{t-1}]
=
\EE[x_{it}^{c,\star}e_{0t}^{\star,(j,i)}\mid\mathcal F_{t-1}]
=0,
\]
so $\{\psi_{0t}^{c,\star}\}$ is a martingale difference sequence. This proves part (i).

Fix now any $h\ge 1$, and define
\[
m_{h,t}^{(j)}\equiv \EE[y_{j,t+h}\mid\mathcal F_t].
\]
Since $\psi_{h,t-h}^{u,\star}=u_{i,t-h}e_{h,t}^{\star,(j,i)}$ is $\mathcal F_t$-measurable, iterated expectations imply
\[
\EE[\psi_{ht}^{u,\star}\psi_{h,t-h}^{u,\star}]
=
\EE\!\Big[
u_{it}\Big(m_{h,t}^{(j)}-A_{h,t-1}^{(j,i)}-\kappa_{h1}^{(j,i)}u_{it}
-\kappa_{h2}^{(j,i)\prime}s_{t-1}u_{it}
-\kappa_{h3}^{(j,i)}u_{it}^2\Big)\psi_{h,t-h}^{u,\star}
\Big].
\]
Thus only the $\mathcal F_t$-measurable part of $y_{j,t+h}$ matters for this covariance.

Iterating the QVAR recursion and taking $\mathcal F_t$-conditional expectations yields
\[
m_{h,t}
=
\Phi_1^h y_t
+
\sum_{r=1}^{h}\Phi_1^{h-r}\Phi_2\,\EE[\vech(s_{t+r-1}s_{t+r-1}')\mid\mathcal F_t],
\]
because $\EE[D(s_{t+r-1})Bu_{t+r}\mid\mathcal F_t]=0$ for every $r\ge1$. For each $r\ge1$, write
\[
s_{t+r-1}
=
\Phi_1^r s_{t-1}+\Phi_1^{r-1}Bu_t+v_{r,t},
\]
where $v_{r,t}$ depends only on $u_{t+1},\ldots,u_{t+r-1}$ and satisfies
\[
\EE[v_{r,t}\mid\mathcal F_t]=0,
\qquad
\EE[v_{r,t}u_t'\mid\mathcal F_t]=0.
\]
Hence
\[
\EE[\vech(s_{t+r-1}s_{t+r-1}')\mid\mathcal F_t]
=
\vech\!\left((\Phi_1^r s_{t-1}+\Phi_1^{r-1}Bu_t)(\Phi_1^r s_{t-1}+\Phi_1^{r-1}Bu_t)'\right)
+
C_{r,t-1},
\]
for some $\mathcal F_{t-1}$-measurable random vector $C_{r,t-1}$.

Substituting
\[
y_t=\Phi_1y_{t-1}+\Phi_2\vech(s_{t-1}s_{t-1}')+D(s_{t-1})Bu_t
\]
into the formula for $m_{h,t}$ and collecting the date-$t$ terms involving $u_t$, we obtain
\[
m_{h,t}^{(j)}
=
A_{h,t-1}^{(j,i)}
+\kappa_{h1}^{(j,i)}u_{it}
+\kappa_{h2}^{(j,i)\prime}s_{t-1}u_{it}
+\kappa_{h3}^{(j,i)}u_{it}^2
+\sum_{\ell\neq i}\chi_{h,j,i\ell}u_{it}u_{\ell t}
+\widetilde R_{ht}^{(j,i)},
\]
where $\widetilde R_{ht}^{(j,i)}$ is $\mathcal F_t$-measurable, satisfies
\[
\EE[\widetilde R_{ht}^{(j,i)}\mid\mathcal F_{t-1}]=0,
\]
and whose date-$t$ part is a linear combination of $u_{\ell t}$, $u_{\ell t}^2-1$, and $u_{\ell t}u_{mt}$ with $\ell,m\neq i$. The coefficient on $u_{it}u_{\ell t}$ is exactly
\[
\chi_{h,j,i\ell}
=
\sum_{r=1}^{h} e_j'\Phi_1^{h-r}\Phi_2\,
\vech\!\left((\Phi_1^{r-1}b_i)(\Phi_1^{r-1}b_\ell)' + (\Phi_1^{r-1}b_\ell)(\Phi_1^{r-1}b_i)'\right).
\]

Next, because $h\ge1$, every term subtracted in the definition of $e_{h,t}^{\star,(j,i)}$ is dated no later than $t-1$. Using again the one-step QVAR equation for $y_t$, we can therefore write
\[
e_{h,t}^{\star,(j,i)}=M_{h,t-1}^{(j,i)}+e_j'D(s_{t-1})Bu_t,
\]
for some $\mathcal F_{t-1}$-measurable random variable $M_{h,t-1}^{(j,i)}$. Hence
\[
\psi_{h,t-h}^{u,\star}
=
u_{i,t-h}M_{h,t-1}^{(j,i)}
+
u_{i,t-h}e_j'D(s_{t-1})Bu_t.
\]

Conditioning on $\mathcal F_{t-1}$, all products involving $u_{it}\widetilde R_{ht}^{(j,i)}$ have zero conditional expectation by parity, and the same is true for
\[
u_{it}\Big(\sum_{\ell\neq i}\chi_{h,j,i\ell}u_{it}u_{\ell t}\Big)\cdot u_{i,t-h}M_{h,t-1}^{(j,i)}.
\]
Therefore only the product of the omitted cross-shock term with the date-$t$ piece of the lagged residual survives:
\[
\EE[\psi_{ht}^{u,\star}\psi_{h,t-h}^{u,\star}]
=
\EE\!\left[
u_{it}\Big(\sum_{\ell\neq i}\chi_{h,j,i\ell}u_{it}u_{\ell t}\Big)
u_{i,t-h}\,e_j'D(s_{t-1})Bu_t
\right].
\]
Since
\[
e_j'D(s_{t-1})Bu_t
=
\sum_{m=1}^n b_{jm}(1+\mathsf{g}_j's_{t-1})u_{mt},
\]
conditioning once more on $\mathcal F_{t-1}$ gives
\[
\EE\!\left[
u_{it}\Big(\sum_{\ell\neq i}\chi_{h,j,i\ell}u_{it}u_{\ell t}\Big)e_j'D(s_{t-1})Bu_t
\ \Big|\ \mathcal F_{t-1}
\right]
=
\sum_{\ell\neq i}\chi_{h,j,i\ell}b_{j\ell}(1+\mathsf{g}_j's_{t-1}),
\]
because $\EE[u_{it}^2u_{\ell t}u_{mt}]=\mathbf 1\{m=\ell\}$ for every $\ell\neq i$.

Multiplying by $u_{i,t-h}$ and taking expectations yields
\[
\EE[\psi_{ht}^{u,\star}\psi_{h,t-h}^{u,\star}]
=
\sum_{\ell\neq i}\chi_{h,j,i\ell}b_{j\ell}\,
\EE\!\left[u_{i,t-h}(1+\mathsf{g}_j's_{t-1})\right].
\]
The constant term vanishes because $\EE[u_{i,t-h}]=0$. Under Assumption~\ref{assump:qvar.dgp},
\[
s_{t-1}=\sum_{r=0}^{\infty}\Phi_1^rBu_{t-1-r},
\]
so independence across dates implies
\[
\EE[u_{i,t-h}s_{t-1}] = \Phi_1^{h-1}b_i,
\qquad
\EE[u_{i,t-h}\mathsf{g}_j's_{t-1}] = \mathsf{g}_j'\Phi_1^{h-1}b_i.
\]
Therefore
\[
\EE[\psi_{ht}^{u,\star}\psi_{h,t-h}^{u,\star}]
=
(\mathsf{g}_j'\Phi_1^{h-1} b_i)\sum_{\ell\neq i} b_{j\ell}\chi_{h,j,i\ell}.
\]
If this quantity is nonzero, then the $(u,u)$ entry of the lag-$h$ slope-score autocovariance matrix is nonzero, so EHW fails at horizon $h$.
\end{proof}

\subsection{Proof of Proposition \ref{prop:proxyerror.h0}}
\begin{proof}
Let
\[
\begin{aligned}
B&\equiv \Sigma_{tr},
\qquad
b_r\equiv Be_r,
\qquad
\Lambda\equiv \Cov(s_{t-1},z_{t-1})\Var(z_{t-1})^{-1},\\
a_\Lambda&\equiv \EE[s_{t-1}] - \Lambda\EE[z_{t-1}],
\qquad
\xi_{t-1}\equiv s_{t-1}-a_\Lambda-\Lambda z_{t-1}.
\end{aligned}
\]
and let $\mathsf{g}_j'$ denote the $j$-th row of $\mathsf{G}$. Then $(a_\Lambda,\Lambda)$ are the coefficients of the linear projection of $s_{t-1}$ on $(1,z_{t-1}')'$, so
\[
\EE[\xi_{t-1}]=0,
\qquad
\EE[\xi_{t-1}z_{t-1}']=0.
\]

At $h=0$,
\[
y_{jt}
=
e_j'\!\left[\Phi_1y_{t-1}+\Phi_2\vech(s_{t-1}s_{t-1}')\right]
+
e_j'D(s_{t-1})Bu_t.
\]
Using $s_{t-1}=a_\Lambda+\Lambda z_{t-1}+\xi_{t-1}$ and isolating shock $i$,
\[
\begin{aligned}
e_j'D(s_{t-1})Bu_t
&=
b_{ji}(1+\mathsf{g}_j'a_\Lambda)u_{it}
+
b_{ji}\mathsf{g}_j'\Lambda z_{t-1}u_{it}
+
b_{ji}\mathsf{g}_j'\xi_{t-1}u_{it} \\
&\qquad
+
\sum_{\ell\neq i} b_{j\ell}(1+\mathsf{g}_j's_{t-1})u_{\ell t}.
\end{aligned}
\]
Set
\[
A_{0,t-1}^{(j,i)}\equiv e_j'\!\left[\Phi_1y_{t-1}+\Phi_2\vech(s_{t-1}s_{t-1}')\right].
\]
Since the original regression \eqref{eq:feas_reg_vector} includes an intercept, we can equivalently check orthogonality using $c_{it}\equiv u_{it}^2-1$ in place of $u_{it}^2$. Because $u_{it}$ is independent of $\mathcal F_{t-1}$, the omitted interaction
\[
b_{ji}\mathsf{g}_j'\xi_{t-1}u_{it}
\]
is orthogonal to $u_{it}$, $z_{t-1}u_{it}$, and $c_{it}$:
\[
\EE[(\mathsf{g}_j'\xi_{t-1})u_{it}^2]=\mathsf{g}_j'\EE[\xi_{t-1}]=0,
\]
\[
\EE[(\mathsf{g}_j'\xi_{t-1})z_{t-1}u_{it}^2]
=
\mathsf{g}_j'\EE[\xi_{t-1}z_{t-1}']=0,
\]
\[
\EE[(\mathsf{g}_j'\xi_{t-1})u_{it}(u_{it}^2-1)]
=
\mathsf{g}_j'\EE[\xi_{t-1}]\,\EE[u_{it}(u_{it}^2-1)]
=0.
\]
It is also orthogonal to every additive lag control $m_{t-1}\in L^2(\mathcal F_{t-1})$, because
\[
\EE[(\mathsf{g}_j'\xi_{t-1})u_{it}m_{t-1}]
=
\EE[(\mathsf{g}_j'\xi_{t-1})m_{t-1}]\,\EE[u_{it}]
=0.
\]
Likewise, the cross-shock term
\[
\sum_{\ell\neq i} b_{j\ell}(1+\mathsf{g}_j's_{t-1})u_{\ell t}
\]
is orthogonal to $u_{it}$, $z_{t-1}u_{it}$, $c_{it}$, and to every additive lag control in $L^2(\mathcal F_{t-1})$ by independence across shock coordinates and the zero mean of $u_{\ell t}$.

Therefore, after removing the included regressors $(u_{it},z_{t-1}u_{it},u_{it}^2)$ with their population coefficients and subtracting the orthogonal projection of the remainder onto $L^2(\mathcal F_{t-1})$, the residual is
\[
e_t^{\star,(j,i)}
=
b_{ji}(\mathsf{g}_j'\xi_{t-1})u_{it}
+
\sum_{\ell\neq i} b_{j\ell}(1+\mathsf{g}_j's_{t-1})u_{\ell t}.
\]
Therefore
\[
\psi_{0t}^{u,\star}
=
u_{it}e_t^{\star,(j,i)}
=
b_{ji}(\mathsf{g}_j'\xi_{t-1})u_{it}^2
+
\sum_{\ell\neq i} b_{j\ell}(1+\mathsf{g}_j's_{t-1})u_{it}u_{\ell t},
\]
so
\[
\EE[\psi_{0t}^{u,\star}\mid\mathcal F_{t-1}]
=
b_{ji}\mathsf{g}_j'\xi_{t-1}.
\]
If $b_{j\ell}=0$ for every $\ell\neq i$, the second sum vanishes and
\[
\psi_{0t}^{u,\star}=b_{ji}(\mathsf{g}_j'\xi_{t-1})u_{it}^2.
\]
Thus
\[
\EE[\psi_{0t}^{u,\star}\psi_{0,t-1}^{u,\star}]
=
b_{ji}^2\,
\EE\!\left[(\mathsf{g}_j'\xi_{t-1})(\mathsf{g}_j'\xi_{t-2})u_{it}^2u_{i,t-1}^2\right]
=
b_{ji}^2\,
\EE\!\left[(\mathsf{g}_j'\xi_{t-1})(\mathsf{g}_j'\xi_{t-2})u_{i,t-1}^2\right],
\]
because $u_{it}^2$ is independent of $\mathcal F_{t-1}$ with mean one. If the right-hand side is nonzero, then the lag-one shock-score autocovariance is nonzero and EHW fails already at $h=0$.
\end{proof}

\clearpage

\section{Details and Additional Results for Section \ref{sec:empiric.exp}}
{\noindent {\bf Construction of State Variables.}}

For each monthly log series $x_t\in\{\log(IP_t),\log(CPI_t)\}$, we construct the state using a real-time Hamilton filter, that is, a recursive version of the regression-based detrending approach of \citet{Hamilton2018}. Let $h_H=24$ and $p_H=12$. At each date $t$, we estimate the forecasting regression recursively using only observations available through $t$:
\[
x_{s+h_H}
=
a_t^x+\sum_{j=0}^{p_H-1} b_{jt}^x\, x_{s-j}+e_{s+h_H,t}^x,
\qquad s=p_H,\ldots,t-h_H.
\]
The associated real-time trend estimate at date $t$ is then
\[
\hat\tau_t^x
=
\hat a_t^x+\sum_{j=0}^{p_H-1}\hat b_{jt}^x\, x_{t-h_H-j},
\]
and the cyclical component is defined by
\[
c_t^x=x_t-\hat\tau_t^x.
\]
We apply this construction separately to log industrial production and log CPI, and the state entering the local projection at date $t$ is the lagged two-dimensional vector
\[
z_{t-1}=\big(c_{t-1}^{IP},\,c_{t-1}^{CPI}\big)'.
\]
Because both the regression coefficients and the regressors used to evaluate $\hat\tau_{t-1}^x$ are based only on observations dated $t-1$ and earlier, the resulting state is $\mathcal{F}_{t-1}$-measurable by construction.\\

\noindent {\bf Additional Results For $\textit{Feas}$.}
\begin{itemize}
	\item Figure~\ref{fig:app.exp.coef.quad} shows the coefficient estimates for the quadratic shock term in specification $\textit{Feas}$.
	\item Figure~\ref{fig:app.exp.coef.interact.cpi} shows the coefficient estimates for the interaction term between the shock and lagged CPI in specification $\textit{Feas}$.
        \item Figure~\ref{fig:app.exp.coef.interact.ip} shows the coefficient estimates for the interaction term between the shock and lagged IP in specification $\textit{Feas}$.
\end{itemize}
\clearpage

\begin{figure}[t!]
  \setlength{\abovecaptionskip}{0cm}
	\caption{Coefficients for Quadratic Shock Term}
	\label{fig:app.exp.coef.quad}  
	\begin{center} 
		\begin{tabular}{cc}
			\textit{IP} & \textit{Urate} \\
			\includegraphics[width=3in]{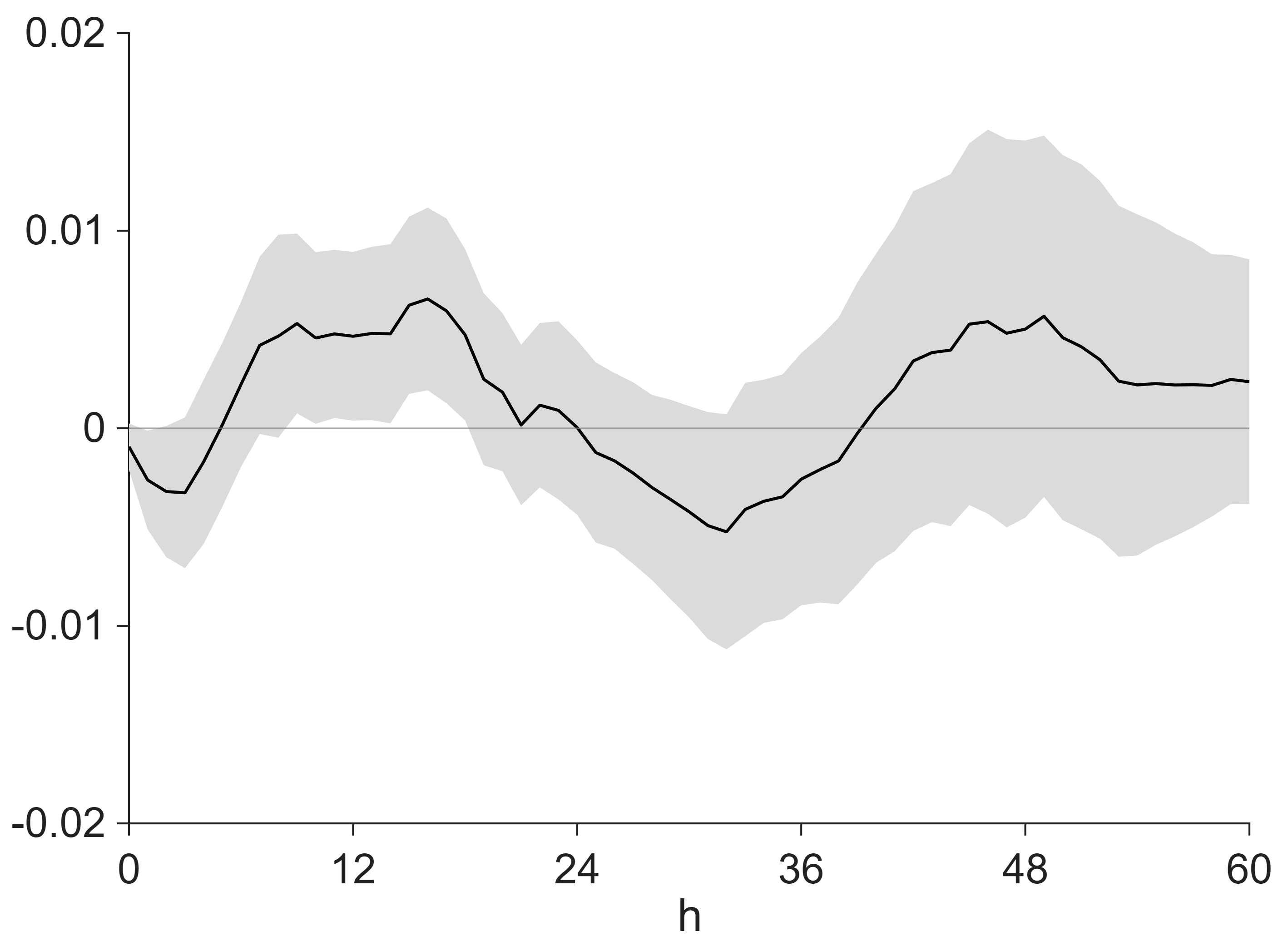} &
			\includegraphics[width=3in]{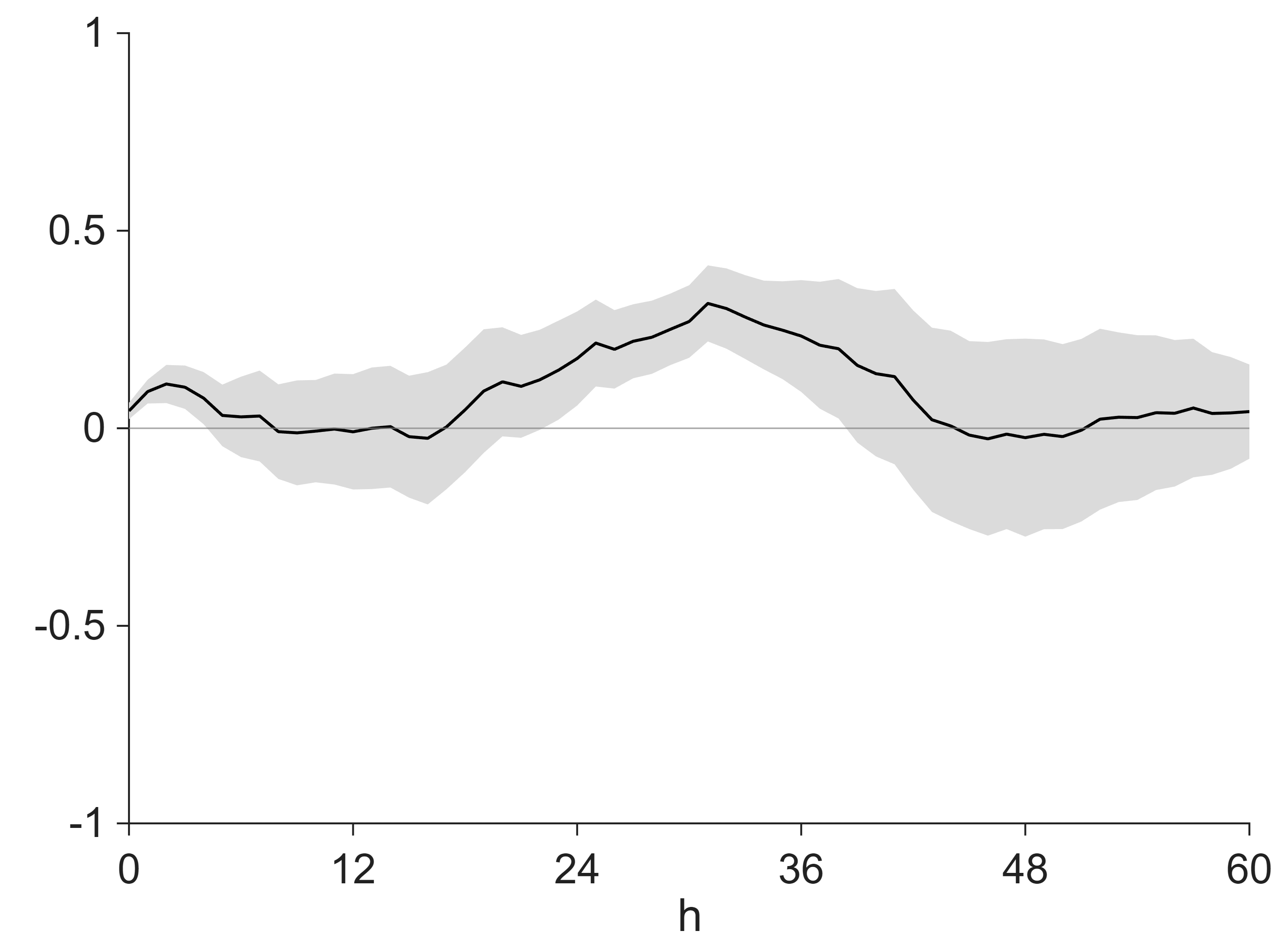} \\
                \textit{CPI} & \textit{FFR}\\ 
            \includegraphics[width=3in]{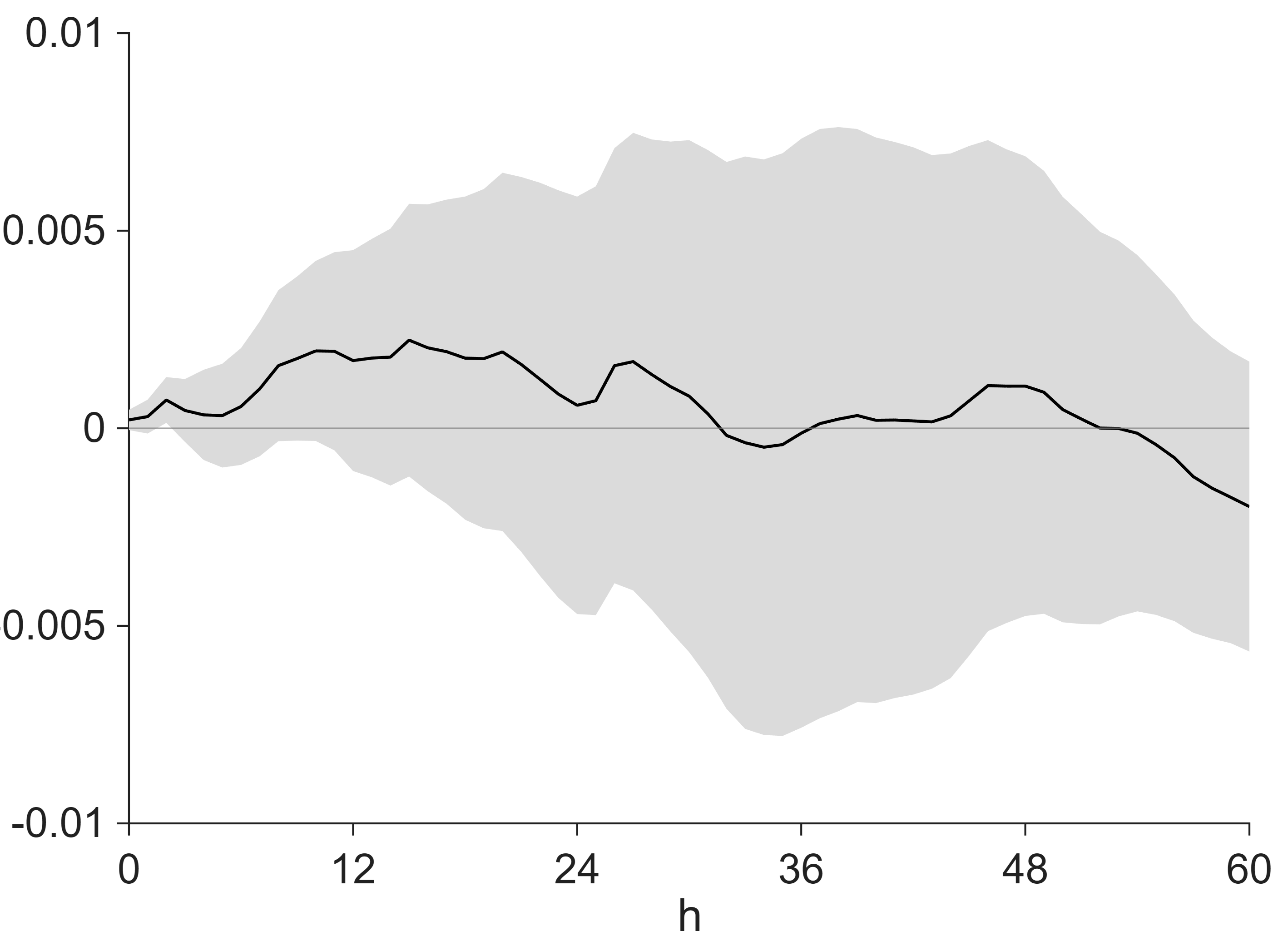} & \includegraphics[width=3in]{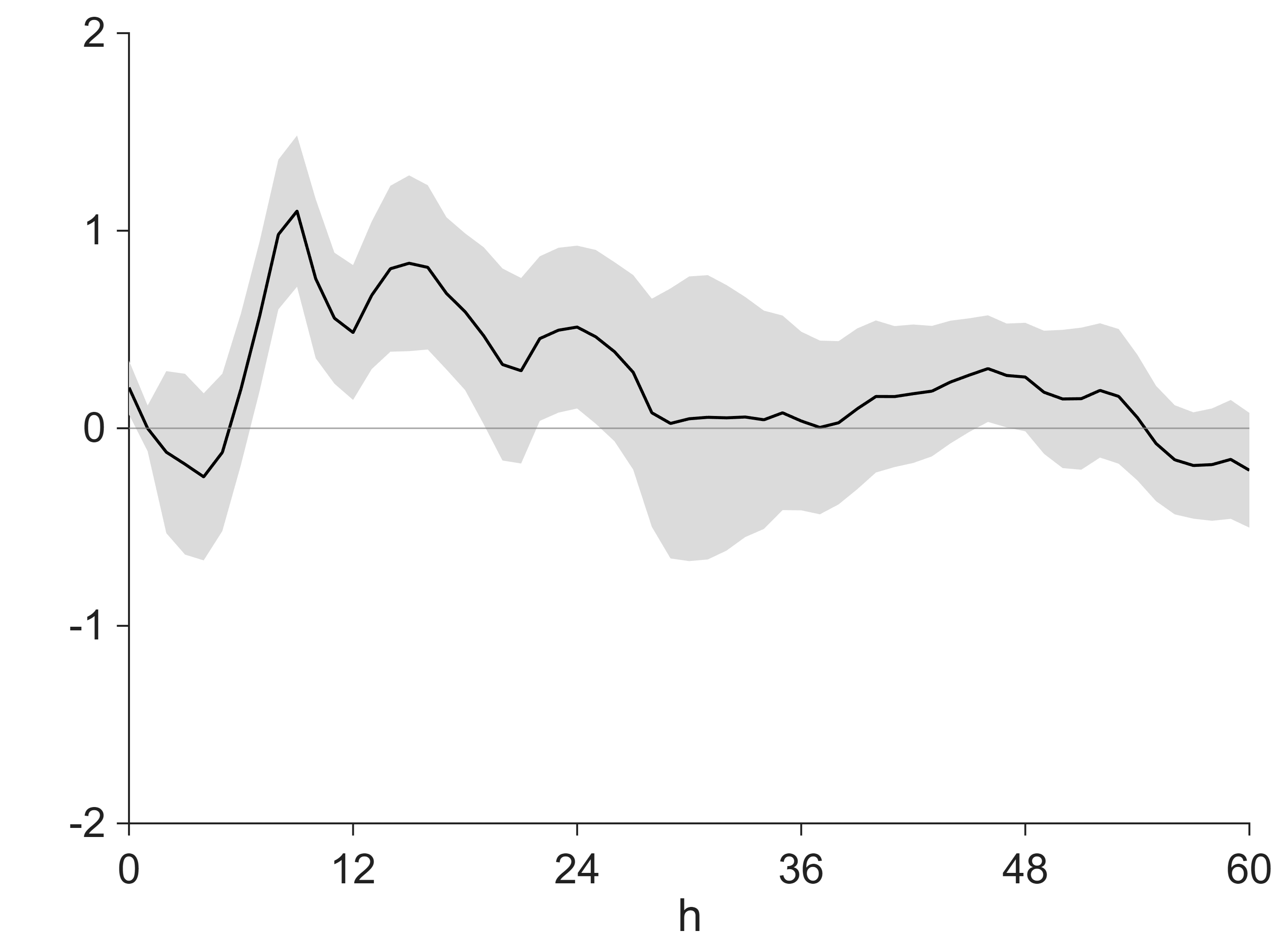}
            \\
		\end{tabular}
	\end{center}
	{\footnotesize {\em Notes}: Coefficient estimates for the quadratic shock term in specification $\textit{Feas}$. Shaded area represent the 90\% confidence band.}\setlength{\baselineskip}{4mm}
\end{figure}

\begin{figure}[t!]
  \setlength{\abovecaptionskip}{0cm}
	\caption{Coefficients for Interaction Terms: Lagged CPI}
	\label{fig:app.exp.coef.interact.cpi}  
	\begin{center} 
		\begin{tabular}{cc}
			\textit{IP} & \textit{Urate} \\
			\includegraphics[width=3in]{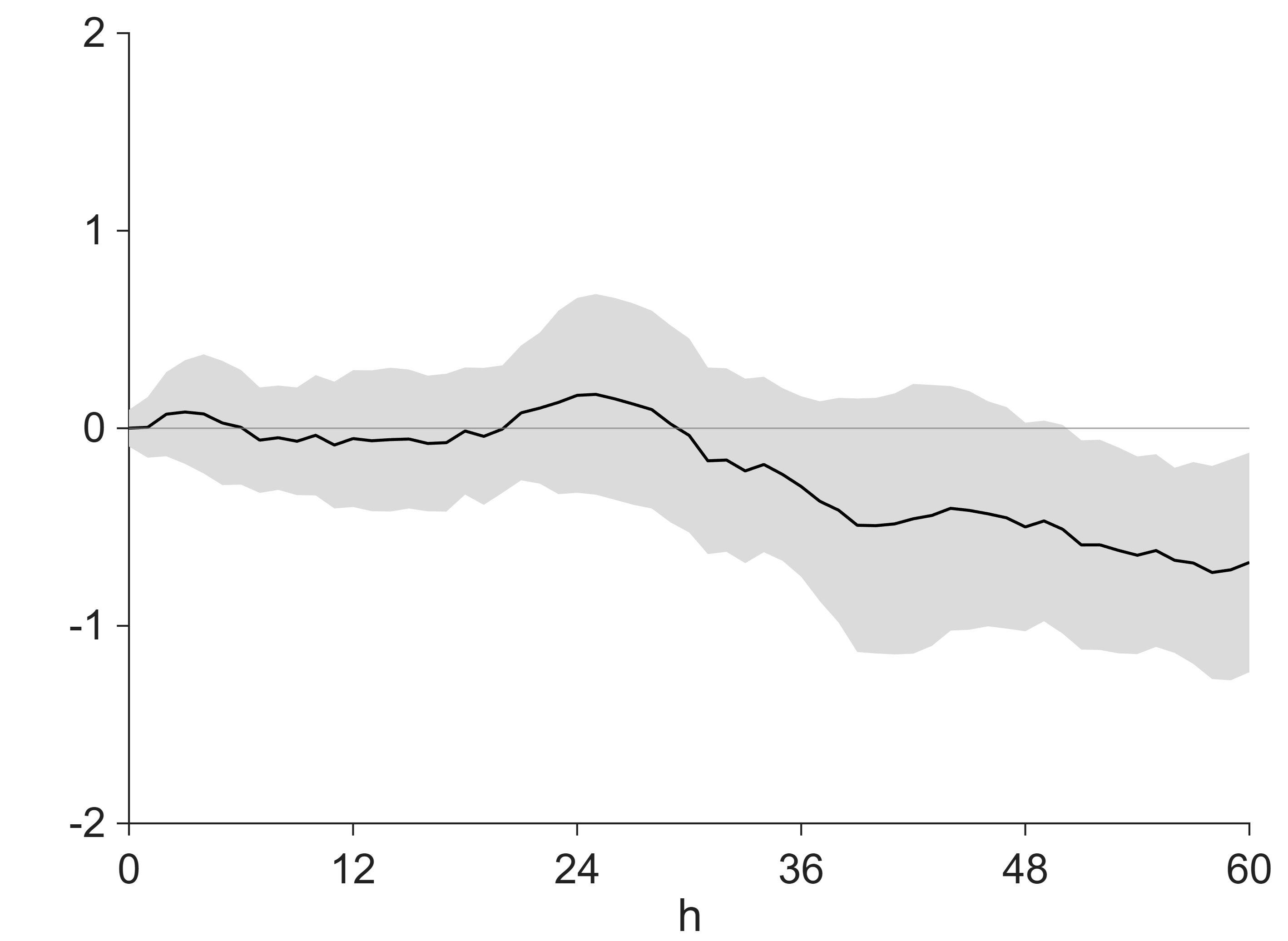} &
			\includegraphics[width=3in]{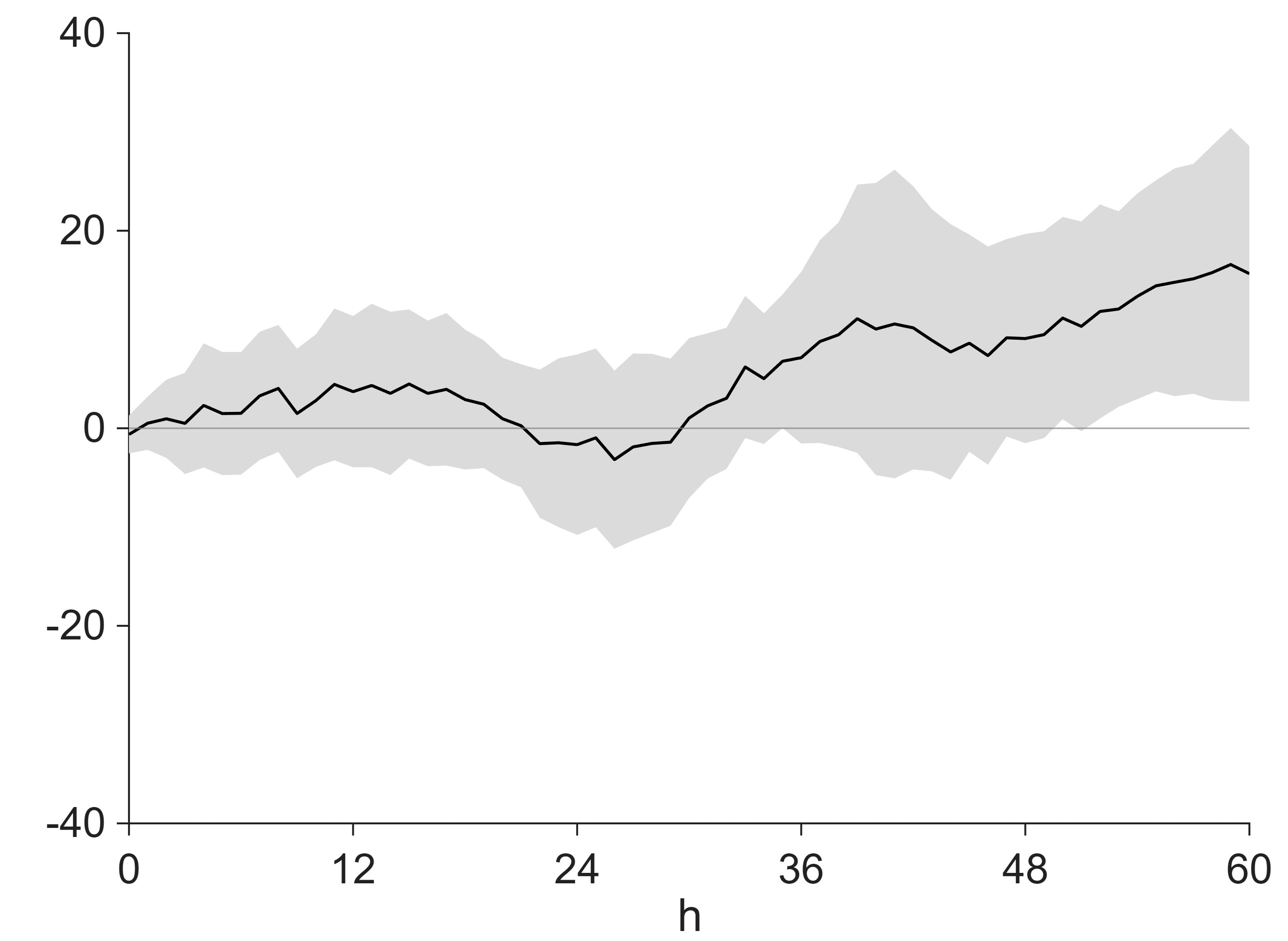} \\
                \textit{CPI} & \textit{FFR}\\ 
            \includegraphics[width=3in]{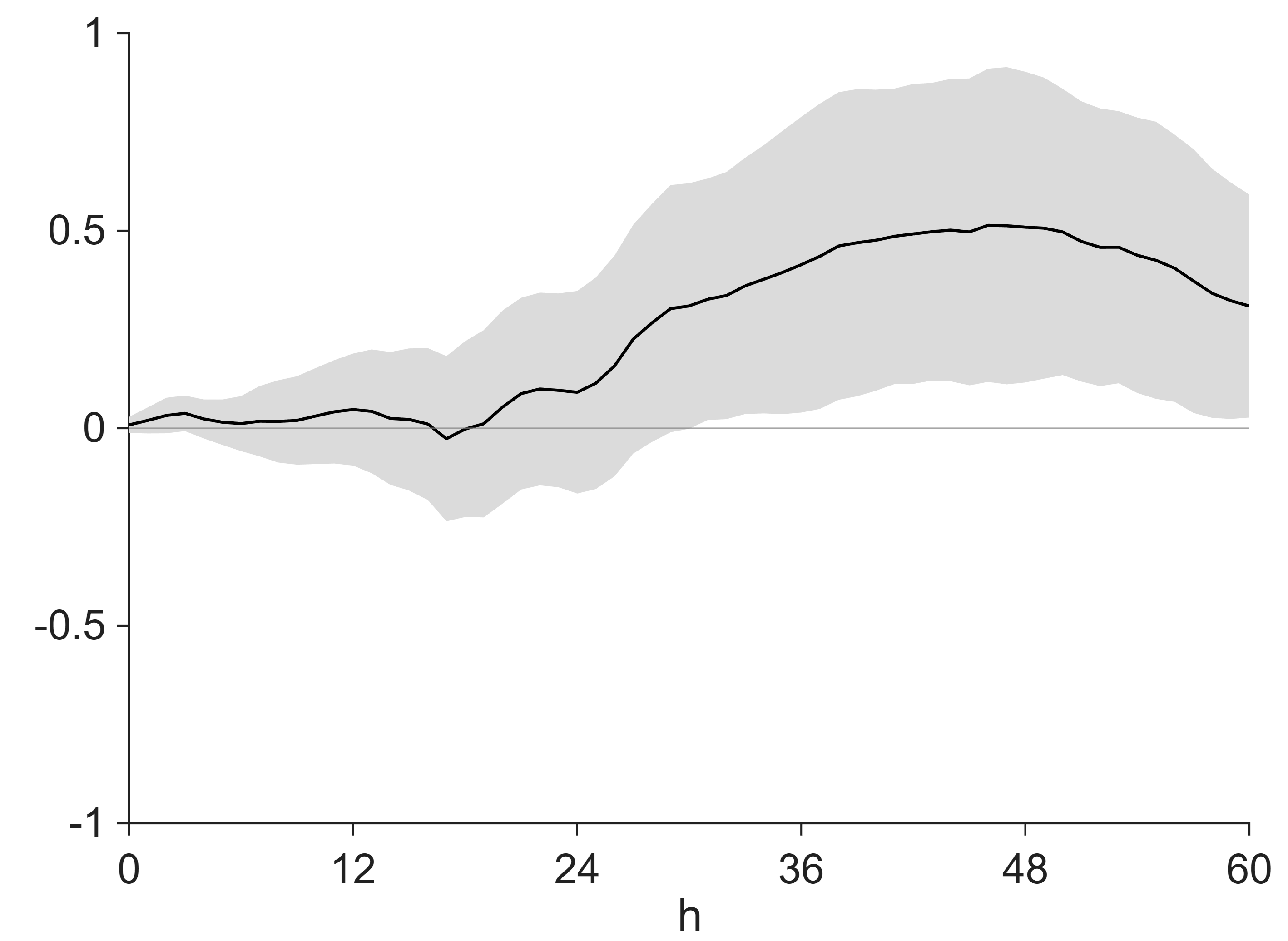} & \includegraphics[width=3in]{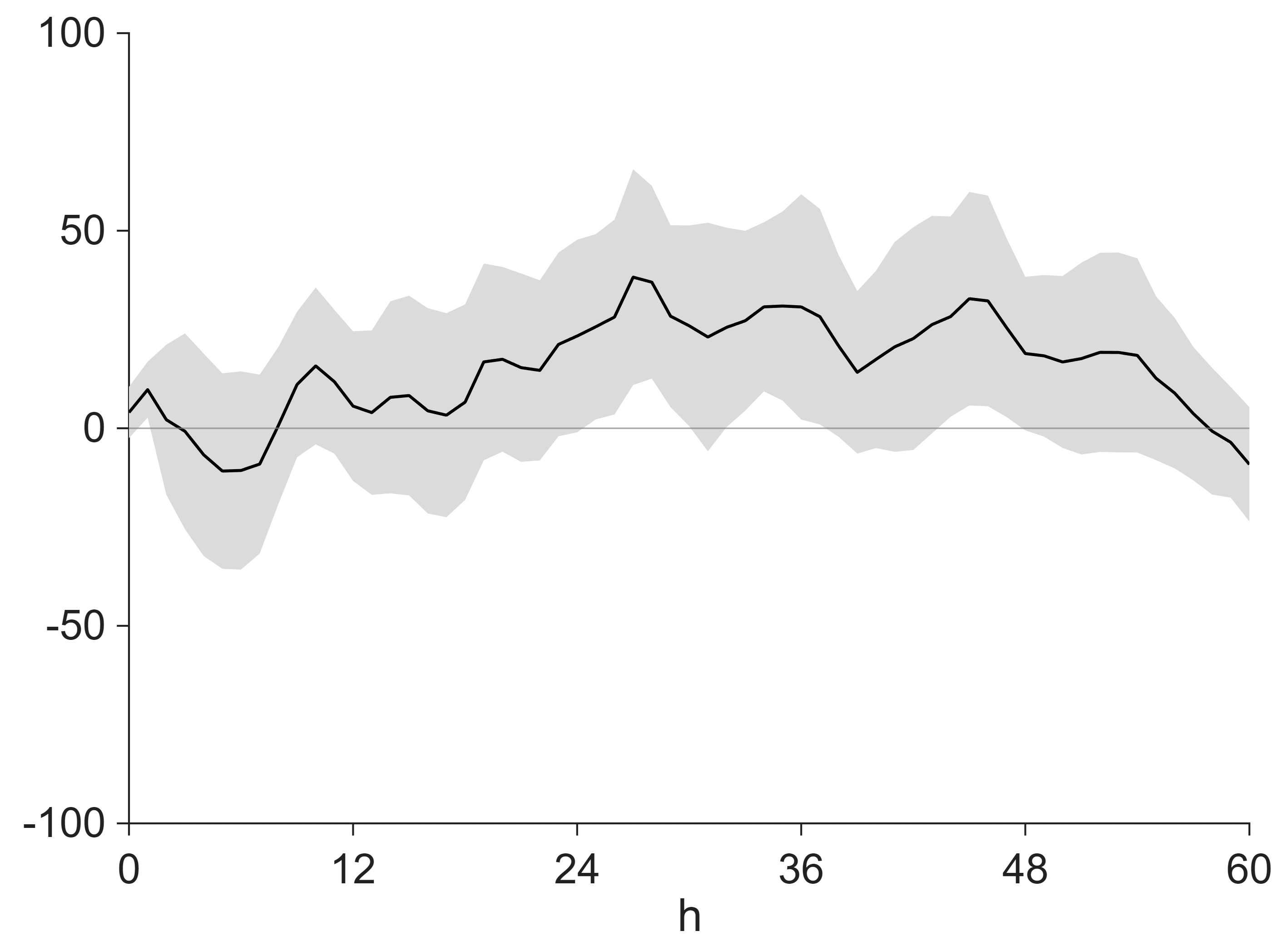}
            \\
		\end{tabular}
	\end{center}
	{\footnotesize {\em Notes}: Coefficient estimates for the interaction term between the shock and lagged CPI in specification $\textit{Feas}$. Shaded area represent the 90\% confidence band. }\setlength{\baselineskip}{4mm}
\end{figure}

\begin{figure}[t!]
  \setlength{\abovecaptionskip}{0cm}
	\caption{Coefficients for Interaction Terms: Lagged IP}
	\label{fig:app.exp.coef.interact.ip}  
	\begin{center} 
		\begin{tabular}{cc}
			\textit{IP} & \textit{Urate} \\
			\includegraphics[width=3in]{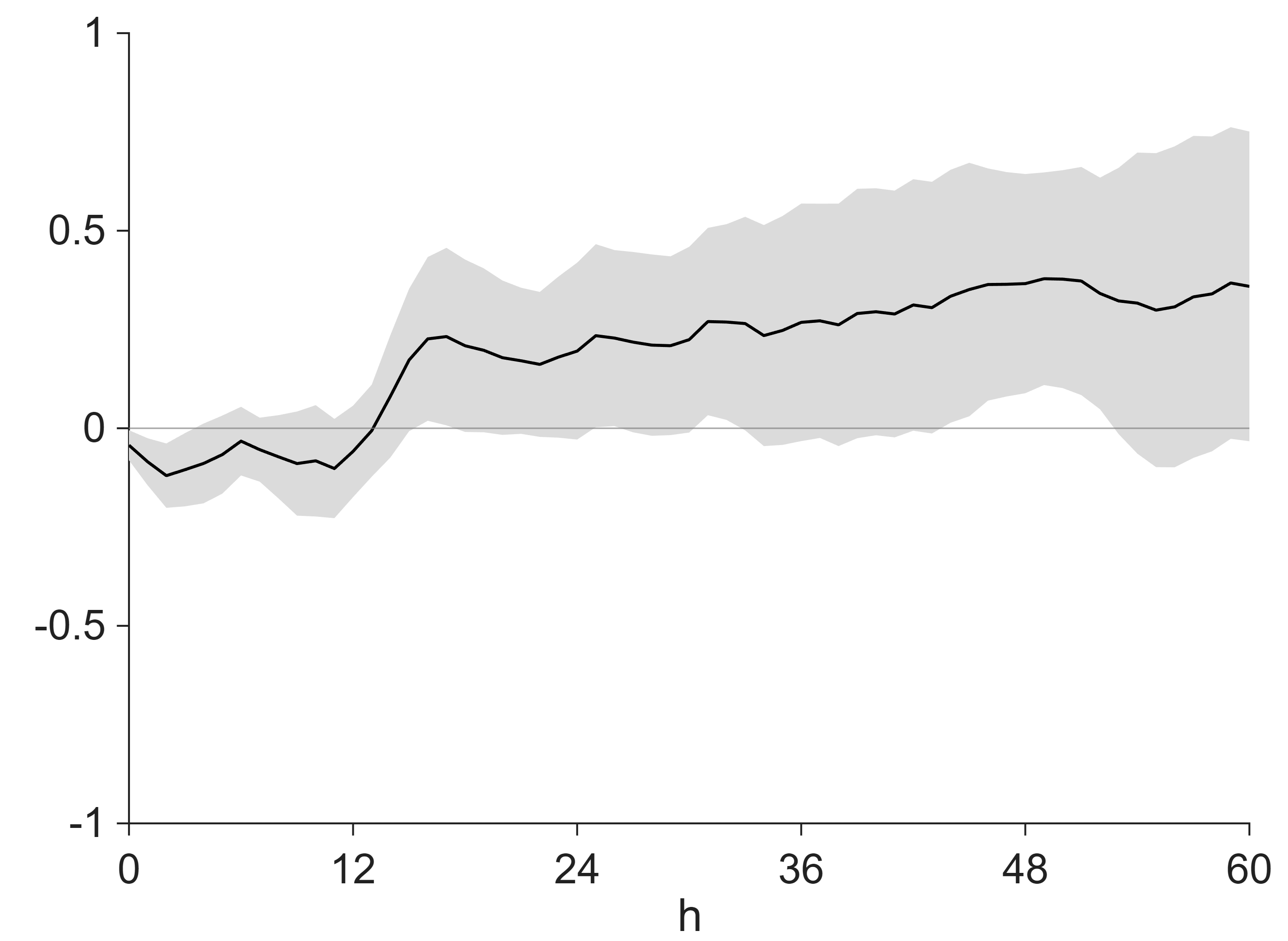} &
			\includegraphics[width=3in]{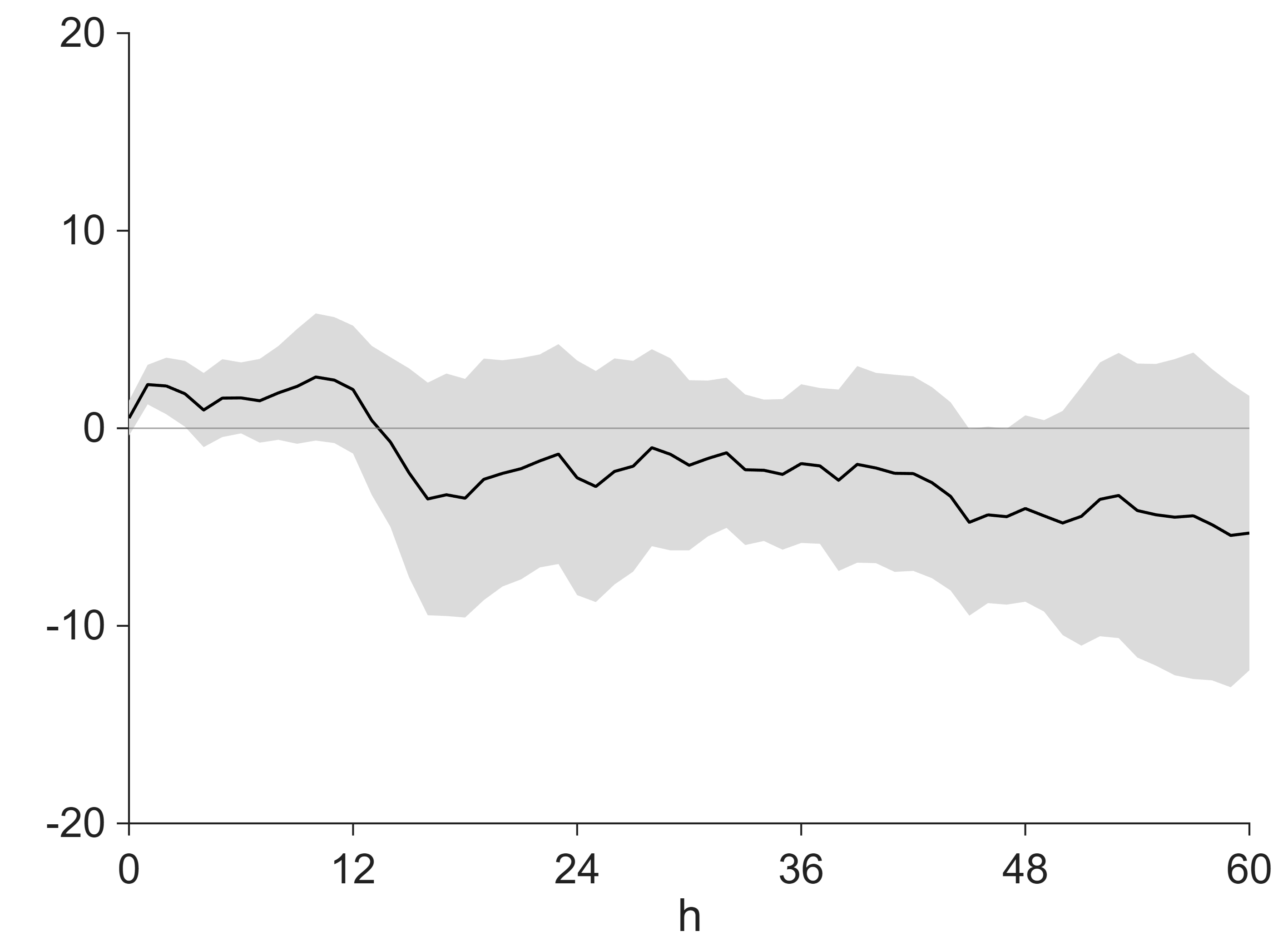} \\
                \textit{CPI} & \textit{FFR}\\ 
            \includegraphics[width=3in]{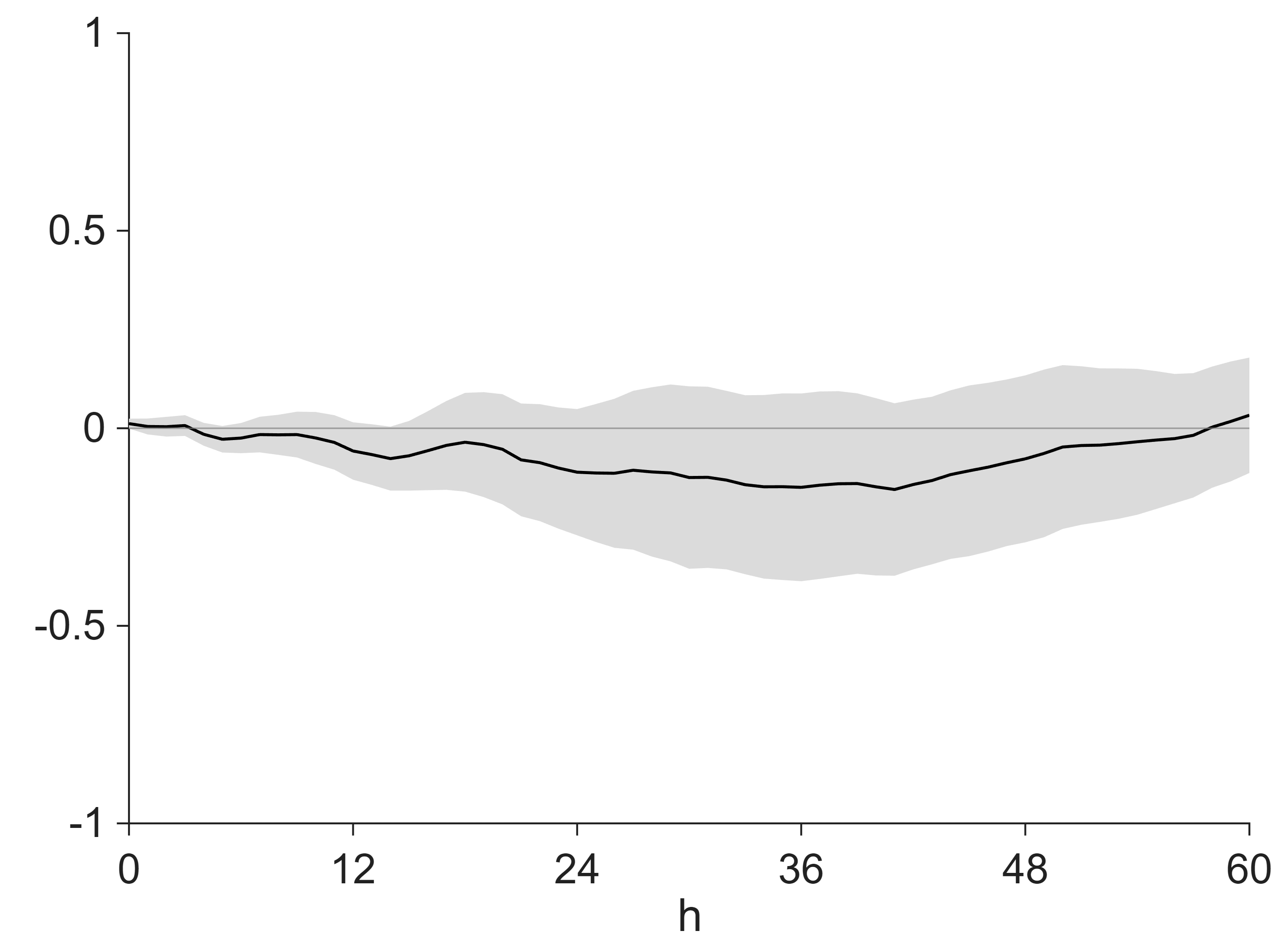} & \includegraphics[width=3in]{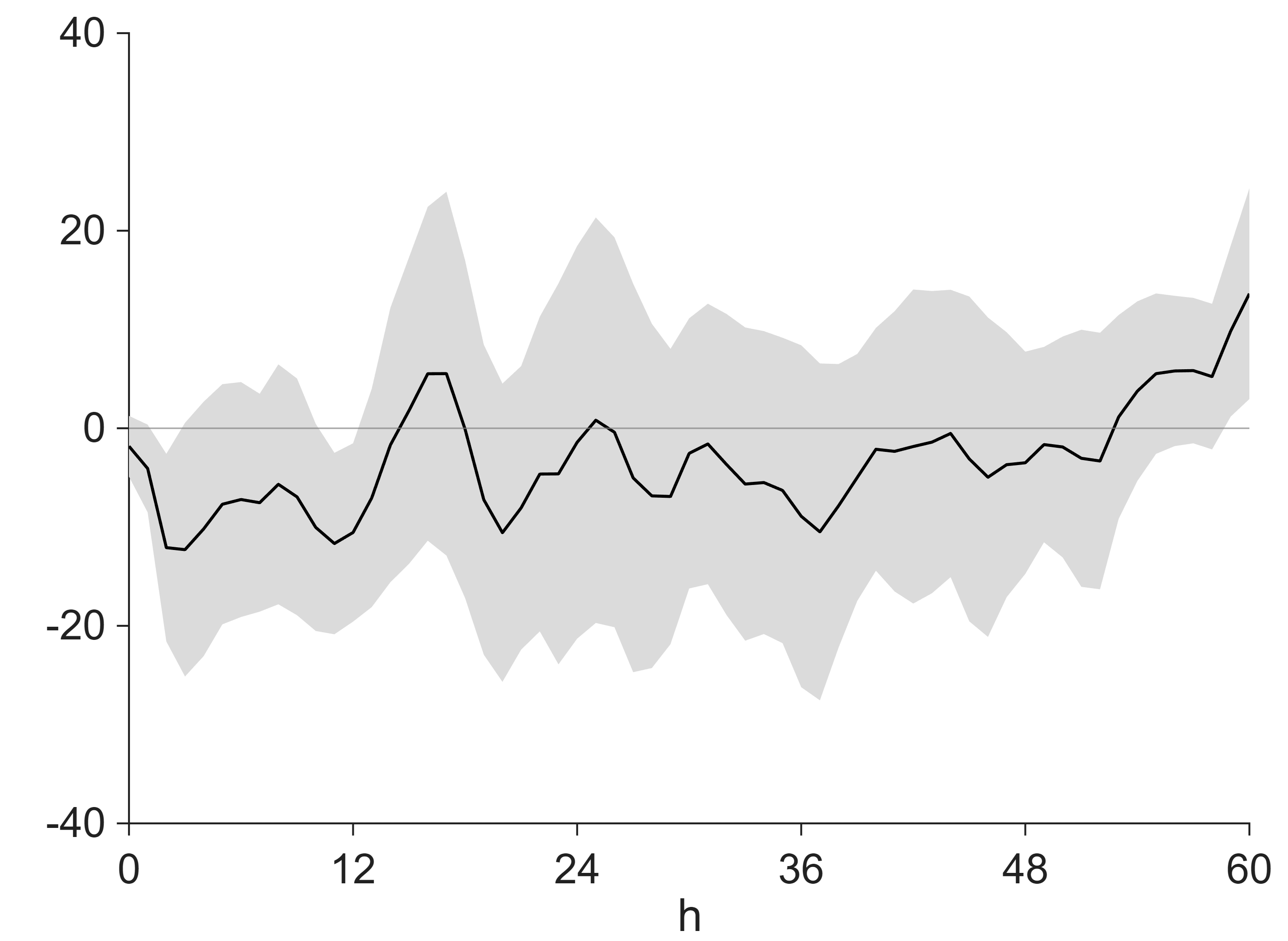}
            \\
		\end{tabular}
	\end{center}
	{\footnotesize {\em Notes}: Coefficient estimates for the interaction term between the shock and lagged IP in specification $\textit{Feas}$. Shaded area represent the 90\% confidence band. }\setlength{\baselineskip}{4mm}
\end{figure}

\clearpage
\noindent {\bf Estimation Details for  \textit{NPLP}.}

We present the estimation details for the control-adjusted nonparametric local projection estimator. We use the same sample, outcomes, shock series, state variables, and controls as in Section~\ref{sec:empiric.exp}.

For each horizon $h$ and each outcome, let
\[
U_t=\left(s^{IP}_{t-1},\ s^{CPI}_{t-1},\ u_t\right),
\]
where $\left(s^{IP}_{t-1},s^{CPI}_{t-1}\right)$ are the two lagged state variables used in \textit{Feas} and $u_t$ is the Romer--Romer shock. Let $W_t$ denote the same linear control vector used in \textit{Linear} and \textit{Feas}. We further impose the partially linear specification
\[
\EE[y_{t+h}\mid U_t,W_t]=m_h(U_t)+\theta_h'W_t,
\]
so the nonlinear component is a function of $(s^{IP}_{t-1},s^{CPI}_{t-1},u_t)$ only, while the controls enter linearly. This avoids conditioning nonparametrically on the full control vector.

The estimator follows Robinson-style partialling out. For a generic scalar variable $V_t$, equal either to $y_{t+h}$ or to one component of $W_t$, we estimate
\[
\mu_V(s^{IP},s^{CPI},u)=\EE[V_t\mid U_t=(s^{IP},s^{CPI},u)]
\]
by trivariate local-linear kernel regression. For the kernel, we work with prewhitened coordinates. Let $\bar U$ and $\widehat{\Sigma}_U$ denote the sample mean and covariance matrix of $U_t$, and define
\[
\widetilde U_t=(U_t-\bar U)\widehat{\Sigma}_U^{-1/2}.
\]
We use the same tilde notation for the transformed evaluation point $(s^{IP},s^{CPI},u)$. The local-linear estimate of $\mu_V(s^{IP},s^{CPI},u)$ is the fitted intercept from
\[
(\hat a,\hat b)
=
\argmin_{a,b}
\sum_{t=1}^{T_h}
\omega_t(s^{IP},s^{CPI},u)
\left(V_t-a-b'\left(\widetilde U_t-\left(\widetilde s^{IP},\widetilde s^{CPI}, \widetilde u\right)'\right)\right)^2,
\]
where
\[
\omega_t(s^{IP},s^{CPI},u)
=
\phi\!\left(\frac{\widetilde U_{t1}-\widetilde s^{IP}}{b_h}\right)
\phi\!\left(\frac{\widetilde U_{t2}-\widetilde s^{CPI}}{b_h}\right)
\phi\!\left(\frac{\widetilde U_{t3}-\widetilde u}{b_h}\right),
\quad
b_h=c\,T_h^{-1/7},
\]
and $\phi(\cdot)$ is the standard normal density. Rather than fixing the scalar constant $c$ a priori, we choose it separately for each outcome variable by blocked cross-validation. Let $\mathcal C=\{0.5,0.75,1,1.25,1.5,2,2.5,3,4\}$ denote the candidate set and let $\mathcal H_{CV}=\{0,6,12,\ldots,60\}$ denote the subset of horizons used for bandwidth selection. For each $h\in\mathcal H_{CV}$, we divide the available sample into five contiguous validation blocks. When one block is used for validation, we drop from the training sample not only that block but also an exclusion window of $\max\{6,h\}$ observations on each side, so as to reduce dependence induced by overlapping LP outcomes. We require at least 120 observations in the training sample and at least 20 observations in the validation block for a fold to be used.

Using the same smoother, we compute $\hat\mu_{y,h}(U_t)$ and $\hat\mu_W(U_t)$ at each observed $U_t$ and form
\[
\tilde y_{t+h}=y_{t+h}-\hat\mu_{y,h}(U_t),
\qquad
\tilde W_t=W_t-\hat\mu_W(U_t).
\]
We then estimate the linear control coefficients by OLS in the residualized regression
\[
\hat\theta_h=\argmin_{\theta}\sum_{t=1}^{T_h}\left(\tilde y_{t+h}-\theta'\tilde W_t\right)^2.
\]
The nonlinear component is recovered as
\[
\hat m_h(s^{IP},s^{CPI},u)
=
\hat\mu_{y,h}(s^{IP},s^{CPI},u)
-\hat\theta_h'\,\hat\mu_W(s^{IP},s^{CPI},u).
\]
For a target state $z=(s^{IP},s^{CPI})$ and shock size $\delta$, the estimated CAR is
\[
\widehat{CAR}^{NP}_h(z,\delta)
=
\frac{1}{T_h}\sum_{t=1}^{T_h}
\Big\{\hat m_h(z,u_t+\delta)-\hat m_h(z,u_t)\Big\}.
\]
For each candidate $c\in\mathcal C$, horizon $h\in\mathcal H_{CV}$, and fold $v$, we estimate the partially linear model on the corresponding training sample and compute the out-of-sample predictor
\[
\widehat y_{t+h}^{(-v,c)}
=
\hat\mu_{y,h}^{(-v,c)}(U_t)
+
\Big(W_t-\hat\mu_{W}^{(-v,c)}(U_t)\Big)'\hat\theta_h^{(-v,c)},
\qquad t\in\mathcal V_{h,v},
\]
where $\mathcal V_{h,v}$ is the validation set for horizon $h$ and fold $v$. We then aggregate the squared prediction errors across all usable horizons and folds and choose
\[
\widehat c
=
\arg\min_{c\in\mathcal C}
\frac{
\sum_{h\in\mathcal H_{CV}}
\sum_v
\sum_{t\in\mathcal V_{h,v}}
\Big(y_{t+h}-\widehat y_{t+h}^{(-v,c)}\Big)^2
}{
\sum_{h\in\mathcal H_{CV}}
\sum_v
|\mathcal V_{h,v}|
}.
\]
This procedure selects one bandwidth constant for each outcome variable, and that selected constant is then used for all horizons reported in the main text. In our sample, the selected constants are $\widehat c=2.5$ for industrial production and CPI, and $\widehat c=4.0$ for unemployment and the federal funds rate. Under this data-driven bandwidth choice, the resulting \textit{NPLP} responses in Figure~\ref{fig:irf.state.depend} display a state-dependence pattern qualitatively similar to the corresponding \textit{Feas} responses.

\end{appendix}

\end{document}